\documentclass[11pt]{article}
\usepackage[margin=1in]{geometry}
\usepackage{fullpage}

\usepackage{amsthm}
\newtheorem{theorem}{Theorem}[section]
\newtheorem{lemma}[theorem]{Lemma}
\newtheorem{corollary}[theorem]{Corollary}
\theoremstyle{definition}\newtheorem{definition}[theorem]{Definition}
\newtheorem{observation}[theorem]{Observation}

\newtheorem{claim}[theorem]{Claim}

\usepackage{amsfonts}
\usepackage{amssymb}
\usepackage{amsmath}
\usepackage{clrscode3e}
\usepackage[hidelinks]{hyperref}
\usepackage{lmodern}% http://ctan.org/pkg/lm

\newcommand{\whpshort}{\textit{w.h.p.}}
\newcommand{\whp}{with high probability}

\makeatletter
\newcommand*{\centerfloat}{%
	\parindent \z@
	\leftskip \z@ \@plus 1fil \@minus \textwidth
	\rightskip\leftskip
	\parfillskip \z@skip}
\makeatother

\title{Locally Consistent Parsing for Text Indexing in Small Space\footnote{
		This research is supported by ISF grant no.\ 1278/16, by a United States - Israel Binational Science Foundation (BSF) grant no.\ 2018364 and by an ERC grant MPM under the European Union's Horizon 2020 research and innovation programme (grant no.\ 683064).}}

\author{Or Birenzwige\\ Bar Ilan University\\\texttt{birenzo@cs.biu.ac.il}
	\and Shay Golan\\ Bar Ilan University\\\texttt{golansh1@cs.biu.ac.il}
	\and Ely Porat\\ Bar Ilan University\\\texttt{porately@cs.biu.ac.il}}

\date{}

\usepackage{makeidx}
\usepackage{tabularx, booktabs}
\usepackage{tabu}
\usepackage{colortbl}
\usepackage{enumitem}
\usepackage{multirow}
\usepackage{graphicx}
\usepackage{booktabs}

\renewcommand{\O}{\mathcal{O}}
\usepackage{listings}

\usepackage{marginnote}

\def\polylog{\operatorname{polylog}}

\newcommand{\ceil}[1]{\left\lceil{#1}\right\rceil}
\newcommand{\floor}[1]{\left\lfloor{#1}\right\rfloor}

\newcommand{\modulo}{\operatorname{mod}}
\newcommand{\naive}{na\"{\i}ve }

\newcommand{\E}{\operatornamewithlimits{\mathbb{E}}}
\newcommand{\ID}{\mbox{ID}}

\relpenalty=\maxdimen
\binoppenalty=\maxdimen

\newcommand{\LCE}{\textsf{LCE}}
\newcommand{\SST}{\textsf{SST}}
\newcommand{\lab}{\normalfont{\textsf{label}}}
\newcommand{\suc}{\textsf{succ}}
\newcommand{\rv}{\textsf{right-violation}}

\begin{document}
	
	\maketitle
	
	\thispagestyle{empty}
	\setcounter{page}{0}

	\begin{abstract}	
		We consider two closely related problems of text indexing in a sub-linear working space. The first problem is the Sparse Suffix Tree (SST) construction, where a text $S$ is given in read-only memory, along with a set of suffixes $B$, and the goal is to construct the compressed trie of all these suffixes ordered lexicographically, using only $\O(|B|)$ words of space. The second problem is the Longest Common Extension (LCE) problem, where again a text $S$ of length $n$ is given in read-only memory with some parameter $1\le \tau\le n$, and the goal is to construct a data structure that uses $\O(\frac {n}{\tau})$ words of space and can compute for any pair of suffixes their longest common prefix length. 
		We show how to use ideas based on the Locally Consistent Parsing technique, that were introduced by Sahinalp and Vishkin~\cite{SV94}, in some non-trivial ways in order to improve the known results for the above problems. 
		We introduce new Las-Vegas and deterministic algorithms for both problems. 
		
		For the randomized algorithms, we introduce the first Las-Vegas SST construction algorithm that takes  $\O(n)$ time.  This is an improvement over the last result of Gawrychowski and Kociumaka~\cite{GK17} who obtained $\O(n)$ time for Monte Carlo algorithm, and $\O(n\sqrt{\log |B|})$ time with high probability for Las-Vegas algorithm. In addition, we introduce a randomized Las-Vegas construction for a data structure that uses $\O(\frac {n}{\tau})$ words of space, can be constructed in linear time with high probability and answers LCE queries in $\O(\tau)$ time.
		
		For the deterministic algorithms, we introduce an SST construction algorithm that takes $\O(n\log \frac {n}{|B|})$ time (for $|B|=\Omega(\log n)$).
		This is the first almost linear time, $\O(n \cdot \polylog{n})$, deterministic   SST construction algorithm, where all previous algorithms take at least $\Omega\left(\min\{n|B|,\frac {n^2}{|B|}\} \right)$ time.
		For the LCE problem, we introduce a data structure that uses $\O(\frac {n}{\tau})$ words of space and answers LCE queries in $\O(\tau\sqrt {\log^*n})$ time, with $\O(n\log \tau)$ construction time (for $\tau=\O(\frac {n}{\log n})$). This data structure improves both query time and construction time upon the results of Tanimura et al.~\cite{TBIPT16}. 
		
	\end{abstract}

	\newpage

	\section{Introduction}\label{sec:intro}
	
	Text indexing is one of the most fundamental problems in the area of string algorithms and information retrieval (e.g., see \cite{Weiner73,KU96,Ukkonen95,McCreight76,FG99,BFGKSV16,FH08,FIK16,BFO13,FMN09,AN08,KLAAP01,KS03,BK03,PT08}). 
	In this paper we consider two closely related problems of text indexing in sub-linear working space. 
	In both problems, we are given a string $S$ of length $n$ to process, in read-only memory.
	
	The first problem is the Sparse Suffix Tree ($\SST$) construction which is a generalization of the classical Suffix Tree data structure. The Suffix Tree index of $S$ is a compact trie of all $S$'s suffixes. 
	Since the '70s, several algorithms were introduced for the suffix tree construction which take $\O(n)$ time~\cite{Weiner73,McCreight76,Ukkonen95}. 
	In the Sparse Suffix Tree construction problem only a subset of the suffixes, $B\subseteq\{1,2,\dots,n\}$, need to be indexed.
	This generalization was introduced in 1996 by  K{\"{a}}rkk{\"{a}}inen and Ukkonen~\cite{KU96} who introduced a construction algorithm for the special case where all the suffixes are in the positions which are integer multiples of some $k\in\mathbb N$. 
	The algorithm of~\cite{KU96}  takes $\O(n)$ time and uses linear working space in the number of suffixes.
	For the general case of an arbitrary set of suffixes, one can construct an $\SST$ by pruning of the full Suffix Tree, which requires $\Theta(n)$ words of working space. However, when the space usage is limited to be linear in $|B|$, another approach is required. Recently,  Gawrychowski and Kociumaka~\cite{GK17} introduced an optimal randomized Monte Carlo algorithm which builds an $\SST$ in linear time, using only $\O(|B|)$ words of space.

	The second problem we consider in this paper is the Longest Common Extension ($\LCE$) problem. In this problem we are given a text $S$ of length $n$ in read-only memory, and the goal is to construct a data structure that can compute for any pair of suffixes their longest common prefix length, or formally $\LCE(i,j)=\min\{k\ge 0 \mid S[i+k]\ne S[j+k]\}$. 
	It is well known that one can solve this problem using $\O(n)$ words of space with $\O(n)$ construction time and $\O(1)$ query time~\cite{GG88,LV88}.
	We consider the trade-off version, where there is a parameter $1\le\tau\le n$ and the goal is to construct a data structure using only $\O(\frac n\tau)$ words of working space while achieving fast construction and query time.
	The trade-off $\LCE$ problem is closely tied to the $\SST$ problem, and in several papers the $\LCE$ was used as the main tool for constructing the $\SST$ efficiently \cite{IKK14,BFGKSV16,FIK16,GK17}. 
	
	For the $\LCE$ problem, Kosolobov~\cite{Kosolobov17} proved a lower bound in the cell-probe model, that any data structure for the $\LCE$ problem that uses $S(n)= \Omega(n)$ 
	\emph{bits} of working space and $T(n)$ query time, must have $S(n)\cdot T(n)=\Omega(n\log n)$ for large enough alphabet. Kosolobov also conjectured that for data structure which used $S(n)< o(n)$ bits of space a similar lower bound of $T(n)\cdot S(n)=\Omega(n\log S(n))$ holds. 
	A trade-off of $\O(\tau)$ query time and $\O(\frac n\tau)$ words of space is already known, both in randomized and deterministic construction algorithms \cite{BGKLV15}, and therefore the research efforts are invested in reducing the construction time. Tanimura et al.~\cite{TBIPT16} introduced a deterministic data structure with $\O(n\tau)$ construction time algorithm, and with a slightly increasing of the query time to $\O(\tau\min\{\log\tau,\log\frac{n}{\tau}\})$.
	Gawrychowski and Kociumaka~\cite{GK17} have shown a linear time Monte Carlo construction of an $\LCE$ data structure with $\O(\tau)$ query time. 
	In a close model of \emph{encoding} data structures, Tanimura et al.~\cite{TNBIT17} introduce a data structure that encodes the string, and then can delete the original string. The complexity of the data structure is dependent on the Lempel--Ziv 77 factorization of the string, and in some cases it can beat the lower bound of Kosolobov~\cite{Kosolobov17}.

	\subsection{Our Contribution}
	
	\paragraph{$\LCE$ data structures.}
	We provide new $\LCE$ data structures that answer queries deterministically and use only $\O(\frac n\tau)$ words of space. The construction algorithms are either deterministic or randomized Las-Vegas algorithms. Our results are stated in the following theorems, see a full comparison in Appendix~\ref{app:table_compare}.

	\begin{theorem}[Las-Vegas Algorithm]\label{thm:lasvegasLCE}
		For any $1\le \tau\le \O(\frac{n}{\log^2 n})$ there exists a data structure of size $\O(\frac n\tau)$ words that can be constructed in $\O(n)$ time with high probability\footnote{With probability $1-n^{-\Theta(1)}$ at least.} and answers $\LCE$ queries in $\O(\tau)$ time.  
		For $\tau=\omega(\frac n{\log^2 n})$ the same data structure can be constructed in \emph{expected} $\O(n)$ time.
	\end{theorem}
	
	For the deterministic trade-off, we begin with a data structure with the following complexities.
	
	\begin{lemma}[Deterministic Algorithm]\label{lem:deterministicLCE}
		For any $1\le \tau\le \O(\frac n{\log n})$ there exists a data structure of size $\O(\frac n\tau)$ words which can be constructed deterministically  in $\O(n \log \tau )$ time and answers $\LCE$ queries in $\O(\tau\log^*n)$ time. 
	\end{lemma}

	Then, we improve the results as stated in the following theorem. 
	
	\begin{theorem}[Deterministic Algorithm]\label{thm:deterministicLCE2}
		For any $1\le \tau\le \O(\frac n{\log n})$ there exists a data structure of size $\O(\frac n\tau)$ words which can be constructed deterministically  in $\O(n\log \tau )$ time and answers $\LCE$ queries in $\O(\tau\sqrt{\log^*n})$ time. 
	\end{theorem}

	\paragraph{Sparse Suffix Tree construction.}
	For the $\SST$ construction problem we introduce an optimal randomized Las-Vegas algorithm, which improves upon the previous algorithm of Gawrychowski and Kociumaka~\cite{GK17} who developed an optimal Monte Carlo algorithm and a Las-Vegas algorithm with additional $\sqrt{\log|B|}$ factor. 
	In addition, we introduce  the first almost-linear deterministic $\SST$ construction algorithm. Our results are stated in the following theorems.
	\begin{theorem}\label{thm:lasvegasSST}
		
		Given a set $B\subseteq[n]$   of size  $\Omega (\log^2n)\le |B|\le n$, there exists a randomized Las-Vegas algorithm which constructs the Sparse Suffix Tree of $B$ in  $\O(n)$ time \whp{} using only $\O(|B|)$ words of space. 
		For $|B|<o(\log^2 n)$ there exists a randomized Las-Vegas algorithm with the same complexities except that the running time is just \emph{expected} to be $\O(n)$.

	\end{theorem}

	\begin{theorem}\label{thm:deterministicSST}
		Given a set $B\subseteq[n]$   of size  $\Omega (\log n)\le |B|\le n$, there exists a deterministic algorithm which constructs the Sparse Suffix Tree of $B$ in  $\O(n\log \frac n{|B|})$ time using only $\O(|B|)$ words of space.

	\end{theorem}

	\subsection{Algorithmic Overview}\label{sec:algorithmic_overview}
	
	We first give an overview of the $\LCE$ data structures and then describe how to use them (with additional techniques) to build an $\SST$ using only a small amount of space. 
	All the techniques mentioned in the overview are described later in detail.
	For both problems, we use a parameter $1\le\tau\le n$, such that the working space of the data structures and the construction algorithms are limited to $\O(\frac n\tau)$ words of space.
	
	\paragraph{$\LCE$ data structure.} 
	The \naive method to compute $\LCE(i,j)$ is to compare pairs of characters until a mismatch is found. Since this \naive methos takes $\O(1+\LCE(i,j))$ time, the case where $\LCE(i,j)=\O(\tau)$ is  solvable in $\O(\tau)$ time without any preprocessing. 
	Therefore, we invest our efforts on speeding up the computation of queries where $\LCE(i,j)=\ell>\omega(\tau)$.
	In this case, we use a special partitioning of $S$  into blocks.
	The partitioning has the property that if a string appears as a substring of $S$ multiple times, then all these substrings are composed of exactly the same blocks, except for possibly small margins of the substrings. 
	More detailed, if $S[i..i+\ell-1]=S[j..j+\ell-1]$ are long enough substrings, then both $S[i..i+\ell-1]$ and $S[j..j+\ell-1]$ can be considered as composed of three parts, a prefix, a suffix and the middle part  between them with the following properties. Both the prefix and the suffix are short strings (e.g., of length $\O(\tau)$ for the randomized algorithm), or they have some other special properties, and the middle part is composed of full blocks, which are exactly the same in  both substrings, regardless the length of this part.
	Thus, using this partitioning, the computation of $\LCE(i,j)$ is done in three phases: first, the algorithm verifies that $\LCE(i,j)$ is long enough and if so --- the algorithm finds the beginning of the common middle part. 
	Then, the algorithm computes the length of the middle part (which might be very long) by a designated data structure that performs this computation in constant time. At the end, the algorithm just computes the length of the short suffix, right after the end of the last block of the middle part.

	The key concepts of this technique are similar  to the Locally Consistent Parsing of Sahinalp and Vishikin~\cite{SV96,CM07} and also to the techniques of Kociumaka et al.~\cite{KRRW15}.  
	However, we design novel methods to implement these ideas in a small space. For the randomized algorithms, we combine Karp--Rabin fingerprints with approximately min-wise hash functions, and for the  deterministic algorithms we show how to compute the parsing in an online fashion in order to reduce the space usage.
	In previous results for the $\LCE$ problem~\cite{BFGKSV16,BGKLV15,GK17,TBIPT16}, a set of positions of size $\O(\frac n\tau)$ was also considered by the data structures. In contrast to these algorithms, where the selected positions were dependent only on the \emph{length} of the text, we introduce a novel method that exploits the \emph{actual text} using local properties in order to decide which positions to select.

	\paragraph{Partitioning set.}
	In order to achieve the desired special partitioning of the string into blocks, our algorithms pick positions in $S$, which will be the endpoints
	of the blocks. 
	Our goal is to pick $\O(\frac n\tau)$ positions, such that the distance between any two consecutive positions is $\O(\tau)$ and, for the locality, the decision whether a position $p\in\{1,2,\dots,n\}=[n]$ is picked, depends only on the \emph{neighbourhood} of $p$ up to $\O(\tau)$ 
	characters away, which is the substring $S[p-\delta..p+\delta]$  for some $\delta\in \O(\tau)$. 
	
	The above discussion captures the key concepts of the algorithms, however, in certain cases it is impossible to pick such a set of positions. For example, if the string is $S=aaa\dots a$, it is obvious that any local selection must pick all the positions (except possibly for the margins) or none of them. Thus, it is impossible to fulfill both the requirement of $\O(\frac n\tau)$ positions and that the distance between any two consecutive positions will be $\O(\tau)$. Generally, this problem occurs when the string $S$ contain long substrings whose \emph{period} (see Section~\ref{sec:preliminaries}) length is $\O(\tau)$. 
	Thus, our algorithms use period-based techniques to treat positions in such substrings. So, the partitioning of $S$ will contain two types of blocks: periodic blocks, which are substrings with period length at most $\tau$, and regular blocks, that contains no repetition.
	For the sake of clarity, all the details about periodic blocks are omitted in this overview.
	
	It remains to describe how the algorithm picks the special positions, and how the algorithm uses these positions to answer $\LCE$ queries on these positions fast.
	We introduce two methods, and therefore two algorithms, for the position selection. The former is a randomized Las-Vegas method and the later is a deterministic method, with a slightly worse construction and query time.
	
	We note that after discover the idea of partitioning set, we found that similar ideas were invented in the previous decade in the area of Bioinformatics, without theoretical analysis, by Schleimer et al.~\cite{SWA03} called \emph{winnowing} and by Roberts et al.~\cite{RHHMY04} called \emph{minimizers}. 
	In addition, Kempa and Kociumaka~\cite{KK19,Kociumaka18}  was influenced by our definition of partitioning set, while defining \emph{synchronizing set}.  s

	\paragraph{The randomized method.}

	For the randomized data structure, the positions are picked in the following way. 
	The algorithm uses a novel hashing technique, that maps each string of length $\tau$ into an integer value in the range $[1..n^c]$ for some constant $c$. The mapping function is \emph{almost min-wise} function, which means that given $k$ different strings the probability of each one of them to be mapped into the minimum value  among all the strings is at most $\O(\frac 1 k)$. For any position $p\in [n]$ the algorithm computes the fingerprint of $S[p..p+\tau-1]$ and considered it as the ID of position $p$. 
	Then, the algorithm scans the positions with a sliding window of length $\tau$, and for any consecutive $\tau$ positions, the algorithm picks the position with the minimum ID. Notice that it is possible, and even advisable, that a single position will be picked due to multiple overlapping shifts of the sliding window. 
	It is straightforward that in any sequence of $\tau$ positions this method will pick at least one position and that the selection of each position depends only on a substring of  $S$ of length $3\tau$. The only possible problem with fulfilling the desired properties is the possibility of choosing too many positions ($\omega(\frac n\tau)$).
	With a simple probabilistic analysis of this method we prove that under random selection of the hash function, the \emph{expected} number of picked positions is $\O(\frac n\tau)$, assuming that in each offset of the sliding window, each substring of length $\tau$ appears only once, i.e., the non-periodic case. 
	The details of this method, including the treatment of periodic substring, appears in Section~\ref{sec:randomSelection}.
	Later, in Section~\ref{sec:ens_rand_whp} we show how to improve this method to $\O(\frac n\tau)$ picked positions, with linear construction time \emph{with high probability}. 
	The main idea of this improvement is to evaluate whether the selected min-wise hash function provides a small enough set of positions by sampling a few substrings of $S$, and if the set of picked positions in these substrings is too large, the algorithm chooses another function.

	\paragraph{The deterministic method.}
	In the deterministic method, we mimic the randomized selection, but with  a technique of deterministic coin-tossing, going back to Cole and Vishkin~\cite{CV86} and appeared in several other papers~\cite{CM07,GPS87,MSU97,SV96,ABR00,I17,FIK16}. 
	In the high level, 
	this technique transforms a string over an alphabet $\Sigma$ into a same length string over a constant alphabet in $\O(n\log^* n)$ time. The transformed string has the special property that for any pair of consecutive positions, if the original characters before the transformations were distinct, then the characters on these positions after the transformation are distinct as well. In addition, the transformation of each position depends only on the \emph{neighbourhood} of the position of $\O(\log^* n)$ positions, assuming that $S$ does not contain any kind of periodic substring. 
	Using this transformation, the algorithm splits $S$ into blocks of size at least $2$ and at most $\O(1)$, by selecting all the positions which their transformed value is a local minimum.
	By a repeatedly applying this transformation, while after the first transformation we consider each block as a character over a larger alphabet, the algorithm creates $\O(\frac n\tau)$ blocks of size $\O(\tau)$ characters. 
	Each of these blocks depends only on a substring of $S$ of length $\O(\tau \log^* n)$, which adds only a factor of $\log^*n$ over the complexities of algorithms that use the randomized selection. 
	The total computation time is $\O(n\log\tau)$ and with a novel scheduling technique, which scans the text in one pass from right to left, the computation is implemented using only $\O(\frac n\tau + \log \tau)$ words of space.  
	During the transformations of the string into blocks, the algorithm recognizes all the long substrings of $S$ with a period length of at most $\tau$ characters for the special treatment. 
	
	The technique is similar to previous algorithms as~\cite{CM07}, however,  this is the first implementation where the recursive construction does not use any randomness and also does not require $\Omega(n)$ working space. 
	Another previous algorithm that uses similar techniques is by Fischer et al.~\cite{FIK16} which also consider the $\LCE$ and $\SST$ problems using the Locally Consistent Parsing. However, the model of computation considered in~\cite{FIK16} is weaker, where in addition to the $\O(\frac n\tau)$ working space, the algorithm may also edit the space of the input, as long as the algorithm recovers the original input at the end of the construction.
	While in~\cite{FIK16} the usage of the Locally Consistent Parsing technique is involved, and a lot of effort is invested in merging parsing of different parts of the text, in our algorithm we introduce a space- and time-efficient implementation of the parsing for the whole text.

	\paragraph{Compute $\LCE$ query using the selected positions.}
	Given the partitioning of $S$ induced by the selected positions, we define  for any two indices $i',j'$ which are beginning of some blocks, the longest common block-prefix as the largest $\ell'$ such that $S[i'..i'+\ell'-1]=S[j'..j'+\ell'-1]$ and all the blocks that compose the two substrings are exactly the same.
	In order to compute this length, we create a new string $S'$ by converting each block into a unique symbol. 
	To detect repeated occurrences of the same block, and give them the same symbol, we sort the blocks in linear time using a technique from~\cite{GK17}.
	The result string is of length $\O(\frac n\tau)$ and the alphabet size is also $\O(\frac n\tau)$, so the algorithm builds the suffix tree of this string in $\O(\frac n\tau)$ time and space. Using this suffix tree augmented with an LCA data structure, one can compute the  value of $\ell'$ in constant time.

	As described above, given the partitioning of $S$ induced by the selected positions, for any two indices $i,j\in[n]$ with $\LCE(i,j)=\ell>\Omega(\tau)$ the substrings $S[i..i+\ell-1]=S[j..j+\ell-1]$ are composed of a short prefix, a short suffix, and a middle part which is composed of exactly the same blocks in both substrings. 
	Thus, the computation of $\LCE(i,j)$ follows this decomposition and computes the $\LCE$ in three phases. In the first phase, the algorithm compares $\O(\tau)$ pairs of characters (or $\O(\tau\log^*n)$ for the deterministic selection). 
	If there is an exact match, the algorithm finds the beginning of the middle part and computes its end using the suffix tree of the blocks. Then, the algorithm just compares additional $\O(\tau)$  pairs of characters (or $\O(\tau\log^*n)$ for the deterministic selection) after the end of the middle part, and it is guaranteed that the mismatch will be found.

	\paragraph{Sparse Suffix Tree construction.}

	It is well-known that the $\SST$ of a set $B\subseteq[n]$ of suffixes can be built in $\O(|B|)$ time and space given the Sparse Suffix Array (SSA) and the LCP array~\cite{KLAAP01} of $B$, which are the lexicographic order of the suffixes, and the longest common prefix of each adjacent suffixes due to this order, respectively.
	Hence, our focus is on sorting the suffixes and on computing their LCP.
	Our algorithm works in two steps.
	At the beginning, the algorithm selects a partitioning set $P$ for $\tau=\frac n{|B|}$ and builds the SSA of all suffixes starting at positions of $P$.
	Then, the algorithm uses the SSA of $P$ to construct the SSA of $B$.

	Our novel technique for computing the SSA of $P$ is based on a special property that is guaranteed by our algorithms that selects partitioning sets.
	This property ensures that if two suffixes of positions from $P$ share a long prefix, then both prefixes are composed of the same blocks in the partitioning, except for possibly a small portion at the end of the common prefix.
	Intuitively, the algorithm sorts the blocks and creates a new string $S'_P$ of length $|P|$, such that the $i$th character in $S'_P$ is the rank of the $i$th block.
	Then, we prove that the lexicographic order of any two suffixes of $S$ starting at positions from $P$ is exactly the same as the lexicographic order of the corresponding suffixes of $S'_P$.
	The actual construction is a little bit more complicated since it is required to extend the blocks we sort in order to achieve the aforementioned result.

	On the second stage for creating the SSA of $B$, the algorithm uses methods which are similar to the methods that were used to compute the SSA of $P$.
	The algorithm creates for any suffix $i\in B$ a short representative string, which basically, ignoring some details, is a prefix of the $i$th suffix of a proper length. 
	Then the algorithm sorts all the representative strings of the suffixes. 
	Ranking of the representative strings is sufficient to determine the order of two suffixes with different corresponding representative strings.
	To determine the order of two suffixes which have the same representative string we use the SSA of $P$ from the first stage.
	The length of the representative string ensures that if two suffixes from $B$ share the same representative string then their order can be determined solely by the order of two corresponding suffixes from $P$. 
	Thus, by combining the representative strings ranking with the SSA of $P$, the algorithm can sort all suffixes of $B$ in $\O(n)$ time using only $\O(\frac n \tau)$ words of space.

	\paragraph{Even better Deterministic $\LCE$ data structure.}
	Using the $\SST$ construction algorithm, we show how to combine the deterministic selection with the  difference covers~\cite{Maekawa85,BK03} technique, to get almost optimal trade-off for the $\LCE$ data structure of $\O(\frac n\tau)$ space and $\O(\tau\sqrt{\log^*n})$ query time with $\O(n\log \tau)$ deterministic construction time. 
	We mention that although the technique of difference covers has already been used for the $\LCE$ problem (see~\cite{PT08,BGKLV15,GK17}), the usage in our algorithm is different. 
	While in all previous work the considered positions were all the positions in a given range, in our algorithm the positions are a subset of the locally selected positions.
	Thus, some properties such as the ability to compute in constant time the common difference for any two indices are impossible to exploit. 
	Even though we need to spend time on finding the synchronized selected positions, our selection gives an  $\O(\tau\sqrt{\log^*n})$ bound on the number of comparison steps needed to take until the synchronized positions are found.
	
	\paragraph{Organization.}
	In section~\ref{sec:preliminaries} we give the basic definitions and tools we use in the paper. Then in Section~\ref{sec:LCE_construction} we give the formal definition for the set of selected positions and show how to use it to create an $\LCE$ data structure. We complete the data structures with the randomized Las-Vegas position selection in Section~\ref{sec:randomSelection} and with the deterministic selection in Section~\ref{sec:determinsiticSelection}. In Section~\ref{sec:SST_construction} we show how to use these selections in order to achieve $\SST$ construction algorithms.
	Sections~\ref{sec:better_deterministic} and~\ref{sec:ens_rand_whp} show how to improve the basic results, where Section~\ref{sec:better_deterministic} discuss on the trade-off improvement for deterministic $\LCE$ data structure, and Section~\ref{sec:ens_rand_whp} discuss on how to ensure the Las-Vegas selection will run in linear time  with high probability. In these sections we omit the discussion about the case of periodic block, which we postpone to the appendices, so in Section~\ref{sec:periodic_overview} we overview the main technique that we use for this case.

	\paragraph{Acknowledgments.} We thank an anonymous reviewer for his usefull comments and suggestions.

	\section{Preliminaries}\label{sec:preliminaries}

	For $1\leq i<j\leq n$ denote the \emph{interval} $[i..j] = \{i,i+1,\dots,j\}$, and $[n]=\{1,2,\dots,n\}$.
	We consider in this paper strings over an integer alphabet $\Sigma=\{1,2,\dots, n^{\O(1)}\}$.
	A string $W$ of length $|W|=\ell$ is a sequence of characters $W[1]W[2]\dots W[\ell]$ over  $\Sigma$.
	A \emph{substring} of $W$ is denoted by $W[x..y]=W[x]W[x+1]\dots W[y]$ for $1\leq x \leq y \leq \ell$. If $x=1$ the substring is called a \emph{prefix} of $W$, and if $y=\ell$, the substring is called a \emph{suffix} of $W$.

	A prefix of $W$ of length  $y\geq 1$ is called a \emph{period} of $W$ if and only if $W[i]=W[i+y]$ for all $1\leq i \leq \ell-y$.
	The shortest period of $W$ is called \emph{the principal period} of $W$, and its length is denoted by  $\rho_W$.
	If $\rho_W\leq \frac{|W|}{2}$ we say that $W$ is \textit{periodic}.
	
	\paragraph{Sparse suffix trees.} Given a set of strings, the compact trie~\cite{Morrison68}
	of these strings is the tree induced by shrinking each path of nodes of degree one in the trie~\cite{DeLaBriandaisde59,Fredkin60} of the strings.
	Each edge in the compact trie has a label which is a substring of part of the given strings. Each node in the tree is associated with a string, which is the concatenation of all the labels of the edges  of the simple path from the root to the node. 
	The suffix tree of a  string $S$  is the compact trie of all the suffixes of $S$.
	The sparse suffix tree of a string $S$ of length $n$, with a set $B\subseteq[n]$ is the compact trie of all the suffixes $\{S[i..n]\mid i\in B\}$.

	\paragraph{Fingerprints.}
	Given $n$, for the following let $u,v\in \bigcup _{i=0} ^n {\Sigma^i}$  be two strings of size at most $n$.
	Porat and Porat~\cite{PP09} and Breslauer and Galil~\cite{BG14} proved that for every constant $c>1$ there exists a \textit{fingerprint function} $\phi : \bigcup _{i=0} ^n {\Sigma^i} \rightarrow[n^c]$, which is closely related to the hash function of Karp and Rabin~\cite{KR87} and Dietzfelbinger et al.~\cite{DGMP92}, such that:
	\begin{enumerate}
		\item If $|u|=|v|$ and $u\neq v$ then $\phi(u)\neq\phi(v)$ with high probability (at least $1-\frac{1}{n^{c-1}}$).
		\item \emph{The sliding property:} Let $w$=$uv$ be the concatenation of $u$ and $v$. If $|w|\leq n$ then given the length and the fingerprints of any two strings from $u$,$v$ and $w$, one can compute the fingerprint of the third string in constant time.
	\end{enumerate}

	\paragraph{Approximately Min-Wise Hash Function.}
	We say that  $\mathcal{F}$ is an $(\epsilon,k)$-min-wise independent hash family if for any $X\subset [n], |X| < k$ and $ x\in [n] \setminus X$ when picking a function $h\in\mathcal{F}$ uniformly we have:
	$ \frac{1-\varepsilon}{|X|+1} \le \Pr_h\left[h(x) < \min\{h(y)\mid {y\in X}\}\right] \le  \frac{1+\varepsilon}{|X|+1}$
	
	Indyk~\cite{Indyk01} have proved that any $\O(\log \frac{1}{\epsilon})$-independent hash family is also a $(\epsilon,\frac{\epsilon n}{c})$-min-wise independent hash family, for some constant $c>1$. Furthermore, since the family of polynomials of degree $\ell-1$ over $GF(n)$, assuming $n$ is a prime, is an $\ell$-independent hash family, then for $\ell=\O(\log \frac{1}{\epsilon})$ the polynomials' family is $(\epsilon,\frac{\epsilon n}{c})$-min-wise independent hash family. For any polynomial from the family, evaluating its value for a given input can be done in $\O(\log \frac{1}{\epsilon})$ time.

	\paragraph{Other tools.}
	Other useful lemmas and theorems are presented in Appendix~\ref{app:useful_tools}.

	\section{\LCE{} with Selected Positions}\label{sec:LCE_construction}
	Following the discussion above in Section~\ref{sec:algorithmic_overview}, we introduce here the formal definition of the selected positions that defines the partitioning of $S$ into blocks.

	The partitioning of $S$ is induced by a set $P$ of indices from $[n]$, such that each block is a substring beginning in some index of  $P\cup\{1\}$ and the concatenation of all the blocks is exactly $S$. 
	The following definition introduces the essential properties of the set $P$, required for our data structures.
	
	\begin{definition}\label{def:landmarks_set}

		A set of positions $P\subseteq [n]$ is called a \emph{{$(\tau,\delta)$-partitioning set}} of $S$, for some parameters $1\le \tau\le\delta\le n$, if and only if it has the following properties:
		\begin{enumerate}
			\item{\emph{\textbf{Local Consistency}} --- } 	\label{item:blocks_sync}
			For any two indices $i,j\in[1+\delta..n-\delta]$ such that $S[i-\delta..i+\delta]=\linebreak S[j-\delta..j+\delta]$ we have $i\in P \Leftrightarrow j\in P$.

			\item{\emph{\textbf{Compactness}} --- }  \label{item:blocks_size}
  			Let $p_i<p_{i+1}$ be two consecutive positions in $P\cup\{1,n+1\}$ then we have one of the following.
			\begin{enumerate}
				\item \label{item:regular_block}\emph{(Regular block)}  $p_{i+1}-p_i\le \tau$. 
				\item \label{item:long_blocks}\emph{(Periodic block)} $p_{i+1}-p_i> \tau$ such that $u=S[p_i..p_{i+1}-1]$ is a periodic string with period length $\rho_u\le\tau$. 
				
			\end{enumerate}

		\end{enumerate}
	\end{definition}
	Here we describe how to use a $(\tau,\delta)$-partitioning set of size $\O(\frac{n}{\tau})$ in order to construct an $\LCE$ data structure in $\O(n)$ time using $\O(\frac n\tau)$ words of space. 
	The data structure uses $\O(\frac n\tau)$ words of working space and answers queries in $\O(\delta)$ time.
	In Section~\ref{sec:randomSelection}, we introduce a randomized Las-Vegas algorithm that finds a $(2\tau,2\tau)$-partitioning set of size $\O(\frac{n}{\tau})$ in $\O(n)$ expected time (which we later improve to be with high probability in Section~\ref{sec:ens_rand_whp}), using only $\O(\frac n\tau)$ words of space. 
	In Section~\ref{sec:determinsiticSelection}, we introduce a deterministic algorithm that finds an $(\O(\tau),\O(\tau\log^*n))$-partitioning set of size $\O(\frac{n}{\tau})$ in $\O(n\log\tau)$ time using only $\O(\frac n\tau + \log \tau)$ words of space. 
	
	\paragraph{Remark about border cases.} When using a $(\tau,\delta)$-partitioning set we assume that all the queries to the algorithm are made with indices from the range $[1+\delta..n-\delta]$. 
	If this is not the case, the algorithm considers  string $S'=\#^\delta S \#^\delta$ where $\#$ is a special character. 
	Notice that $|S'|=n+2\delta=\O(n)$ and if $S$ is given in a read only memory, the algorithm is able to simulate $S'$ with additional constant space.

	\paragraph{The partitioning string.}
	Given a $(\tau,\delta)$-partitioning set $P$ of size $\O(\frac{n}{\tau})$ we consider the partitioning of $S$ induced by $P$. 
	Let $p_1<\cdots<p_h$ be all the indices in $P\cup\{1,n+1\}$, then  the $i$th block of $S$ (for $1\le i< h$) is $S_i=S[p_i..p_{i+1}-1]$.
	Since there are at most $\O(\frac n\tau)$ blocks, we associate with each block a unique symbol from $\Sigma'=\{1,\dots,\O(\frac n \tau)\}$, such that if two blocks are equal (correspond to the same string) they have the same symbol. 
	The unique symbols are determined by sorting the blocks in $\O(n)$ time due to Lemma~\ref{lem:string_sorting}, where the only exception are periodic blocks whose length is larger than $\tau$. For those blocks, we take only the first $2\tau$ characters into the sorting and add at the end a new character which is the length of the block (we assume that the alphabet of $S$ contains $\{1,2,\dots,n\}$, otherwise we extend the alphabet). The value of each block is its rank among all the (distinct) blocks.
	It is straightforward that the value of each block is a unique identifier which equals among all copies of the block.
	Using these values, we consider the \emph{partitioning string} $S_P$ of length $h-1$ over $\Sigma'$, defined such that $S_P[i]$ is the value of the $i$th block $S[p_{i}..p_{i+1}-1]$.  The string $S_P$ is constructed in $\O(n)$ time from the partitioning set $P$ and its length is $\O(\frac n\tau)$.

	\paragraph{The data structure.}
	The algorithm builds a suffix tree $T_P$ on the string $S_P$ using the algorithm of Farach-Colton~\cite{Farach97} in $\O(\frac n\tau)$ time. 
	Each node in $T_P$ maintains, as its \emph{actual depth}, the length of the substring of $S$  (over $\Sigma$) corresponding to the substring of $S_P$ written on the path from the root of $T_P$ to the node.
	Moreover, $T_P$ is preprocessed to support  LCA queries (see Theorem~\ref{lem:LCA}), thus for any two leaves  one can compute the length of the longest common prefix composed of complete blocks for the two corresponding positions in $P$ in constant time.
	
	All the positions in $P$ are maintained in an ordered linked list, and each $p_i\in P$ is maintained with a pointer to the leaf in $T_P$ corresponding to the block beginning at position $p_i$.
	In addition, let $\suc_P(\alpha)=\min\{p\in P\,|\,p\ge \alpha\}$ be the successor of position $\alpha$,  the data structure stores an array $\mathit{SUC}$ 
	of length $\floor{n/\tau}$ such that $\mathit{SUC}[j]$ stores a pointer to the node of $\suc_P(j\cdot\tau)$   in the linked list.
	Notice that for any $i\in[n]$, the successor $\suc_P(i)$ can be found in $\O(\tau)$  time by sequential search on the list from the position pointed to by $\mathit{SUC}\left[\floor{\frac i \tau}\right]$. 
	
	\paragraph{Query Processing.}
	Due to property~\ref{item:blocks_size} of Definition~\ref{def:landmarks_set} each block is either of length at most $\tau$ or it has period length of at most $\tau$.
	For the sake of intuition we first consider the case where the length of each block is at most $\tau$. 
	Then, in Appendix~\ref{app:LCE_with_periodic_blocks} we describe how to generalize the processing  for the case where there exist large blocks with a short period as well.
	
	We consider the computation of $\LCE(i,j)=\ell$ for some $i,j\in[n]$. The computation consist of (at most) three phases. At the first phase the algorithm compares  $S[i+k]$ and $S[j+k]$ for $k$ from $0$ to $3\delta\ge 2\delta+\tau$. If a mismatch is found, then the minimum $k$ such that $S[i+k]\ne S[j+k]$ is $\LCE(i,j)$. Otherwise, since  $\LCE(i,j)>3\delta$ we use the partitioning into blocks.  We use the following lemma that states that all the blocks of $S[i..i+\ell-1]$ and $S[j..j+\ell-1]$ are the same except for the margins of at most $\delta$ characters. The proof of the lemma appears in Appendix~\ref{app:LCE_missing_proof}.
	\begin{lemma}\label{lem:center_blocks}
		Let $i,j\in[n]$ be two indices and let $\ell=\LCE(i,j)$. If $\ell>2\delta$ then: 
		$$\{p-i\,|\,p\in(P\cap[i+\delta..i+\ell-\delta-1])\}=\{p-j\,|\,p\in(P\cap[j+\delta..j+\ell-\delta-1])\}.$$
	\end{lemma}

	In the second phase,  the algorithm finds a correlated offset $\alpha=\min\{p-i\,|\,p\in(P\cap[i+\delta..i+\ell-\delta])\}$ which due to Lemma~\ref{lem:center_blocks} has the property that $i+\alpha\in P$ and $j+\alpha\in P$.
	The computation of $\alpha$ can be done in $\O(\tau)$ time by retrieving $\suc_P(i+\delta)$. 
	In addition, it must be that the substrings $S[i+\alpha..i+\ell-1]$ and $S[j+\alpha..j+\ell-1]$ are partitioned into blocks exactly the same, except for possibly the last $\tau+\delta$ characters.
	The algorithm finds the end of the last common block using an LCA query on $T_P$ in constant time.
	Since we assume that all the blocks are of length $\tau$ at most, it must be that the offset of this block is at least $\ell-\delta-\tau-1$.
	Hence, in the third phase the algorithm compares the characters from the end of the last common block until a mismatch is encountered.
	It is guaranteed that such a mismatch will be found after at most  $\tau+\delta$ comparisons.

	To summarize, we show in this section the following lemma.

	\begin{lemma}\label{lem:LCE_construction}
		Given a $(\tau,\delta)$-partitioning set of size $\O(\frac n\tau)$ there exists a deterministic construction algorithm of an $\LCE$ data structure that takes $\O(n)$ time using $\O(\frac n\tau)$ words of space. The space usage of the data structure is $\O(\frac n\tau)$ and its query time is $\O(\delta)$. 
	\end{lemma}
	
	Thus, combining Lemma~\ref{lem:LCE_construction} with the randomized  selection of Sections~\ref{sec:randomSelection} and~\ref{sec:ens_rand_whp} (Lemma~\ref{lem:random_selection_whp}) we have proved Theorem~\ref{thm:lasvegasLCE}. 
	Similarly, combining Lemma~\ref{lem:LCE_construction} with the deterministic  selection of Section~\ref{sec:determinsiticSelection} (Corollary~\ref{cor:determinstic_set_complete}) we have proved Lemma~\ref{lem:deterministicLCE}.
	
	\section{Randomized Selection}\label{sec:randomSelection}

	Using fingerprint hash functions and approximately min-wise independent hash functions, we prove that a $(2\tau,2\tau)$-partitioning set consisting of $\O(\frac{n}{\tau})$ positions can be constructed in linear time with high probability. In this section we show a method for constructing such a partitioning set of size  $\O(\frac{n}{\tau})$ using linear time in expectation. Later, in Section~\ref{sec:ens_rand_whp} we improve the method to ensure linear construction time with high probability. All the algorithms use $\O(\frac n\tau)$ words of working space.
	
	We define a function $\ID:[1..n-\tau+1]\rightarrow \mathbb R$ 
	such that $\ID(j)$ depends only on $S[j..j+\tau-1]$, hence  if $S[j..j+\tau-1] = S[k..k+\tau-1]$ for some indices $j$ and $k$ then $\ID(j)=\ID(k)$. 
	Using \ID{}, we create $P$, a partitioning set of $S$. For each interval $[i..i+\tau-1]$ the positions within the interval with the smallest \ID{} value are picked for $P$. In other words,
	$$P = \left\{j\in[1..n-\tau] \mid \exists \ell\in[j-\tau+1..j] : \ID(j) = \min_{k\in[\ell..\ell+\tau-1]}\left\{\ID(k)\right\}\right\}$$
	
	From the locality properties of the $\ID$ function, the following lemma follows immediately.
	
	\begin{lemma}\label{lemma:id_is_partition_set}
		Given an $\ID$ function as described above, the set $P$ induced by $\ID$ is a  $(2\tau,2\tau)$-partitioning set.
	\end{lemma}
	\begin{proof}We will prove that the two properties of Definition~\ref{def:landmarks_set} hold.
		
		\paragraph{Local Consistency.} 
		Let $i,j\in[1+2\tau..n-2\tau]$ such that $S[i-2\tau..i+2\tau]= S[j-2\tau..j+2\tau]$.
		For any $a\in [i-2\tau..i+\tau+1]$, since $\ID(a)$ depends only on $S[a..a+\tau-1]$, we have that $\ID(a)=\ID(a-i+j)$.
		
		Assume $i\in P$ and let $\ell\in [i-\tau+1..i]$ be an index such that $\ID(i)=\min_{k\in[\ell..\ell+\tau-1]}\{\ID(k)\}$. 
		We have that $[\ell..\ell+\tau-1]\subset [i-\tau+1..i+\tau-1]\subset[i-2\tau..i+\tau+1]$. 
		Hence, for any $k\in[\ell..\ell+\tau-1]$ we have $\ID(k)=\ID(k-i+j)$. In particular,
		$$\ID(j)=\ID(i)=\min_{k\in[\ell..\ell+\tau-1]}\{\ID(k)\}=\min_{k\in[\ell-i+j..\ell-i+j+\tau-1]}\{\ID(k)\}.$$
		Thus, due to the definition of $P$, $j\in P$.
		Hence, we prove that $i\in P \Rightarrow j\in P$. The opposite direction is symmetric.
		
		\paragraph{Compactness.}
		Let $1=p_0<p_1<\cdots<p_k<p_{k+1}=n+1$ be all the positions in $P\cup\{1,n+1\}$.
		Notice that if the first selected position is not $p_0=1$, then $p_1$ is selected from the interval $[1..\tau]$, and therefore $p_1 - p_0 < \tau$.
		Moreover, the rightmost position in $P$, $p_k$, is the position picked on the interval $[n-2\tau+1..n-\tau]$.
		Thus, $p_{k+1}-p_k = n+1 - p_k \le 2\tau$ 
		Finally, for the general case where both $p_i,p_{i+1}\in P$ assume by a  contradiction that  $p_{i+1} - p_i > \tau$.
		Therefore no position is selected in the interval $[p_i+1..p_i+\tau] \subseteq[p_i+1..p_{i+1}-1]$, and by definition of $P$ this cannot happen.
	\end{proof}
	
	\paragraph{Remark about the periodic case.}
	Notice that if $S$ contains a long periodic substring with a small period, then the same \ID{}s will appear in every period (except perhaps in the last $\tau$ positions within the periodic substring). 
	In particular, many positions will have the minimum \ID{} value, hence $P$ will be a very large partitioning set, no matter which \ID{} function we use.
	This means that long periodic substring with small periods should be taken with special care. Formally we define:
	\begin{definition}\label{def:drho-run}
		A substring $S'=S[i..j]$ is a \emph{$(d,\rho)$-run} of $S$  if and only if $|S'|\ge d$ and the principal period length of $S'$ is at most $\rho$.
	\end{definition}
	
	\paragraph{Overview of the solution.}
	We show here the solution for the non-periodic case: We prove that when $S$ contains no $(\tau,\frac{\tau}{6})$-runs, it is possible to construct a partitioning set of $\O(\frac{n}{\tau})$ positions in expected linear time.
	In Appendix~\ref{app:rand_selection_periodic_blocks}, we show that all $(\tau,\frac{\tau}{6})$-runs can be found in linear time and constant space using fingerprints.
	We then combine the two methods to partition the string into blocks \textemdash find all $(\tau,\frac{\tau}{6})$-runs and mark them as blocks, and then use the non-periodic algorithm to partition the substring between every two runs in a total of expected linear time and $\O(\frac{n}{\tau})$ space.

	\paragraph{The \ID{} function.}
	Our goal is to use an $\ID$ function that induces a partitioning set of size $\O(\frac n\tau)$ that we will be able to compute in $\O(n)$ time in expectation.
	To achieve this, we fix a fingerprint function $ \phi $ and an approximately min-wise hash function $ h\in \mathcal{F} $ where $\mathcal{F}$ is a $\left(\frac{1}{2},\tau\right)$-min-wise independent hash family (see Section~\ref{sec:preliminaries}). 
	We denote  $ \phi_i = \phi\left(S[i..i+\tau-1]\right)$ and the algorithm uses the function $\ID_{\phi,h}(i) = h(\phi_i)$.
	The fingerprint function is used for the locality of $\ID_{\phi,h}$, and the hash function is used for diversifying the fingerprints that will have minimal value.
	
	As mentioned before, we assume $S$ contains no  $(\tau,\frac{\tau}{6})$-runs.
	We show that when using $\ID_{\phi,h}$ where $\phi$ and $h$ are picked uniformly, the size of the partitioning set $P$ is $\O(\frac{n}{\tau})$ in expectation.
	
	\begin{lemma}\label{lem:space_exp}
		If $\phi$ and $h$ are picked uniformly and $S$ contains no $(\tau,\frac{\tau}{6})$-runs then $\E_{\phi,h}\left[|P|\right] = \O\left(\frac{n}{\tau}\right)$.
	\end{lemma}
	\begin{proof}[Proof Sketch; full proof appears in Appendix~\ref{app:rand_selection_missing_proofs}]
		We assume that there are no fingerprints collisions, since the contribution to the expectation of the other case is constant due to the small probability of fingerprints collision.
		We use the linearity of expectation, which means that  
		$\E_{\phi,h}[|P|] = \linebreak\sum_{i=1}^n{\Pr_{\phi,h}(A_i)}$ where $A_i$ is the event that $i\in P$.
		We bound $\Pr_{\phi,h}(A_i)$ by noticing that for any $i$, we have $i\in P$ only if $h(\phi_i)$ is the minimum ID among all the IDs of $[i-\tau/2..i]$ or the minimum among all the IDs of $[i..i+\tau/2]$.
		Due to the almost min-wise function $h$, the probability of each one of these events is $\O(\frac 1\tau)$ and therefore $\Pr_{\phi,h}(A_i)\le \O(\frac 1\tau)$. Hence, $\E_{\phi,h}[|P|] = \sum_{i=1}^n{\Pr_{\phi,h}(A_i)}=\O(\frac n\tau)$. 
	\end{proof}

	\begin{lemma}\label{lemma:id_alg}
		The partitioning set $P$ induced by the function $\ID_{\phi,h}$ on a substring containing no $(\tau,\frac \tau 6)$-runs  can be constructed in expected $\O(n)$ time using expected $\O(\frac{n}{\tau})$ words of space.
	\end{lemma}
	\begin{proof}
		For every interval of size $\tau$, the algorithm will find the positions with minimum \ID{}s and add them to $P$, using only $\O(1)$ additional space.
		When sliding the interval by one step, using the sliding property of the fingerprints, the next \ID{} is calculated in $\O(1)$ time.
		Let $p_j\in P$ be the rightmost position picked for $P$ for some interval $[k..k+\tau-1]\subset [1..n-\tau]$.
		Since $p_j$ has the minimum \ID{} value in that interval, then if $p_j\in[k+1..k+\tau]$ the algorithm can find out in $\O(1)$ if $k+\tau\in P$ by calculating  $\ID(k+\tau)$ and comparing it with $\ID(p_j)$.

		Otherwise, if $p_j = k$, then the algorithm must recalculate all the \ID{}s on the interval $[k+1..k+\tau]$ to find the position with the minimum \ID{} value. 
		This recalculation takes $\O(\tau)$ time due to the sliding property of the fingerprint, and it happens only if the new interval does not contain the currently rightmost position in $P$ (i.e., $p_j=k$).
		Therefore we can charge the $\O(\tau)$ time on $p_j$. Notice that $p_j$ can be charged at most once since immediately after we charge $p_j$, a new rightmost position in $P$ is picked.
		
		Therefore, the algorithm runs in $ \O(n+\tau\cdot |P|)$ using $\O(|P|)$ space, which by Lemma~\ref{lem:space_exp} is expected to be $\O(n) $ time using expected $\O(\frac{n}{\tau})$ words of space. 
	\end{proof}

	Due to the Lemma~\ref{lem:space_exp} and Lemma~\ref{lemma:id_alg}, using Markov's inequality we conclude with the following corollary.
	\begin{corollary}\label{cor:randomized-partitioning-withot-period}
		A $(2\tau,2\tau)$-partitioning set of $S$ can be constructed using $\O(\frac{n}{\tau})$ words of space with expected $\O(n)$ construction time, assuming $S$ contains no $(\tau,\frac{\tau}{6})$-runs.
	\end{corollary}

	\section{Deterministic Selection}\label{sec:determinsiticSelection}
	In this section we introduce a deterministic algorithm that finds a  $(\O(\tau),\O(\tau\log^*n))$-partitioning set of size $\O(\frac n\tau)$.
	The key concepts we use are based on the  techniques of Cormode and Muthukrishnan~\cite{CM07}, which in turn are similar to other previous papers~\cite{GPS87,MSU97,CV86,SV94,FIK16}.
	Here we focus on the definition of the positions selected by the algorithm. 
	Later, in Appendix~\ref{app:deterministic_in_space} we describe how to select these positions using only $\O(\frac n\tau+\log \tau)$ words of space in $\O(n\log \tau)$ time. The complete correctness proof appears in Appendix~\ref{app:deterministic_correctness}.

	The algorithm creates a hierarchical decomposition of $S$, such that the top-level has $\O(n/\tau)$ blocks. 
	At the bottom level (level $0$) each character in $S$ is considered as a block, so the set of picked positions is all $[n]$. 
	We create levels in a bottom-up fashion, preserving the following properties. At each level $\mu$ there are at most $\O\left(\frac n {(3/2)^\mu}\right)$ blocks. Moreover, the length of each block is $\O((3/2)^\mu)$, except for blocks which are substrings of $S$ that have period length of at most $(3/2)^\mu$, which can be longer. 
	In addition, the locality of the partitioning is obtained by limiting the dependency of each block in level $\mu$ to at most $\O(\log^*n)$ blocks in level $\mu-1$. 
	These properties are a generalization of  Definition~\ref{def:landmarks_set} (partitioning set) of size $\O(\frac n\tau)$, so by repetitively applying the transformation from level $\mu=0$ up to level $\mu=\log_{3/2}\tau$ we obtain  a $(\O(\tau),\O(\tau\log^*n))$-partitioning set. 
	
	Now we explain how the decomposition of level $\mu$ is defined based on the decomposition of level $\mu-1$.
	For clarity of representation, when describing the block of level $\mu$, we call the blocks of level $\mu-1$ \emph{sub-blocks}.
	Given the  $(\mu-1)$th level of the hierarchy, we consider parsing of $S$ into four types of maximal non-overlapping substrings.
	
	\begin{enumerate}
		\item \label{item:type_large} Single sub-block of length at least $\left( 3 /2\right)^\mu$.
		\item \label{item:type_contiguous} Maximal contiguous sequence of eqaul sub-blocks, each of which is not of type~\ref{item:type_large}.
		
		\item \label{item:type_cointossing} Sequences of at least $c\log^* n$ sub-blocks, that are not of type~\ref{item:type_large} or~\ref{item:type_contiguous} (where $c$ is a constant to be determined).
		\item \label{item:type_edge}Sequences of less than $c\log^* n$ sub-blocks, that are not of type~\ref{item:type_large} or~\ref{item:type_contiguous}.
	\end{enumerate}
	We identify each block with the index of its leftmost endpoint, which we called the \emph{representative position} of the block. 
	As a consequence, position $1$ is a representative position at any level of the hierarchy.
	The algorithm treats each type of such maximal substrings separately. 
	For the sake of simplicity, we describe it now as an algorithm which has no space constraints, which we overcome in Appendix~\ref{app:deterministic_in_space}.

	The first type contains sub-blocks of length at least $(3/2)^\mu$ and the algorithm just maintains them as blocks in level $\mu$. 
	The substrings of the second type are sequences of (at least two) identical sub-blocks. In this case, the algorithm merges all the sub-blocks in each sequence into one large block.
	To identify such sub-blocks, the algorithm compares each sub-block with its successor, and if they are same, the algorithm merges them in level $\mu$.
	
	\paragraph{The alphabet reduction.} For the third type of substrings, we use the following method. For the sake of intuition, we begin with the special case where $\mu=1$, which is the transformation from level $0$ (where each character of $S$ has a unique block) to the next level. We follow the technique of~\cite{CM07} and reduce  each character into a number from the set $\{0,1,2,3,4,5\}$ using a \emph{local} method. 
	The alphabet reduction is done in $\O(\log^*n)$ iterations, which we call \emph{inner-iteration}s (to distinguish them from the \emph{outer}-iterations of $\mu$ from $0$ to $\log_{3/2}\tau$). 
	We assume that the alphabet of $S$ is $\Sigma=\{1,2,\dots,|\Sigma|\}$ of size $|\Sigma|=\O(n^{\O(1)})$. Hence, before the reduction begins each character can be considered as a bit string of length $\O(\log n)$. At each inner iteration, the algorithm reduces the length of the binary string logarithmically. 
	We denote by $\lab(i,0)=S[i]$ the character of $S$ in position $i$. 
	For any two consecutive characters of the $(j-1)$th inner iteration, $\lab(i,j-1)$ and $\lab(i+1,j-1)$, we define $\lab(i,j)$ as follows. Let $\psi$ be the index of the least significant bit in which $\lab(i,j-1)$ and $\lab(i+1,j-1)$ are different\footnote{Notice that the label of position $i$ depends on the label of position $i+1$, instead of $i-1$. This change from the original Locally Consistent Parsing, is very useful in the $\SST$ construction since it makes the partitioning set \emph{forward synchronized} (see Definition~\ref{def:forward_sync}).},
	and let $bit(\lab(i,j-1),\psi)$ be the value of the $\psi$th bit of $\lab(i,j-1)$. We define $\lab(i,j)=2\psi+bit(\lab(i,j-1),\psi)$, i.e., the concatenation of the index $\psi$ with the value of the bit in this index.

	\begin{lemma}[{\cite[Lemma 1]{CM07}}]
		For any $i$, if $S[i]\ne S[i+1]$ then $\lab(i,j)\ne \lab(i+1,j)$ for any $j\le 0$.
	\end{lemma}
	Each inner iteration of this reduction, reduces an alphabet of size $\sigma$  to an alphabet of size $2\ceil{\log\sigma}$. 
	We run such inner-iterations until the size of the alphabet is constant. 
	Hence,  after $\O(\log^*n)$ inner-iterations the size of the alphabet becomes constant, we denote by $c$ the appropriate constant for the $\O(\log^*n)$ inner-iterations.
	By the following lemma, the size of the final alphabet is exactly $6$. 
	
	\begin{lemma}[{\cite[Lemma 2]{CM07}}]
		After the final inner-iteration of alphabet reduction, the alphabet size is $6$.
	\end{lemma}

	The main challenge with this technique is how to repeat this process with larger and larger blocks. In~\cite{CM07} there are two methods: a randomized algorithm that uses fingerprints to reduce the size of each block to a constant number of words ($\O(\log n)$ bits), or an algorithm with expensive $\O(n)$ additional working space that uses bucket sort. 
	Since we want a deterministic algorithm, and since we have just $\O(\frac n\tau)$ words of working space we use a different method, which costs us $\log \tau$ factor in the running time and uses the following observation. We observe that each block at any outer-iteration is a string of size at most $\O(n)$, thus we consider each block as a number between one and $\O(|\Sigma|^n)=\O(n^{O(n)})$. 
	Formally, we assume without loss of generality that $0\notin \Sigma$, the alphabet of $S$, and we define  $\Pi=\{0,1,\dots,2^{\ceil{\log_2 (|\Sigma|+1)}n}\}$ as the \emph{blocks' alphabet}. 
	Each block $u=S[x..y]$ at any iteration is a string of length at most $n$ over $\Sigma$, and is represented by one symbol from $\Pi$ by concatenating the binary representation of $u$'s characters, with padding of zeros to the left, i.e., the symbol of $u$ is $\sum_{k=0}^{|u|-1}u[|u|-k]\cdot 2^{\ceil{\log_2(|\Sigma|+1)}k}\in \Pi$. 
	The alphabet reduction for any iteration $\mu>1$ is done exactly as described for $\mu=1$ except that instead of converting characters from $\Sigma$ the algorithm considers the symbols of $\Pi$ induced by the blocks of level $\mu-1$.
	
	Notice that $\Pi=\O(|\Sigma|^n)=\O(n^{\O(n)})$ and therefore, the alphabet reduction for any outer-iteration takes $\O(\log^*(n^{\O(n)}))=\O(\log^*n)$ inner-iterations.
	For any outer-iteration the algorithm executes $\O(\log^*n)$ inner iterations. Each first inner-iteration takes $\O(n)$ time since it has to take into account all the blocks which their length summed to $\Theta(n)$. However, after one inner-iteration each block reduced from $\O(\log n^{\O(n)})=\O(n\log n)$ bits to $\O(\log(n\log n))=\O(\log n)$ bits, and therefore can be maintained in a constant number of words. Thus, after the first inner-iteration each block comparison takes constant time. Since the number of blocks reduces exponentially and it is $\O\left(\frac n{(3/2)^\mu}\right)$, the total time for all the inner iterations, except the first inner iteration of each outer-iteration is $\O(n\log^*n)$. Together with the first inner-iterations the total time required for all the transformations is $\O(n(\log\tau+\log^*n))$. However, using lookup-table, we reduce this time to $O(n\log\tau)$, see Appendix~\ref{app:deterministic_in_space}.

	Remark that at the end of the alphabet reduction, the blocks' symbols are all in $\{0,1,2,3,4,5\}$.
	At this point, the algorithm scans all blocks and compares their symbol's value with those of their right and left neighbors.
	The algorithm picks all positions that have a local minimum symbol value as the selected positions in level $\mu$.
	The only exceptions are regarding the margins of the sequence. 
	For the rightmost $c\log^*n$ characters, the alphabet reduction cannot reach the last level, and therefore they do not have a corresponding number from $\{0,\dots,5\}$, thus, the algorithm treats them as sequences of type~\ref{item:type_edge}.
	In addition, the leftmost position in the sequence is always selected by the algorithm. Moreover, if the second block in the sequence (from left) receives a symbol which is a local minimum, the algorithm ignores it in order to avoid the case that the first sub-block will become a block in level $\mu$.
	So, by the merging due to the minimal symbol, each block of level $\mu$ is composed of at least $2$ sub-blocks of level $\mu-1$ and at most $11$ sub-blocks ($10$ is the maximal distance between minimal positions over the symbols $\{0,\dots,5\}$ and additional one block for the case of the leftmost block).

	\paragraph{Margin sub-blocks (type~\ref{item:type_edge}).}
	Each sequence of type~\ref{item:type_edge} contains at most $c\log^*n$ sub-blocks. If the sequence contains exactly one sub-block, this sub-block is kept as a block in level $\mu$ of the hierarchy.
	If the sequence contains two or three sub-blocks, the algorithm creates one block in level $\mu$ which is the concatenation of all these three sub-blocks. 
	As for the last case, if the sequence contains at least four sub-blocks, the algorithm creates one block in the $\mu$th level which is the union of the two rightmost blocks, and repeats the process from right to left until there are no sub-blocks of level $i-1$ in the sequence which are not part of any block in level $\mu$. 
	Notice that each block created in such sequence contains at least two sub-blocks from the lower level and at most three of them, except for the case where the sequence contains exactly one sub-block.
	
	\paragraph{Conclusion.} The set of positions induced by  level $\mu=\log_{3/2}\tau$  of the hierarchical decomposition, is a $(\O(\tau),\O(\tau\log^*n))$-partitioning set of size $\O(\frac n\tau)$. The  correctness proof appears in Appendix~\ref{app:deterministic_correctness}. Moreover, the algorithm finds this set using only $\O(\frac n\tau+\log \tau)$ words of space in $\O(n\log\tau))$  time, as described in Appendix~\ref{app:deterministic_in_space}.

	\section{Sparse Suffix Tree Construction} \label{sec:SST_construction}
	In this section we introduce a method to build the Sparse Suffix Tree of a set $B\subseteq [n]$ of suffixes of $S$. Since the output size is $\Theta(|B|)$ words of space, we are able to use $\O(|B|)$ words of working space without affecting the space usage of the algorithm. Thus, we will use the $\LCE$ data structure with $\tau=\frac n{|B|}$ that uses $\O(\frac n\tau)=\O(|B|)$ words of space.
	We focus on computing the Sparse Suffix Array (SSA) of the suffixes, which is an array that stores $B$'s suffixes ordered by their lexicographic order. Then, one can compute the longest common prefix (LCP) of any pair of consecutive (by the lexicographic order) suffixes using the $\LCE$ data structure from Section~\ref{sec:LCE_construction}. Using the SSA and  LCP information one can build the desired $\SST$ in $\O(|B|)=\O(\frac n\tau)$ time, using the algorithm of Kasai et al.~\cite{KLAAP01}.
	Our construction algorithms will use the same $(\tau,\delta)$-partitioning sets that were introduced in Section~\ref{sec:randomSelection} and Section~\ref{sec:determinsiticSelection}. 
	We will prove that these sets have a special property that enables us to compute efficiently the SSA of the set $P$ that was selected by the algorithm. 
	Then, we will show how to use this SSA to compute the desired SSA of the set $B$.

	Recall that Lemma~\ref{lem:center_blocks}, which was induced by the definition of $(\tau,\delta)$-partitioning set, said that for any two indices with large common extension both common  substrings are composed of the same blocks except for possibly the right and left margins of length $\delta$. 
	For the $\SST$ construction, we require an additional property. We want that for every pair of selected positions, i.e., positions from $P$, with large $\LCE$, the two suffixes also begin with composition into exactly the same blocks, and the only difference in the composition of the common extensions is at the right margin. 
	Formally, we say that a partitioning set is \emph{forward synchronized} if it has the following property.
	
	\begin{definition}[Forward synchronized partitioning set]\label{def:forward_sync}
		A $(\tau,\delta)$-partitioning set $P$ is called \emph{forward synchronized} if and only if for any $p_i,p_j\in P$, which their consecutive positions in $P$ are $p_{i+1}$ and $p_{j+1}$, respectively, if $\LCE(p_i,p_j)>p_{i+1}-p_i+\delta$ then $p_{i+1}-p_i=p_{j+1}-p_j$.
	\end{definition}

	We will show how to build the sparse suffix tree of a forward synchronized $(\tau,\delta)$-partitioning set of size $\O(\frac n\tau)$ in $\O(n\cdot \frac \delta\tau)$ time, using $\O(\frac n\tau)$ words of space. Fortunately, our two constructions of partitioning sets, from Section~\ref{sec:randomSelection} and Section~\ref{sec:determinsiticSelection} are  forward synchronized, as claimed in the following lemma. The proof is given separately for the randomized algorithm (Lemma~\ref{lem:rand_is_forward_sync}) and for the deterministic algorithm (Lemma~\ref{lem:det_is_forward_sync}).

	\begin{lemma}
		Any $(\tau,\delta)$-partitioning set that is a result of the algorithms from Section~\ref{sec:randomSelection} or Section~\ref{sec:determinsiticSelection} is forward synchronized.
	\end{lemma}
	
	\subsection{SSA of Forward Synchronized  \texorpdfstring{$(\tau,\delta)$}{(t,d)}-Partitioning Set }
	
	Let $P$ be a forward synchronized  $(\tau,\delta)$-partitioning set,  we show here how to compute the SSA of $P$ efficiently.
	Let us consider the decomposition of $S$ into blocks, induced by $P$.
	According to this decomposition, all the suffixes of $S$ beginning in positions from $P$ are composed only of complete blocks. 
	Moreover, since $P$ is forward synchronized, it is guaranteed that whenever comparing two suffixes with long $\LCE$, the two suffixes must begin with identical blocks. 
	As with the definition of $S_P$ in Section~\ref{sec:LCE_construction}, we aim to create a partitioning string $S'_P$ of length $|P|$ such that each block of $S$ will be transformed into one symbol in this string, and any two suffixes begins in positions from $P$ will have the same order as the corresponding suffixes of the partitioning string.
	Recall that when we defined $S_P$ we sorted the substrings of the blocks and transformed each block into its rank. This method seemed to work also for our aim, since when comparing two suffixes, we hope that the first non-identical blocks will determine the order of the suffixes properly, according to the sorting. However, this method has one problem, which is the order of two blocks when one of them is a prefix of the other. In such a case, the order of the suffixes is not determined by these blocks, but by pair of characters which at least one of them is not in these two blocks.

	We use the following method to create a $S'_P$, which overcomes the difficulty described above. For each regular block (of type~\ref{item:regular_block} in Definition~\ref{def:landmarks_set}) beginning at position $p_i\in P$ we consider the \emph{representative string} of the block as $s_i=S[p_i..p_i+3\delta-1]$. Notice that this string may contain also characters outside the block of $p_i$. 
	More detailed, since $\delta\ge \tau$, and due to property~\ref{item:blocks_size} of Definition~\ref{def:landmarks_set}, the only case where $s_i$  do not contain characters outside the block is when the block $[p_i..p_{i+1}-1]$  is a large block with a period length of at most $\tau$ characters.
	Hence, when comparing two suffixes which their first non-identical block is regular (of type~\ref{item:regular_block}), the sorting of the representative strings guarantees that the order of these blocks is equal to the order of the suffixes. Thus, we just need to explain how to treat periodic blocks (of type~\ref{item:long_blocks}). 
	
	\paragraph{Periodic blocks.}
	When comparing two periodic blocks which their prefixes of length $3\delta>2\tau$ are equal, it must be the case that all the corresponding characters in the two blocks match (see Lemma~\ref{lem:twiceperiod_equal}), and therefore the shorter block is a prefix of the longer block. 
	Of course, if the two blocks have exactly the same length then the corresponding substrings of $S$ are identical, and comparison between the suffixes should be done according to the consecutive blocks.

	For the sake of intuition let us consider first the case that all the periodic blocks are maximal. In this case for a periodic block $[p_i..p_{i+1}-1]$ if  $S[p_i..p_{i+1}-1]$ has period of length  $\rho$, then $S[p_{i+1}]\ne S[p_{i+1}-\rho]$.
	In such a case, when comparing two suffixes that have the same period, the lexicographic order of the suffixes is defined by the character right after the end of the short block and the corresponding character in the longer block. 
	We want the representative string of each block to be composed of a long enough prefix of the block, say of length $3\delta$, with some additional characters to deal with the situation of one block that is a prefix of another block. 
	Let us consider two suffixes $p_i,p_j\in P$ that begin with periodic blocks and one is the prefix of the other.
	Assume without loss of generality that the shorter block appears at position $p_i$, and that its length is $\alpha$.
	In this case, the lexicographic order of the suffixes is determined by $S[p_i+\alpha]$ and the corresponding character in the second suffix, $S[p_j+\alpha]$.
	The suffix $S[p_i..n]$ should appear before $S[p_j..n]$ in the lexicographic order if and only if $S[p_i+\alpha]$ is larger than $S[p_j+\alpha]$. 
	Both the block of $p_i$ and the block of $p_j$ have the same period, hence, let $\rho$ be the length of this period.
	Thus, for any position $x$ in the block, we have $S[x]=S[x-\rho]$ except for possibly the first $\rho\le\tau$ positions. 
	In particular, $S[p_j+\alpha]=S[p_j+\alpha-\rho]$. 
	On the other hand, since the shorter block is a prefix of the longer block we have $S[p_j+\alpha-\rho]=S[p_i+\alpha-\rho]$ and so $S[p_j+\alpha]=S[p_i+\alpha-\rho]$. 
	Hence, instead of comparing $S[p_i+\alpha]$ with $S[p_j+\alpha]$ it is sufficient to compare $S[p_i+\alpha]$ with $S[p_i+\alpha-\rho]$. 
	Thus, for the periodic block $[p_i,p_{i+1}-1]$ we focus on comparing the character $S[p_i+\alpha]$, which we call \emph{the violation character} with the character  $S[p_i+\alpha-\rho]$, which we call \emph{the period character}.  
	We distinguish between two cases: either the violation character is larger than the period character, or vice versa. 
	If the violation character is larger than the period character, then this block should be after longer blocks 
	with the same period while sorting the suffixes. 
	Otherwise, if the violation character is smaller than the period character, then this block should appear before longer blocks while sorting the suffixes.
	When comparing two blocks that their violation characters are larger than the period characters, the longer blocks should be first, and when comparing two blocks which their violation characters are smaller than the period characters, the longer block should be last.
	Thus, for periodic blocks of length at least $3\delta$, instead of just defining the representative string $s_i=S[p_i..p_i+3\delta-1]$ we concatenate to this string two symbols.
	If the violation character is smaller than the period character, the first symbol added is $-1$ and the second is the block's length.
	Otherwise, the first symbol added is $+1$ and the second is the negative value of the block's length ($\cdot (-1)$).

	In the previous paragraph, we assumed that the periodic blocks are maximal, and the characters right after the blocks violate the periodicity. 
	However, this is not necessarily the case in a forward-synchronized partitioning set. 
	It is possible that a periodic block will be terminated before the end of the periodic substring. 
	Moreover, in the constructions that we introduce in this paper the periodic blocks are indeed terminated before the right end of the periodic substring. 
	The exact position of the right endpoint could depend on the $\delta$ characters after the end of the periodic substring. Now, the character right after the end of the block could be part of the periodic substring.
	Thus, we define for every periodic block the \emph{right-violation} of the block as the first position after the end of the periodic substring. 
	\begin{definition}
		Let $[p_i..p_{i+1}-1]$ be a periodic block such that the principal period of $S[p_i..p_{i+1}-1]$ is $\rho$, then $\rv(p_i)=\min\{j\mid j>p_i+\rho, S[j]\ne S[j-\rho]\}$.
	\end{definition}

	We remark that both in the randomized construction of the partitioning set and in the deterministic construction, during the construction the algorithm computes the period length of any periodic block. Therefore in $\O(n)$ time the algorithm computes the $\rv(p_i)$ for any $p_i\in P$ which is an endpoint of a periodic block, without affecting the asymptotic run-time.
	
	Following the discussion above, now the violation character of the periodic block $p_i$ is \linebreak $S[\rv(p_i)]$ instead of $S[p_{i+1}]$. Moreover, instead of inserting the length of the block $p_i$ to $s_i$, we insert the distance $\rv(p_i)-p_i$. 
	There is only one case that this construction does not cover.
	In the case when two blocks have the same period, and the same distance to the end of the period --- they will have exactly the same representative string.
	Notice that if the two blocks have the same length, then it is fine that they have the same representative string since the comparison between the suffixes will be done according to the following blocks.
	However, it is possible that the blocks have different lengths.
	Since the right endpoint of the block depends on the $\delta$ characters after the periodic substring, the blocks could be of different lengths only if the $\delta$ characters after the end of the periodic substrings are different. 
	Thus, to cover also this case, we add to the representative string of the periodic block $p_i$ after the sign ($+1$ or $-1$) and number also the substring $S[\rv(p_i)..\rv(p_i)+2\delta-1]$. 
	
	To conclude for a periodic block $[p_i..p_{i+1}-1]$ we define the representative string $s_i$ as to be the concatenation of the following. The substring (1) $S[p_i..p_i+3\delta-1]$, (2) if  the violation character is smaller than the period character, the first character is $-1$ and the second character is the length of the periodic substring beginning at $p_i$ (that is $\rv(p_i)-p_i$), 
	otherwise, the first character is $+1$ and the second character is the negative value of the length of the periodic substring beginning at $p_i$  (that is $(\rv(p_i)-p_i)\cdot (-1)$).
	At the end we add (3) $S[\rv(p_i)..\rv(p_i)+2\delta-1]$.

	\paragraph{Conclusion.}
	To conclude, for any position $p_i\in P$ we have a representative string of length at most $3\delta+2+2\delta=\O(\delta)$ over an alphabet of size $\O(n)$. 
	The algorithm sorts all these substrings, in $\O(n\cdot\frac\delta\tau)$ time due to Lemma~\ref{lem:string_sorting} or Lemma~\ref{lem:string_sorting_log_star} (recall that $\delta=\Theta(\tau)$ for the randomized selection and $\delta=\Theta(\tau\log^*n)$ for the deterministic selection). After the sorting, the algorithm passes over the sorted array, and computes for each $s_i$ its rank, which is the number of distinct $s_j$s smaller than $s_i$.
	We define the string $S'_P$ as the string of length $|P|$ where the $i$th character is the rank of $s_i$. 
	There is a natural bijective mapping between suffixes of $S'_P$ and the suffixes of $S$ beginning at positions of $P$.
	Moreover, in the following lemma we prove that for any $p_i,p_j\in P$ the suffix $S[p_i..n]< S[p_j..n]$ if and only if $S'_P[i..|P|]< S'_P[j..|P|]$ (the $<$ relation denotes here the lexicographic less than relation).
	\begin{lemma}\label{lem:SSA_of_blocks}
		Let $P$ be a forward synchronized  $(\tau,\delta)$-partitioning set and let  $p_i,p_j\in P$ be two distinct positions, then 
		$S[p_i..n]<S[p_j..n]$ if and only if $S'_P[i..|P|]<S'_P[j..|P|]$.
	\end{lemma}
	\begin{proof}
		
		We will prove $S[p_i..n]<S[p_j..n]\Rightarrow$ $ S'_P[i..|P|]<S'_P[j..|P|]$, the proof is symmetric in the second case.
		
		If the representative strings $s_i$ and $s_j$ are different, then by the discussion above, their relative rank determines the order of the corresponding suffixes, and we have $S'_P[i..|P|]<S'_P[j..|P|]$.
		
		Otherwise, if $s_i=s_j$ we will prove that $S[p_i..p_{i+1}-1]=S[p_j..p_{j+1}-1]$.
		When proven, it follows that $S[p_i..n]<S[p_j..n]$ if and only if $S[p_{i+1}..n]<S[p_{j+1}..n]$.
		Using an induction argument on the length of the sequence of equal pairs $s_{i+k},s_{j+k}$, we have that if $S[p_i..n]<S[p_j..n]$ then $S'_P[i..|P|]<S'_P[j..n]$.
		
		To prove that if $s_i=s_j$ then $S[p_i..p_{i+1}-1]=S[p_j..p_{j+1}-1]$, let us distinguish between the two cases of $s_i$ construction:
		
		\begin{itemize}
			\item
			If $[p_i..p_{i+1}-1]$ is a regular block (of type~\ref{item:regular_block}), then $p_{i+1}-p_i\le \tau\le \delta$ and $s_i=S[p_i..p_i+3\delta-1]$.
			Since $s_i=s_j$, we have $\LCE(p_i,p_j)\ge 3\delta\ge p_{i+1}-p_i+\delta$. 
			Thus, due to the fact that $P$ is forward synchronized, we have that $p_{i+1}-p_i=p_{j+1}-p_j$. 
			It follows that $S[p_i..p_{i+1}-1]$ and $S[p_j..p_{j+1}-1]$ are of the same length, and they are both prefixes of $s_i=s_j$, therefore they must be equal.
			
			\item 
			Otherwise, if $[p_i..p_{i+1}-1]$ is a periocic block (of type~\ref{item:long_blocks}), then it must be that $[p_j..p_{j+1}-1]$ is also a periodic block (of type~\ref{item:long_blocks}) since otherwise we can use the proof of the previous case to prove that $[p_i..p_{i+1}-1]$ is not a periodic block.
			Since $S[p_i..p_i+3\delta-1]$ and $S[p_j..p_j+3\delta-1]$ are prefixes of $s_i=s_j$ respectively, then $S[p_i..p_i+3\delta-1]=S[p_j..p_j+3\delta-1]$. 
			Due to the definition $s_i$ and $s_j$ we have that $s_i[3\delta+2]$ is $\rv(p_i)-p_i$ in absolute value, and the same goes for $p_j$.
			Thus, it follows that $\rv(p_i)-p_i=\rv(p_j)-p_j$.
			Remark that $S[p_i..\rv(p_i)-1]$ and $S[p_j..\rv(p_j)-1]$ are two strings of the same length, and they both share periodic substrings with period lengths at most $\tau\le\delta$. 	
			Hence, due to Lemma~\ref{lem:twiceperiod_equal}, we have that $S[p_i..\rv(p_i)-1]=$ $ S[p_j..\rv(p_j)-1]$.
			In addition, due to the definition of $s_i$ and $s_j$ we have that  $S[\rv(p_i)..\rv(p_i)+2\delta-1] = \linebreak S[\rv(p_j)..\rv(p_j)+2\delta-1]$. 
			Therefore, it must be that $\LCE(p_i,p_j)\ge  \linebreak 2\delta+\rv(p_i)-p_i > \delta+p_{i+1}-p_i$.
			Since $P$ is forward synchronized it follows that $p_{i+1}-p_i=p_{j+1}-p_j$ and thus $S[p_i..p_{i+1}-1]=S[p_j..p_{j+1}-1]$.
		\end{itemize}
	\end{proof}
	
	Hence, to compute the SSA of $S$ with the suffixes of $P$, the algorithm computes the suffix array of $S'_P$ using the algorithm of K{\"{a}}rkk{\"{a}}inen and Sanders~\cite{KS03} in $\O(|S'_P|)=\O(\frac n\tau)$ time. To retrieve the SSA of $P$'s suffixes on the string $S$ the algorithm just need to replace every index $i$ in the SA with the index $p_i$.
	So, in total, the construction of the SSA of $P$ takes $\O(n\cdot\frac\delta\tau)$ time.

	\subsection{Construct \texorpdfstring{$B$}{B}'s \texorpdfstring{$\SST$}{SST} Using The SSA of \texorpdfstring{$P$}{P}}
	Given the SSA of $P$'s suffixes the algorithm builds the SST of $B$'s suffixes, by first computing the SSA of $B$, then computing the LCP array of these suffixes and finally building the SST based on the SSA and LCP array. 
	To compute the SSA of $B$, the algorithm uses the same technique we already used for the construction of $P$'s SSA. 
	For any $i\in B$, let $v_i=S[i..i+3\delta-1]$ be the string of length $3\delta$ from position $i$.
	In addition, similar to the construction of the SSA of $P$, if the central part $S[i+\delta..i+2\delta-1]$ is part of a periodic block we add two characters due to the periodicity violation and the distance between the $\rv$ of the block and $i$, followed by $3\delta$ characters from the right violation. 
	The algorithm sorts all the $v_i$s for $i\in B$ in $\O(n\cdot\frac \delta\tau)$ time and computes the rank of each $v_i$, denoted as $r_i$. In addition, the algorithm computes for each $i\in B$ the index $\suc_P(i+\delta)$ which is the first selected position after index $i+\delta$. Let us denote by $x_i$ the rank of the suffix $S[\suc_P(i+\delta)..n]$ among all the suffixes of $P$ (which can be retrieved from the SSA of $P$). For each $i\in B$ we associate the pair $(r_i,x_i)$.
	The SSA of $B$ is obtained by sorting the list of pairs, due to the following lemma.
	The lemma's proof is omitted since it is almost the same as the proof of Lemma~\ref{lem:SSA_of_blocks}.
	\begin{lemma}
		Let $i,j\in B$ be two given suffixes, then $S[i..n]<S[j..n]$ if and only if $(r_i,x_i)<(r_j,x_j)$ (where $<$ is the lexicographic smaller than relation).
	\end{lemma}
	
	The sorting of all the pairs is done in $\O(n)$ time using radix sort if $\tau\le\sqrt n$ or by any standrad efficient sorting in  $\O(\frac n\tau\log\frac n\tau)=\O(n)$  time if $\tau>\sqrt n$.
	After computing the SSA of $B$, to compute the LCP array, the algorithm uses the \LCE{} data structure from Section~\ref{sec:LCE_construction} with the same partitioning set $P$.
	Since each $\LCE$ query takes $\O(\delta)$ time, and the algorithm computes $|P|-1=\O(\frac n\tau)$ $\LCE$ queries, the total time for this phase is $\O(n\cdot\frac\delta\tau)$.
	With the SSA and the corresponding LCP array, the algorithm computes the SST in $\O(|P|)=\O(\frac n\tau)$ time using the algorithm of Kasai et al.~\cite{KLAAP01}.
	
	Concluding, given a set of suffixes $B$, one can build the SST of $B$ using only $\O(|B|)$ working space, by setting $\tau=\frac n{|B|}$. Using the randomized selection, this method yields a randomized Las-Vegas algorithm with $\O(n)$ running time with high probability for $|B|\ge\Omega(\log^2 n)$, and $\O(n)$ expected running time for $|B|<o(\log^2 n)$ completing the proof of Theorem~\ref{thm:lasvegasSST}.
	Using the deterministic selection this method yields a deterministic algorithm with $\O(n\log\tau)$  running time, for any $|B|\ge\Omega (\log n)$, completing the proof of Theorem~\ref{thm:deterministicSST}.

	\section{Even Better Deterministic \texorpdfstring{$\LCE$}{LCE}}\label{sec:better_deterministic}
	In this section we show how to use the deterministic selection technique from Section~\ref{sec:determinsiticSelection}, combining with concepts from the so-called difference covers~\cite{Maekawa85} technique and the $\SST$ construction from Section~\ref{sec:SST_construction} (Theorem~\ref{thm:deterministicSST}), to improve the $\LCE$ query time from $\O(\tau\log^*n)$  to $\O(\tau\sqrt{\log^*n})$  without any asymptotic penalty in the space usage or  construction time.
	At the very high level, we consider a partitioning set for $\tau'=\frac \tau{\sqrt{\log^*n}}$. By applying the algorithm of Section~\ref{sec:determinsiticSelection}, one can find a $(\O(\tau'),\O(\tau'\cdot\log^*n))=$ $\left(\O\left(\frac \tau{\sqrt{\log^*n}}\right),\O(\tau\cdot\sqrt{\log^*n})\right)$-partitioning set $P$ of size $\O(\frac n{\tau'})=\O\left(\frac {n}\tau \cdot \sqrt{\log^*n}\right)$.
	Using this partitioning set, one can build a data structure that answers $\LCE$ queries in $\O(\tau\cdot \sqrt{\log^*n})$ time, using the construction of Section~\ref{sec:LCE_construction}. However, the size of this data structure will be $\Theta\left(\frac {n}\tau \cdot \sqrt{\log^*n}\right)$ which is too large in our settings. 
	Therefore, our goal is to use a special subset of $P$ that has only $\O( 1/{\sqrt{\log^*n}})$ fraction of $P$'s positions, and with these positions, the algorithm still will be able to answer $\LCE$ queries fast enough.

	Formally, we consider $P$ as the set of positions induced by the $\mu=\log_{3/2}\left(\frac \tau{\sqrt{\log^*n}} \right)$ level of the hierarchical decomposition.
	As in previous sections we postpone the discussion about periodic large blocks to Appendix~\ref{app:better_deterministic}, and focus here on the case where all the blocks in the partitioning set are regular (of type~\ref{item:regular_block} in Definition~\ref{def:landmarks_set}), so their length is at most $\tau'$. 
	In addition, we assume for now that $\tau\ge \sqrt{\log^* n}$ (and therefore $\tau'\ge 1$), and the discussion on the case where $\tau<\sqrt{\log^*n}$ also appears in Appendix~\ref{app:better_deterministic}.

	\paragraph{Combine $\SST$ to simplify the $\LCE$ computation.} Recall that the computation of $\LCE(i,j)=\ell$ query, in the data structure of Section~\ref{sec:LCE_construction} was composed of three phases. In the first phase the algorithm reads $\O(\delta)$ characters, then the algorithm finds a length $k$ such both $i+k$ and $j+k$ are positions in the partitioning set. The second phase treats a large common substring from $i+k$ and $j+k$ until at most $\delta$ characters prior to $i+\ell$ and $j+\ell$ which are the end of the common substring. Notice that, using the $\SST$ construction from Section~\ref{sec:SST_construction} on the positions of the partitioning set, with LCA support, after finding the common $k$ such that $i+k,j+k\in P$ and $S[i..i+k]=S[j..j+k]$, one can compute $\LCE(i,j)$ by adding $k$ to the common extension of $i+k$ and $j+k$. Therefore, one can merge the second and third phases to be done in constant time.
	In our construction we will use this idea, and focus on finding an offset $k$ such that $S[i..i+k]=S[j..j+k]$ and both $i+k$ and $j+k$ are in the set of positions.

	\paragraph{The filtered positions.} 
	To reduce the number of positions from the set $P$ the algorithm uses the difference covers technique. The following lemma states the properties that we will use.
	
	\begin{lemma}[{\cite{Maekawa85}}]\label{lem:diff-cover}
		Let $m$ and $t$ be integers such that $t<m$.
		There exists a set $DC\subseteq [m]$ of size $|DC|=\O\left(\frac m{\sqrt t}\right)$ such that, for every $i,j\in [m-t]$ there exists $k\le t$ such that both $i+k\in DC$ and $j+k\in DC$.
	\end{lemma}
	We will use this technique to reduce the size of $P$, with $m$ as the size $|P|$ and $t=\log^*n$. 
	We begin with the set $P$,  which is a $(\tau',\O(\tau'\log^*n))$-partitioning set of size $\O(\frac n{\tau'})$. Let $p_1,p_2\dots,p_h$ be the positions in $P$ sorted in ascending order. We define the set $Q$ as the set that contains all the position $p_i$ such that $i\mod \sqrt{\log^*n}=0$ or $i \modulo \log^*n<\sqrt{\log^*n}$. It is straightforward  that $|Q|=\left(\frac {|P|}{\sqrt{\log^*n}}\right)$.
	Since $|P|=\O(\frac n\tau \sqrt{\log^*n})$ we have that $|Q|=\O(\frac n\tau)$. Our data structure will build the $\SST$ of the set $Q$, using the deterministic algorithm from Section~\ref{sec:SST_construction}. Notice that, one can decide whether a position from $P$ is also in $Q$ without actually store $P$, since the decision depends only on the \emph{rank} of the position in $P$. Hence, the algorithm  computes the set $Q$ with only $\O(\frac n\tau)$ words of working space (see Appendix~\ref{app:deterministic_in_space}).

	\paragraph{The query computation.}
	To compute $\LCE(i,j)$ query the algorithm simply reads the suffixes $S[i..n]$ and $S[j..n]$ simultaneously where for each offset $k=0,1,\dots$, when comparing $S[i+k]$ and $S[j+k]$ the algorithm checks whether both $i+k\in Q$ and $j+k\in Q$. This process is stopped when one of the following occurs. If a mismatch is reached in offset $k$ ($S[i+k]\ne S[j+k]$), then $k$ is reported as $\LCE(i,j)$. Otherwise, if the algorithm reaches a common offset in $Q$ ($i+k,j+k\in Q$), then $k+\LCE(i+k,j+k)$ is reported as $\LCE(i,j)$. 
	Notice that after spending $\O(\tau)$ time, the algorithm can check in additional constant time whether $i+k\in Q$, or $j+k\in Q$, using the linked list of positions and the auxiliary array, as described for the successor computation.
	In the following lemma we prove that this process will never take more than $\O(\tau\cdot\sqrt{\log^*n})$ time.
	
	\begin{lemma}\label{lem:better_deterministic_non_periodic}
		There exists a constant $c$ such that 
		for any $i,j\in [n]$ with $\LCE(i,j)=\ell>c\tau\sqrt{\log^*n}$ there exists $k<c\tau\sqrt{\log^*n}$ such that both $i+k\in Q$ and $j+k\in Q$.
	\end{lemma} 
	\begin{proof}
		Let $c'>1$ be a constant such that $P$ is a $(\tau',c'\tau'\log^*n)$-partitioning set (of size $\O(\frac n{\tau'})$), 
		in the terms of Section~\ref{sec:LCE_construction} we have $\delta=c'\tau'\log^*n= c'\tau\sqrt{\log^*n}$.
		We  will prove the lemma holds for $c=4c'$. 
		Let $i,j\in[n]$ be positions such that $\LCE(i,j)>c\tau\sqrt{\log^*n}$.
		By Lemma~\ref{lem:center_blocks}, we have:
		$$\{p-i\,|\,p\in(P\cap[i+\delta..i+\ell-\delta-1])\}=\{p-j\,|\,p\in(P\cap[j+\delta..j+\ell-\delta-1])\}.$$ 
		In particular, since $\ell>4\delta$ we have that:
		$$\Delta=\{p-i\,|\,p\in(P\cap[i+\delta..i+3\delta-1])\}=\{p-j\,|\,p\in(P\cap[j+\delta..j+3\delta-1])\}.$$
		Notice that  $|[i+\delta..i+3\delta-1]|=|[j+\delta..j+3\delta-1]|=$ $2\delta=2c'\tau'\log^*n$. Hence, since we have assumed that the length of each block induced by $P$ is at most $\tau'$, we have that the number of elements (offsets) in $\Delta$ is at least $\frac {2\delta}{\tau'}=2c'\log^*n>2\log^*n$.
		For every element $k\in \Delta$, we have $i+k\in P$ and $j+k\in P$.
		Let $p_\alpha,p_{\alpha+1},\dots,p_{\alpha+2\log^*n-1}$ be the first $2\log^*n$ elements in $\{p-i\,|\,p\in(P\cap[i+\delta..i+\ell-\delta-1])\}$.
		There must be an index $\alpha\le\alpha'\le \alpha+\log^*n$  such that $\alpha'\modulo\log^*n=0$. Therefore, all the positions $p_{\alpha'},p_{\alpha'+1},\dots, p_{\alpha'+\sqrt{\log^*n}-1}$ are positions in $Q$. 
		Let $\beta'$ be the index corresponding to $\alpha'$ such that $p_{\alpha'}-i=p_{\beta'}-j$. We consider the positions $p_{\beta'},p_{\beta'+1},\dots, p_{\beta'+\sqrt{\log^*n}-1}$. Since these are $\sqrt{\log^*n}$ consecutive elements in $P$, it must be that at least one of them is taken into $Q$ (since its index in $P$ is an integer multiple of $\sqrt{\log^*n}$). Since for any $0\le x< \sqrt{\log^*n}$ we have $p_{\alpha'+x}-i=p_{\beta'+x}-j$ it must be that there exists at least one  $x$ such that both $p_{\alpha'+x}\in Q$ and $p_{\beta'+x}\in Q$ and therefore for $k=p_{\alpha'+x}-i$ we have $i+k\in Q$ and $j+k\in Q$, and $k=p_{\alpha'+x}-i<3\delta\le c\tau\sqrt{\log^*n}$.
	\end{proof}
	
	\paragraph{Complexities.} The size of the set $Q$ is  $\O(\frac n\tau)$, and using the $\SST$ --- the data structure answers  $\LCE$ queries in $\O(\tau\sqrt{\log^*n})$ time. The construction time of the data structure is $\O(n\log\tau)$ for finding $P$ and $Q$, and for the construction of $Q$'s  $\SST$ due to the construction of Section~\ref{sec:SST_construction}. Thus, we have proved Theorem~\ref{thm:deterministicLCE2}.

	\section{Ensuring Randomized Selection in Linear Time With High Probability}\label{sec:ens_rand_whp}
	In Section~\ref{sec:randomSelection}, we introduced a randomized algorithm that selects a small $(2\tau,2\tau)$-partitioning set  in $\O(n)$ expected time.
	The algorithm uses a fingerprint hash function $\phi$ and an approximately-min-wise hash function $h$ to compute an \ID{} for each position, and selects positions with minimum values in a range of size $\tau$ for the partitioning set.
	Remark that only the partitioning set's size is in question \textemdash the expected size is $\O(\frac n\tau)$ but the probability that the set size is indeed $\O(\frac n\tau)$ is not guaranteed to be high.
	If we ensure with high probability that a pair of hash functions $\phi$ and $h$ can be found fast such that $\ID_{\phi,h}$ will induce small partitioning set, then the running time of the algorithm is linear with high probability.
	
	In this section, like in Section~\ref{sec:randomSelection}, we assume that $S$ does not contain long periodic substrings with short periods. 
	With very small modifications, as described in Appendix~\ref{app:rand_selection_periodic_blocks}, the methods that appear next could work in the general case by ``skipping'' all long periodic substrings while running the algorithms.
	
	We argue that a good selection of $(\phi,h)$ pair can be determined in $\O(n)$ time \whp{} while not exceeding $\O(\frac{n}{\tau}$) words of working space. 
	More precisely, since $\phi$ may render ``bad" results only when it has  collisions on substrings of $S$, and this event occurs with inverse polynomial probability, we ignore this case (this event is absorbed in the high probability of the algorithm). 
	Thus, we focus on a fast method to pick $h\in\mathcal{F}$ that will make the average amount of partitioning positions per interval of size $\tau$ to be constant \whp{}.
	
	There are two distinct cases: either $\tau < \log^2 n$ or $\tau \ge \log^2 n$, and the algorithm treats each one of them separately.
	In the first case since $\tau$ is small we can use a sampling of not too many intervals and count the number of selected positions from $P$ in each interval, and use that to approximate the total size of $P$ \whp{}.
	In the second case, we run the algorithm on $\O(\log n)$ different hash functions, which will ensure that with high probability at least one of those hash functions induces a small enough partitioning set.
	We pick the partitioning set as a subset of a larger $(\log n,\log n)$-partitioning set, instead of picking candidates from every position, to allow the parallel process to run in a total of $\O(n)$ time.

	\subsection{Case 1: \texorpdfstring{$\tau \le \log^2 n$}{Small t } }
	
	Fix a fingerprint function $\phi$ with no collisions on $S$.
	We will use sampling methods in order to evaluate the amount of partitioning positions that a given approximate min-wise independent hash function $h$ renders when using $\ID_{\phi,h}$ as the \ID{} function.
	
	For any $0\leq i \leq \frac{n-2}{\tau}$ denote $P_i = P \cap [i\tau +1..(i+1)\tau]$ and denote $C_i = |P_i|$. By using the algorithm described in the proof of Lemma~\ref{lemma:id_alg} on the substring $S[i\tau+1..(i+1)\tau]$, $C_i$ can be calculated in $\O(\tau+|P_i|\tau)$ time using $\O(1)$ words of space (by maintaining only the rightmost partitioning position picked so far and only counting the amount of selected positions). For every $\xi>0$ we can pick $m=\frac{\xi}{2}\log^5 n$ intervals $[i_k\tau +1..(i_k+1)\tau]$ for $1\le k\le m$ uniformly and independently of the other chosen intervals. Denote $\overline{C} = \frac{1}{m}\sum_{k=1}^m C_{i_k}$.
	
	\begin{observation}\label{obs:sample_whp}
		With high probability, $\overline{C}\cdot\frac{n}{\tau} -\frac{n}{\tau} <|P| <\overline{C}\cdot\frac{n}{\tau} +\frac{n}{\tau} $.
	\end{observation}
	\begin{proof}
		Remark that $\E[C_{i_k}] = |P|\cdot\frac{\tau}{n}$ and therefore also $\E[\overline{C}]= |P|\cdot\frac{\tau}{n}$. By Hoeffding's inequality, since $0 \leq C_{i_k}\le \tau$ and all $C_{i_k}$ are independent then:
		$$ \Pr\left[|\overline{C} - \E[\overline{C}]| \ge 1\right] \leq 2\exp\left(-\frac{2m^2}{m\tau^2}\right)  \leq 2\exp\left(-\frac{\xi\log^5 n }{\log^4 n}\right) = \O\left(\frac{1}{n^\xi}\right)$$
		Thus, by sampling $m$ intervals we can calculate $\overline{C}$ and then evaluate \whp{} that $\overline{C}\cdot\frac{n}{\tau} -\frac{n}{\tau} <|P| <\overline{C}\cdot\frac{n}{\tau} +\frac{n}{\tau} $.	
	\end{proof}
	
	\begin{observation}\label{obs:sample_markov}
		There exists a constant $c>0$, such that $\Pr_{\phi,h}\left(\overline{C} \leq c\right) > \frac{2}{3}$.
	\end{observation}
	\begin{proof}
		By Lemma~\ref{lem:space_exp}, we know that $\E_{\phi,h}[|P|] =\O(\frac{n}{\tau})$ and thus by Markov's inequality for some $c'>0$: $\Pr_{\phi,h}\left(|P| \leq \frac{c'n}{\tau}\right) > \frac{2}{3}$. But $|P| \leq \frac{c'n}{\tau} \Rightarrow \overline{C} \leq c'+1$ \whp{}. Therefore, $\Pr_{\phi,h}\left(\overline{C} \leq c'+1\right) > \Pr_{\phi,h}\left(|P| \leq \frac{c'n}{\tau}\right) > \frac{2}{3}$.
	\end{proof}
	
	Calculating $\overline{C}$ takes at most $\O(m\tau^2) = \O(\log^9 n)$ time since for each sampled interval there could be as much as $\tau$ selected partitioning positions, and it may cost up to $\O(\tau)$ time to select each of them in the worst case. 
	With high probability, after at most $\xi\log_{3} n$ repetitions of the algorithm for different approximately min-wise hash functions, $\overline{C}$ will be small enough, due to Observation~\ref{obs:sample_markov}. 
	The partitioning set induced by the selected functions must be calculated and verified as small enough. 
	This will happen \whp{} as well due to Observation~\ref{obs:sample_whp}. 
	Thus, \whp{} after $\O(\polylog(n)+n)$ time a small partitioning set will be found. 
	If the run will fail, the failure will be detected within $\O(n)$ time. Therefore, the total running time is $\O(n)$ time \whp{}.
	
	\subsection{Case 2: \texorpdfstring{$\tau > \log^2 n$}{Large t}}
	
	The algorithm for this case is composed of two steps: In the first step the algorithm finds $\phi,h_0$ such that the size of the $(\log n,\log n)$-partitioning set induced by $\ID_{\phi,h_0}$, denoted $P_0$, is at most $\O\left(\frac{n}{\log n}\right)$. 
	In the second step the algorithm draws $\log n$ approximately min-wise independent hash functions $h_1,h_2\dots h_{\xi\log_3 n}\in\mathcal{F}$ in random, and calculates for each $1\le i \le\log n$ a $(2\tau,2\tau)$-partitioning set using $\ID_{\phi,h_i}$ that will be a subset of $P_0$.
	
	For each $h_i$ and in each interval $[j..j+\tau-1]$ the algorithm selects the positions with the smallest ID value from $P_0\cap [j..j+\tau-1]$.
	Hence, 
	$$P_i = \left\{ k\in P_0 \mid \exists j\in[k-\tau+1..k]:\ID_{\phi,h_i}(k) = \min_{\ell\in[j..j+\tau-1]\cap P_0}\left\{ \ID_{\phi,h_i}(\ell)\right\} \right\}$$.
	
	We prove that $P_i$ is $(2\tau,2\tau)$-partitioning set and that the expected size of $P_i$ is $\O(\frac{n}{\tau})$.
	\begin{lemma}
		$P_i$ is a $(2\tau,2\tau)$-partitioning set.
	\end{lemma}
	\begin{proof}
		We prove that $P_i$ has the properties of Definition~\ref{def:landmarks_set}:
		
		\subsubsection*{Local Consistency} 
		Let $j,k\in[n-2\tau]$ be two indices such that $S[j-2\tau..j+2\tau]=S[k-2\tau..k+2\tau]$. 
		Since $\tau > \log^2 n \ge \log n $ then for all $\ell\in[-\tau..\tau]$ it follows that $S[j+\ell-\log n..j+\ell + \log n] = S[k+\ell-\log n..k	+\ell + \log n]$. 
		Remark that since $P_0$ is a $(\log n,\log n)$-partitioning set then $j+\ell\in P_0 \Leftrightarrow k+\ell\in P_0$. 
		Notice that because $S[j-2\tau..j+2\tau]=S[k-2\tau..k+2\tau]$ then it also holds that for all $\ell\in[-\tau..\tau]$ : $\ID_{\phi,h_i}(j+\ell) = \ID_{\phi,h_i}(k+\ell)$. 
		Since $j\in P_i$ if and only if $j\in P_0$ and it has the minimal ID in some interval, and since the same holds for $k$ then from all the above  it follows that $j\in P_i \Leftrightarrow k\in P_i$.
		
		\subsubsection*{Compactness} 
		Remark that since  $\tau > \log^2 n \ge \log n $ and since $P_0$ is a $(\log n,\log n)$-partitioning set then in any interval of size $\tau$ there is a position from $P_0$.
		Therefore, the minimal ID on positions from $P_0$ within an interval of size $\tau$ is well-defined.
		Hence, the proof of Property~\ref{item:blocks_size} from Lemma~\ref{lemma:id_is_partition_set} holds the same for $P_i$.
	\end{proof}
	
	\begin{lemma}\label{lemma:sampling_space_exp}
		$\E\left[\,|P_i|\,\right] = \O(\frac{n}{\tau})$
	\end{lemma}
	\begin{proof}
		
		The proof is based on the proof of Lemma~\ref{lem:space_exp}, with one change: positions in $P_i$ are picked from $P_0$ which is a ${(\log n,\log n)}$-partitioning set, and not all the positions in the interval are candidates. 
		Each interval of size $\frac \tau 2$ will contain at least $\frac{\tau}{2\log n}$ candidates to $P_i$ from $P_0$. 
		As in the proof of Lemma~\ref{lem:space_exp}, this implies that $\Pr\left(j\in P_i\right) = \O(\frac{\log n}{\tau})$.
		Hence, by linearity of expectation over all $\O(\frac{n}{\log n})$ positions in $P_0$: $\E\left[\,|P_i|\,\right] = \O(\frac{n}{\tau})$.	
	\end{proof}
	
	From the above lemma and by Markov's inequality it follows that there exists some constant $c'>0$ such that $\Pr(|P_i| < \frac{c'n}{\tau}) > \frac{2}{3}$. 
	Therefore, \whp{} at least one of the approximate min-wise independent hash functions in $\{h_i\}_{i=1}^{\xi\log_3 n}$ renders a partitioning set $P_i$ such that $|P_i| = \O(\frac{n}{\tau})$.
	
	In order to pick a hash function such that $|P_i| < \frac{c'n}{\tau}$ we calculate $|P_i|$ for all $h_i$, and \whp{} the minimal value will be below $\frac{c'n}{\tau}$. A naive approach would be to use a modification of the algorithm from the proof of Lemma~\ref{lemma:id_alg}, but on all the functions $h_i$ simultaneously. When moving forward simultaneously, this will cost:
	\begin{enumerate}
		\item $\O(n)$ time in total for calculating positions of $P_0$ throughout the run.
		\item $\O(n)$ time for the fingerprint calculation.
		\item $\O(\frac{n}{\log n})$ hash function evaluations on the fingerprints of positions in $P_0$, for each hash function summing up to $\O(n)$ time, and an extra $\O(\log n)$ words of space for the positions counter and the rightmost picked position of all $P_i$.
		\item $\O(\tau)$ for each step-back of the rolling interval algorithm.
	\end{enumerate}
	
	All those operations sum up to $\O(n)$ total time and $\O(\frac n \tau + \log n)$ words of space, plus the amount of time consumed on step-backs for all $h_i$. In the algorithm described in the proof of Lemma~\ref{lemma:id_alg}, after each selected position a step-back may occur. This sums up to $\O(n\log n)$ time consumed for step-backs in total, and thus bounds the total work time asymptotically.
	
	In Appendix~\ref{app:fast_random_parallel_count}, we present a method to calculate $|P_i|$ or determining that it is bigger than $\frac{c'n}{\tau}$ for all $h_i$ using $\O(n)$ time and $\O(\frac{n}{\tau}+\log n)$ words of space. The main change is the way we handle step-backs. This improvement uses $\O(\frac{n}{\tau\log n})$ words of space for each hash function $h_i$ to maintain ``future" $P_i$ candidates. By doing so, we reduce the total amount of cases where we have to step back and recalculate all the \ID{}s to at most $\O(\log n)$ step-backs per hash function in the worst case. Thus, the amount of time spent on step-backs will sum up to $\O(\tau\log^2 n)$ time, which is $\O(n)$ for $\tau \leq \frac{n}{\log^2 n}$.
	
	Therefore, the following lemmas will apply immediately:
	\begin{lemma}
		There exists an algorithm that for any $i=1,2,\dots,\xi\log_3 n$, computes the sizes $|P_i|$ or determines that  $|P_i| > \frac{c'n}{\tau}$. The algorithm takes  $\O(n+\tau\log^2 n)$ time using $\O(\frac{n}{\tau}+\log n)$ words of space.
		
	\end{lemma}
	\begin{lemma}\label{lem:random_selection_whp}
		For $\log^2 n < \tau  < \O(\frac{n}{\log^2 n})$, there exists an algorithm that computes a  $(2\tau,2\tau)$-partitioning set of size $\O(\frac{n}{\tau})$ in $\O(n)$ time \whp{} using $\O(\frac{n}{\tau})$ words of space.
	\end{lemma}
	\begin{proof}
		After computing the size $|P_i|$ or determining that it is too big for all $h_i$ in linear time, \whp{} there exists at least one hash function $h_k$ such that $|P_k| < \frac{c'n}{\tau}$. Then, the algorithm calculates $P_k$ explicitly in $\O(n)$ time.
	\end{proof}
	
	\subsection{Summary}
	We have proved that for all $1 < \tau < \O(\frac{n}{\log^2 n})$ there exists an algorithm that computes a $(2\tau,2\tau)$-partitioning set of size $\O(\frac n\tau)$ using $\O(\frac{n}{\tau})$ words of working space that runs in linear time \whp{}.

	\section{Overview on Methods for Long Periodic Blocks}\label{sec:periodic_overview}
	In this section we survey the main idea that is used to treat periodic blocks, which are blocks of type~\ref{item:long_blocks} due to Definition~\ref{def:landmarks_set}. The complete details regarding each section appear in the corresponding appendix. 
	
	We clarify the difficulty caused by long periodic blocks in the partitioning of $S$.
	We will focus now on the $\LCE$ computation in such a case; however, the difficulty is similar on the deterministic construction of Section~\ref{sec:better_deterministic}.
	When comparing two suffixes, $S[i..n]$ and $S[j..n]$ that share a large common prefix, of length larger than $2\delta$,  we are guaranteed that their partitioning into blocks is identical except for at most $\delta$ positions at the margins of the common substrings. 
	This guarantee is based on the locality property of the partitioning set (Property~\ref{item:blocks_sync} in Definition~\ref{def:landmarks_set}), and is  formulated in Lemma~\ref{lem:center_blocks}. 
	When we consider the case of only short blocks, we exploit this property to find selected positions with the same distance from $i$ and $j$ efficiently. Since the length of each block is at most $\tau$ in this case, such positions must exist in an offset of at most  $\delta+\tau$ positions following $i$ and $j$.
	These corresponding positions are very useful in order to skip a large common substring, which is composed of full blocks and should end at most $\delta+\tau$ characters prior to the first mismatch.
	
	However, in the general case with long periodic blocks, it is possible that the first selected positions with the same offset from $i$ and $j$ will be in $\omega(\delta)$ positions after $i$ and $j$. In such a case we cannot compare all the characters until the synchronized positions since it will take too much time.
	Thus, we utilize the fact that long blocks are corresponding to periodic substrings of $S$.  
	Given two blocks of type~\ref{item:long_blocks} that share a common substring of length at least $2\tau$, we do not need to compare the rest of the characters in order to verify that the corresponding characters match each other. This is true because the common substrings guarantee that both blocks have the same period, and therefore all the corresponding characters match each other. 
	The claim is formalized in Lemma~\ref{lem:twiceperiod_equal}.

	Hence, after the algorithm verifies that two substrings of length $\delta+2\tau$ are the same, while these substrings are fully contained in long blocks  the algorithm skips all the pairs of corresponding characters in the two blocks. 
	If the two blocks end in exactly the same offset from $i$ and $j$ then the beginning of the following blocks are correlated positions which we use for skipping over the large middle part of the common substrings.
	Otherwise, if the two blocks end in different offsets from $i$ and $j$, then due to the locality of the partitioning set, it must be that the first mismatch between the suffixes is at most $\delta$ positions after the last synchronized offset inside the long blocks. Then, in such a case we can compare at most additional $\delta$ characters until we find the first mismatch.

	\appendix

	\part*{Appendix}
	\section{Comparison between \texorpdfstring{$\LCE$}{LCE} algorithms}\label{app:table_compare}
	
	The following table introduces the previous most efficient data structures for the $\LCE$ problem together with our new results. The integer $\ell$ denoted the $\LCE$ query result.

	\begin{table}[h!]
		
		\centerfloat
		\begin{minipage}{\textwidth}
			\centerfloat
			\begin{tabular}{b{18pt}|c|c|c|c|c|c|}
				\cline{2-7} 
				&\multicolumn{4}{c|}{Data Structure} &  \multirow{2}{*}{Trade-off range} & \multirow{2}{*}{Reference}\\
				\cline{2-5} 
				&Space & Query time & Preprocessing time& Correct & &\\
				\cline{2-7} 
				\cmidrule[0pt]{2-7}
				\cline{2-7} 

				\begin{minipage}{18pt}\tiny{Monte Carlo} \end{minipage} 
				&$\frac n \tau$&$\tau $&$n$&\whpshort{}& $1\le\tau\le n$&\cite[Theorem 3.3]{GK17}\\
				\cline{2-7} 
				
				\cmidrule[0pt]{2-7}
				\cline{2-7} 
				\multirow{5}{5pt}{\rotatebox{90}{\parbox{1.5cm}{\centering Las Vegas}}}
				&$\frac n \tau$&$\tau \log \frac \ell \tau$ & $n(\tau+\log n)$ \whpshort{} &always& $1\le\tau\le n$&\cite{BGSV14}\\
				&$\frac n \tau$& $\tau$ & $n^{3/2}$ \whpshort{} &always& $1\le\tau\le n$&\cite[Section 3.5]{BGKLV15}\\
				&$\frac n \tau$& $\tau$ & $n\sqrt{\log \frac{n}{\tau}}$ \whpshort{} &always&$1\le\tau\le n$ &\cite{GK17,BGKLV15}\footnote{The $\LCE$ data structure is not introduced explicitly, but it can be constructed with the techniques of~\cite{BGKLV15} and~\cite{GK17}.}\\
				&$\frac n \tau$& $\tau$ & $n$  \textit{expected} &always& $1\le \tau\le n$ &\textbf{new}\\
				&$\frac n \tau$& $\tau$ & $n$ \whpshort{} &always& $1\le \tau\le \frac{n}{\log^2 n}$ &\textbf{new}\\
				\cline{2-7} 
				\cmidrule[0pt]{2-7}
				\cline{2-7} 
				
				\multicolumn{1}{r|}{\multirow{5}{5pt}{\rotatebox{90}{\small{Deterministic}}}}&$1$&$\ell$&$1$&always& - &\naive algorithm\\
				&$n$&$1$&$n$&always& - &Suffix tree+LCA\\				
				&$\frac n \tau$&$\tau$&$n^{2+\varepsilon}$&always& $1 \le \tau \le n$ &\cite[Section 4]{BGKLV15}\\
				
				&$\frac n \tau$& $\tau\min\{\log\tau,\log\frac n\tau\}$ & $n\tau$ &always& $1 \le \tau \le n$ &\cite[Corollary 12]{TBIPT16}\\
				
				&$\frac n \tau$& $\tau\sqrt{\log^*n}$& $n\log \tau$ &always& $1\le\tau\le \frac n{\log n}$ &\textbf{new}\\

				\cline{2-7}

			\end{tabular}
		\end{minipage}
		
		\caption{LCE data structures}
		\label{tbl:comparison}
	\end{table}

	\section{Missing Details for the \texorpdfstring{$\LCE$}{LCE} construction}\label{app:LCE_construction_details}
	
	In this appendix we complete the missing details from Section~\ref{sec:LCE_construction} about the construction of the $\LCE$ data structure, based on a $(\tau,\delta)$-partitioning set of size $\O(\frac n\tau)$. Given such a set we construct an $\LCE$ data structure in $\O(n)$ time, using $\O(\frac n\tau)$ words of space, such that any $\LCE$ query can be answered in $\O(\delta)$ time.
	
	\subsection{Missing proof}\label{app:LCE_missing_proof}
	
	\begin{proof}[{Proof of Lemma~\ref{lem:center_blocks}}]
		Let $\alpha\in \{p-i\,|\,p\in(P\cap[i+\delta..i+\ell-\delta-1])\}$. By definition, $\alpha=p-i$ such that $i+\alpha=p\in P$ and $\delta\le\alpha\le\ell-\delta-1$.
		Since $\LCE(i,j)=\ell$, we have that $S[i..i+\ell-1]=S[j..j+\ell-1]$ and in particular $S[i+\alpha-\delta..i+\alpha+\delta]=S[j+\alpha-\delta..j+\alpha+\delta-1]$. Therefore, by property~\ref{item:blocks_sync} of Definition~\ref{def:landmarks_set}, we have that $p=(i+\alpha)\in P\Rightarrow (j+\alpha)\in P$. Thus, $(j+\alpha)-j=\alpha\in \{p-j\,|\,p\in(P\cap[j+\delta..j+\ell-\delta-1])\}$. 
		The opposite direction is symmetric.
	\end{proof}

	\subsection{LCE with periodic blocks}\label{app:LCE_with_periodic_blocks}

	We will follow the discussion above in Section~\ref{sec:periodic_overview}. 
	Recall that the computation of $\LCE(i,j)=\ell$ query in the non-periodic case is composed of three phases. The first phase compares $3\delta$ characters from $i$ and $j$, the last phase reads at most $\delta+\tau$ characters prior to positions $i+\ell$ and $j+\ell$ and the second phase is responsible to all the characters in between.
	In the general case, each one of these phases might be changed due to an occurrence of a periodic block at the end of the area treated by this phase. In such a case, we cannot read the whole block, since its length might be $\omega(\delta)$. 
	However, we can use properties of periodic strings to learn about the substrings without actually read them in the query time.
	We will use  Lemma~\ref{lem:twiceperiod_equal}, and deduce for our algorithm that after reading $2\tau$ characters of two periodic blocks, the content of the remaining suffix of the blocks can be deduced from the substring the algorithm already read.
	Hence,  after reading $2\tau$ characters without any mismatch, the algorithm can skip the rest of the block until the end of the first block.

	For the first phase, if the algorithm reads $3\delta$ characters from position $i$ and $j$ without any mismatch, the algorithm finds the first selected position following $i+\delta$ and the first position following $j+\delta$ which are $\suc_P(i+\delta)$ and $\suc_P(j+\delta)$. Let us denote $\alpha_i=\suc_P(i+\delta)-i$ as the offset of the block's endpoint from $i$ and similarly we denote $\alpha_j=\suc_P(j+\delta)-j$.
	If $\alpha_i=\alpha_j$ i.e., if the positions are $\suc_P(i+\delta)=i+k$ and $\suc_P(j+\delta)=j+k$ for some $k\in\mathbb N$, we can move to the second phase of the algorithm using the blocks beginning in positions $\suc_P(i+\delta)$ and $\suc_P(j+\delta)$. 
	Otherwise, let us also define $\alpha=\min\{\alpha_i,\alpha_j\}$, the minimal of these offsets. Then $\LCE(i,j)\le \alpha+\delta$ as stated in the following lemma.
	\begin{lemma}\label{lem:end_of_periodic_block_close_to_LCE}
		If $S[i..i+3\delta]=S[j..j+3\delta]$ and $\alpha_i\ne  \alpha_j$ then $\LCE(i,j)\le \alpha +\delta$.
	\end{lemma}
	\begin{proof}
		Assume by contradiction that $\LCE(i,j)> \alpha +\delta$.
		Since $\delta\le\alpha <\LCE(i,j)-\delta$, we have that $S[i+\alpha-\delta..i+\alpha+\delta]=S[j+\alpha-\delta..j+\alpha+\delta]$. Thus, due to property~\ref{item:blocks_sync} of Definition~\ref{def:landmarks_set} we have $i+\alpha\in P\Leftrightarrow j+\alpha\in P$, but in this case $\alpha_i=\alpha_j$ which contradicts lemma assumption.
	\end{proof}
	
	Hence, in this case, the algorithm proceeds directly to compare the characters from $S[i+\alpha]$ and $S[j+\alpha]$. In total, the computation time for this case is $\O(\delta)$.

	At the transition from the second phase to the third phase, we have assumed in the non-periodic case that the end of the common block is at most $\tau+\delta$  characters from the end of the common substring. However, this assumption is not guaranteed anymore, since it is possible that positions $i+\ell-\delta$ and $j+\ell-\delta$ are now in a large periodic block. Thus, if the last common block ended at positions $i+k$ and $j+k$, then after comparing $\delta+\tau$ characters if the algorithm does not find any mismatch, it must be that the blocks beginning at positions $i+k$ and $j+k$ are periodic blocks (of type~\ref{item:long_blocks} from Definition~\ref{def:landmarks_set}). Then, we use the same technique as in the first phase. Let us define $\beta_j=\suc_P(i+k+1)-i$ and $\beta_j=\suc_P(j+k+1)-j$, and let $\beta=\min\{\beta_i,\beta_j\}$. Then, due to Lemma~\ref{lem:twiceperiod_equal},  $S[i..i+\beta]=S[j..j+\beta]$. Hence, it must be the case that $\beta_i\ne\beta_j$ since otherwise the second phase of the algorithm should return $k\ge \beta$.
	Thus, similar to Lemma~\ref{lem:end_of_periodic_block_close_to_LCE}, it must be that $\LCE(i,j)\le\beta+\delta$. Hence, the algorithm compares at most $\delta$ pairs of characters, beginning with $S[i+\beta]$ and $S[j+\beta]$ until a mismatch is found.

	\section{Missing Details for the  Randomized Selection}\label{app:randomSelection}
	In this appendix we complete the algorithm for the randomized selection of $(2\tau,2\tau)$-partitioning set. We begin with the missing proofs from Section~\ref{sec:randomSelection} in Appendix~\ref{app:rand_selection_missing_proofs}. Then, in Appendix~\ref{app:rand_selection_periodic_blocks} we complete the missing details for the general case where there are large periodic substrings in $S$ with small period. Together with the techniques of Section~\ref{sec:ens_rand_whp}, these appendices complete the randomized algorithm for a small $(2\tau,2\tau)$-partitioning set selection in $\O(n)$ time \whp{}.

	\subsection{Missing proofs for Section~\ref{sec:randomSelection}}\label{app:rand_selection_missing_proofs}

	\begin{proof}[{Proof of Lemma~\ref{lem:space_exp}}]
		Since $|P|\leq n$ and the probability that any collision occurred when using a random $\phi$ can be bounded by $\frac{1}{n^2}$ then by denoting $B_\phi$ as the event that $S$ has a collision using $\phi$ we can bound:
		\begin{equation*}
		\E_{\phi,h}\left[|P|\right]
		= \E_{\phi,h}\left[|P|\mid \overline{B_\phi}\right]\cdot\Pr\left[\overline{B_\phi}\right] + \E_{\phi,h}\left[|P|\mid B_\phi\right]\cdot\Pr\left[B_\phi\right]
		\leq \E_{\phi,h}\left[\,|P| \mid \overline{B_\phi}\,\right] + 1
		\end{equation*}
		Therefore, in order to prove the lemma it is sufficient to prove that $\E_{\phi,h}\left[\,|P| \mid \overline{B_\phi}\,\right] =\O(\frac{n}{\tau})$, i.e., bound the expected size of $P$ when assuming that there are no collisions.

		Assume that $\phi$ has no collisions over substrings of $S$.
		For any $i\in[n]$, let $ A_i $ be the indicator random variable of the event that $ i\in P $. Thus, $ |P| = \sum_{i=1}^{n}A_i $. We will show that $ \Pr_{\phi,h}\left[A_i=1\right] = \O\left(\frac{1}{\tau}\right) $ for all $i \in[n]$, and by linearity of expectations the lemma follows.
		
		By definition, if $ i \in P $ then there exists an index $j$ such that $ i\in [j..j+\tau-1] $ and also $ h(\phi_i) = \min_{k\in[j..j+\tau-1]}\left\{h(\phi_k)\right\} $. Because $ i\in[j..j+\tau-1] $ then either $ |[j..i]| \ge \frac{\tau}{2} $ or $|[i..j+\tau-1]| \ge \frac{\tau}{2} $. This fact together with the fact that $ h(\phi_i) $ is minimal in some interval of size $ \tau $, implies that if $ i \in P $ then at least one of the following holds:
		\[ \left(h(\phi_i) = \min_{k\in[i+1-\tau/2..i]}\left\{h(\phi_k)\right\}\right) \text{or}  \left( h(\phi_i) = \min_{k\in[i..i-1+\tau/2]}\left\{h(\phi_k)\right\} \right). \]
		
		Since $S$ contains no $(\tau,\frac{\tau}{6})$-runs, then any two overlapping occurrences of a substring of length $\tau$ are at least $\frac{\tau}{6}$ apart. Together with the fact that $\phi$ has no collisions it follows that there are at least $\frac{\tau}{6}$ distinct fingerprints in any interval of size at least $\frac{\tau}{6}$ positions. Therefore, $ |\left\{\phi_j\right\}_{j\in[i+1-\tau/2..i]}| \ge \frac{\tau}{6}$ and the same holds for the interval $ [i..i-1+\tau/2] $. Since $ \mathcal{F} $ is  $ \left(\frac{1}{2},\tau\right)$-min-wise independent, then:
		\[ \Pr_{\phi,h}\left[h(\phi_i) = \min_{k\in[i+1-\tau/2..i]}\left\{h(\phi_k)\right\}\right] \leq \dfrac{6}{\tau}\dot(1+\frac{1}{2}) = \dfrac{9}{\tau}\]
		and the same holds for the interval $ [i..i-1+\tau/2] $.
		
		Therefore, $\Pr_{\phi,h}\left[A_i =1\right]\leq \dfrac{18}{\tau}$, and thus the lemma holds.
	\end{proof}

	\begin{proof}[{Proof of Corollary~\ref{cor:randomized-partitioning-withot-period}}]
		Due to Lemma~\ref{lem:space_exp}, when $\ID_{\phi,h}$ is picked randomly the expected size of the $(2\tau,2\tau)$-partitioning set induced by this $\ID$ function is $\frac{cn}{\tau}$ for some constant $c>1$. 
		By Markov's inequality this means that $\Pr_{\phi,h}(|P|>\frac{3cn}{\tau}) < \frac 1 3$. Running the algorithm described in the proof of Lemma~\ref{lemma:id_alg} until either $P$ is found or the algorithm fails by determining that $|P|>\frac{3cn}{\tau}$ takes $\O(n)$ time, and uses at most $\O(\frac n \tau)$. 
		Since the probability of the algorithm to fail is less than $\frac 1 3$, then the number of different $\ID$ function tries is a geometric random variable with a constant number of tries in expectation.
		Therefore, the expected overall running time is still $\O(n)$. 
	\end{proof}

	\subsection{Details for Periodic Blocks} \label{app:rand_selection_periodic_blocks}

	Recall that a $(d,\rho)$-run is a substring of $S$ of length at least $\delta$ with period length at most $\rho$ (see Definition~\ref{def:drho-run}).
	We have seen in Section~\ref{sec:randomSelection} and Section~\ref{sec:ens_rand_whp} that if $S$ contains no $(\tau,\frac{\tau}{6})$-runs one can build a small partitioning set in a linear time (in expectation or with high probability).
	Now, we will show how to partition $S$ efficiently into two types of substrings, those with short periods and those which do not contain any large enough substring with a short period. 
	Later, we show that together with the techniques of Section~\ref{sec:randomSelection} a small $(2\tau,2\tau)$-partitioning set could be built easily. 
	In summary, our technique works in the following manner:
	\begin{itemize}
		\item Scanning the text in intervals of size $\frac \tau 2$.
		\item Marking all intervals that could be a part of a $(\tau,\frac{\tau}{6})$-run.
		\item Verifying which intervals are actually contained in a $(\tau,\frac{\tau}{6})$-run.
		\item Inducing a partition of the text into $(\tau,\frac{\tau}{6})$-runs and substrings that do not contain any $(\tau,\frac{\tau}{6})$-run.
		\item Partitioning all substrings that do not contain any $(\tau,\frac{\tau}{6})$-run.
		\item Combining all the partitions to create a small $(2\tau,2\tau)$-partitioning set.
	\end{itemize}
	The same technique will also work for the combination with Section~\ref{sec:ens_rand_whp}, to get a small $(2\tau,2\tau)$-partitioning set in $\O(n)$ time with high probability.
	Moreover, the technique we use will produce a forward synchronized partitioning set to allow fast \SST{} construction. 
	Notice that the choice of positions for the partitioning set in the non-periodic case is forward synchronized.
	This is a sub-result of Lemma~\ref{lem:rand_is_forward_sync}.

	\paragraph{Defining candidate intervals.} 
	We prove that for every $(\tau,\frac \tau 6)$-run there is an interval of size $\frac \tau 2$ that is fully contained in the run, and that this interval contains no other run.
	We use this fact to efficiently find all $(\tau,\frac \tau 6)$-runs.
	
	Let $\alpha$ be an integer such that $0\leq \alpha \leq \frac{2n}{\tau}$. 
	Define the $\alpha$th interval as $I_\alpha = [\frac{\alpha\tau}{2}+1..\frac{(\alpha+1)\tau}{2}]$.
	We use a fingerprint function $\psi$ that has no collisions on substrings of length $\frac \tau 6$ from $S$. 
	For all $i\in[n-\frac{\tau}{6}+1]$ denote $\psi_i=\psi(S[i..i+\frac{\tau}{6}-1])$. 
	We show that if $S'=S[i..j]$ is a $(\tau,\frac \tau 6)$-run, then the algorithm could detect $S'$ period easily using $\psi$. 
	First notice that if $S'$ is a $(\tau,\frac \tau 6)$-run, then its length is at least $\tau$.
	Therefore, the next observation follows immediately.

	\begin{observation}\label{obs:candidate_int}
		If $S' = S[i..j]$ is a $(\tau,\frac \tau 6)$-run, then there is an integer $0\leq \alpha \leq \frac{2n}{\tau}$ such that $I_\alpha = [\frac{\alpha\tau}{2}+1..\frac{(\alpha+1)\tau}{2}]\subseteq [i..j]$.
	\end{observation}
	
	We next prove that using $\psi$ the algorithm can detect the principal period of any 
	$(\tau,\frac \tau 6)$-run.
	
	\begin{claim}\label{claim:contained_is_verified}
		Let $\alpha$ be an integer such that $I_\alpha$ is fully contained in $S'=S[i..j]$ which is a $(\tau,\frac \tau 6)$-run, and let $D$ to be the set of positions with minimum $\psi$ in the interval $\left[\frac{\alpha\tau}{2}+1..\frac{\alpha\tau}{2}+\frac{\tau}{3}\right]$:
		$$D= \left\{\ell\in\left[\frac{\alpha\tau}{2}+1..\frac{\alpha\tau}{2}+\frac{\tau}{3}\right] \mid \psi_\ell = \underset{k\in[\frac{\alpha\tau}{2}+1..\frac{\alpha\tau}{2}+\frac{\tau}{3}]}{\min} \psi_k\right\}.$$ 
		Then, $|D|\ge 2$ and $\rho_{S'}=\min \{|\ell-\ell'|\mid \ell,\ell'\in D \text{ and } \ell\ne \ell'\}$ i.e., the minimum distance between two indices in $D$ is exactly the principal period length of $S'$.
	\end{claim}
	\begin{proof}

		We first prove that there are $\ell_1,\ell_2\in D$ such that $\ell_2-\ell_1=\rho_{S'}$. 
		Then, we prove that there are no two elements $\ell_1,\ell_2\in D$ such that $|\ell_2 - \ell_1|<\rho_{S'}$.

		Let $\rho$ be the principal period length of $S'$, since $S'$ is a $(\tau,\frac\tau6)$-run, we have $\rho\le\frac\tau6$. 
		Thus, for all  $k\in[\alpha\frac{\tau}{2}+1..\alpha\frac{\tau}{2}+\frac{\tau}{3}-\rho]$ it follows that
		$S[k..k+\frac{\tau}{6}-1] = S[k+\rho..k+\rho+\frac{\tau}{6}-1] $
		and therefore $\psi_k = \psi_{k+\rho}$.
		In particular let $\ell_1 \in\left[\frac{\alpha\tau}{2}+1..\frac{\alpha\tau}{2}+\frac{\tau}{3}\right]$ be an index such that $\psi_{\ell_1}=\min_{k\in[\frac{\alpha\tau}{2}+1..\frac{\alpha\tau}{2}+\frac{\tau}{3}]} \psi_k$ then assuming $\ell_1\le \frac{\alpha\tau}{2}+\frac{\tau}{3}-\rho$ we have for $\ell_2=\ell_1+\rho$ that $\psi_{\ell_1}=\psi_{\ell_2}$ and $\ell_2-\ell_1=\rho$. If $\ell_1>\frac{\alpha\tau}{2}+\frac{\tau}{3}-\rho$ then let $\ell_2=\ell_1-\rho$ and notice that  $\psi_{\ell_1}=\psi_{\ell_2}$ and $\ell_1-\ell_2=\rho$.
		Hence, in any case, we have two indices in $D$ that their distance is exactly $\rho$.
		
		Assume by contradiction that there are two other positions $\ell<\ell'\in D$ such that $\ell'-\ell<\rho$.
		Then $\psi_{\ell'}=\psi_\ell$. Since we assume that $\psi$ has no collisions on substrings of $S$ , it must be that $S[\ell..\ell+\frac{\tau}{6}-1]=S[\ell'..\ell'+\frac{\tau}{6}-1]$. 
		Since $\ell'-\ell < \rho \le\frac{\tau}{6}$ then the two substrings overlap and thus the substring $S[\ell..\ell'+\frac{\tau}{6}-1]$ has a period of length $\ell'-\ell$. 
		But it is also a substring of $S'$ and thus it has a period of $\rho$.
		From Theorem~\ref{thm:FineAndWilf}, since the length of $S[\ell..\ell'+\frac{\tau}{6}-1]$ is  $\ell'-\ell+\frac{\tau}{6} \ge \ell'-\ell+\rho$ then it has a period of $\gcd(\ell'-\ell,\rho)$.
		It follows from Lemma~\ref{lem:periodSubstring} that since $S[\ell..\ell'+\frac{\tau}{6}-1]$ is a substring of $S'$ and its length is more than $\rho$ then $S'$ also has a period of length $\gcd(\ell'-\ell,\rho)$.
		But since $\gcd(\ell'-\ell,\rho) \le \min\{\ell'-\ell,\rho\}=\ell'-\ell < \rho$ then $S'$ has a period of length less than $\rho$, which contradicts the fact that $\rho$ is the principal period length of $S'$.
	\end{proof}
	
	\paragraph{Finding and Verifying Candidate Intervals.}
	For every interval $I_\alpha=[\frac{\alpha_0\tau}{2}+1..\frac{(\alpha_0+1)\tau}{2}]$ we say it is a candidate if the minimum value fingerprint in the first $\frac{2\tau}{6}$ positions of the interval repeats at least twice. 
	Each candidate is potentially a part of a $(\tau,\frac{\tau}{6})$-run. 
	Checking if an interval is a candidate is done by using the fingerprint function $\psi$ and calculating it by rolling window for the first $\frac{2\tau}{6}$ positions of the interval and maintaining the value and position of the minimum fingerprint value. 
	This process will stop if the minimum fingerprint value repeats at least twice, or if none was found.
	The verification takes in the worst case $\O(\tau)$ time per interval.
	For every candidate interval, the difference between the two positions with the minimum fingerprint value is maintained. This difference is the suspected principal period of a run this interval is contained in.
	
	We say a candidate interval is \emph{verified} if it is fully contained in a $(\tau,\frac{\tau}{6})$-run. 
	From Observation~\ref{obs:candidate_int}, it follows that each $(\tau,\frac{\tau}{6})$-run has at least one verified candidate.
	The fact whether a candidate interval is contained in a $(\tau,\frac{\tau}{6})$-run can be checked in $\O(\tau)$ time by first verifying that the occurrences of the minimum fingerprint value are at most $\frac{\tau}{6}$ apart, then verifying that there is no fingerprint collision.
	If there is no collision then we now know the principal period length and the period itself. 
	By extending up to at most $\tau$ positions to the left and to the right until the period breaks we can verify that the periodic substring is at least of length $\tau$, and that it contains the candidate interval.
	Therefore, the candidate interval can be verified to be contained in a $(\tau,\frac{\tau}{6})$-run in $\O(\tau)$ time.
	
	For each interval $I_\alpha$ we spend $\O(\tau)$ time and $\O(1)$ words of space to check if it is a candidate, and then another $\O(\tau)$ time to verify if it is a part of $(\tau,\frac{\tau}{6})$-run. 
	Since there are $\O(\frac n \tau)$ intervals, then all the candidates can be found and verified in $\O(n)$ time, using $\O(\frac n \tau)$ words of space to store the verified candidates' indices. 
	Notice that the algorithm does not assume the fingerprint has no collisions. 
	The algorithm rather uses the fingerprint value to hint on positions of interest, and then the algorithm verifies the equality. 
	Since the fingerprint function is collision free \whp{}, then the algorithm is Las-Vegas, and \whp{} its expected running time is $\O(n)$.
	
	\paragraph{Overlapping Runs.}
	Remark that two maximal runs with different periods may overlap if their periods share a common substring.
	We claim next that two maximal $(\tau,\frac \tau 6)$-runs may only overlap in less than $\frac \tau 3$ positions.
	
	\begin{lemma}\label{lem:maxoverlap}
		If $u=S[i..j]$ and $v=S[i'..j']$ are two maximal overlapping $(\tau,\frac \tau 6)$-runs, meaning $i \leq i' \leq j$, then either $j-i'+1 < \frac{\tau}{3}$ or $u=v$
	\end{lemma}
	\begin{proof}
		The second case is trivial. Therefore, assume $i < i' \leq j < j'$.
		Let $\rho_u,\rho_v$ be the lengths of the principal periods of $u,v$ respectively. 
		By way of contradiction, assume that $j-i'+1 \ge \frac{\tau}{3}$.
		Hence, $w=S[i'..j]$ is a string such that both $\rho_u$ and $\rho_v$ are periods of $w$.
		Since $\frac{\tau}{3} \ge \rho_u+\rho_v$ then $j-i'+1 \ge \rho_u+\rho_v - \gcd(\rho_u,\rho_v)$.
		Therefore, due to Theorem~\ref{thm:FineAndWilf}, $w$ has a period length of $\gcd(\rho_u,\rho_v)$.
		
		$w=S[i'..j]$ is a substring of both $u$ and $v$ of length greater than their period lengths, and also $w$ has a period length that divides both $\rho_u$ and $\rho_v$. 
		Hence, by Lemma~\ref{lem:periodSubstring} $\gcd(\rho_u,\rho_v)$ is a period length of both $u$ and $v$, but since $\rho_u$ and $\rho_v$ are the principal period length of $u$ and $v$ respectively, and hence the minimum period length, then $\rho_u=\rho_v=\gcd(\rho_u,\rho_v)$. 
		
		Denote $\rho = \rho_u=\rho_v$. From the assumption that $i < i'$ we know that $i\leq i'-1$. 
		We also know that $S[i'-1] = S[i'+\rho-1]$ since $u=S[i..j]$ is periodic and $j\ge i'+\rho-1$. 
		But $\rho$ is also a period of $v$, hence the run $v$ can be extended to the left without breaking its periodicity. 
		This contradicts the maximality of $v$, hence it must be that $j-i'+1 < \frac{\tau}{3}$.
	\end{proof}
	
	From Lemma~\ref{lem:maxoverlap}, the next corollary follows.
	
	\begin{corollary}\label{cor:fully_con_in_one}
		An interval $I_\alpha$ may be fully contained in at most one maximal $(\tau,\frac \tau 6)$-run.
	\end{corollary}
	\begin{proof}
		By a way of contradiction, if $I_\alpha$ is fully contained in at least two maximal $(\tau,\frac \tau 6)$-runs, then those two runs overlap with at least $\frac \tau 2$ characters, which contradicts Lemma~\ref{lem:maxoverlap}.
	\end{proof}
	
	Putting all together, the next lemma follows.
	
	\begin{lemma}\label{lem:run_certificate}
		At the end of the candidate interval verification process, for every maximal $(\tau,\frac \tau 6)$-run there is a verified candidate interval that is fully contained only in the run, and not in any other maximal $(\tau,\frac \tau 6)$-run. 
	\end{lemma}
	\begin{proof}
		From Observation~\ref{obs:candidate_int} there is an interval $I_\alpha$ that is fully contained in $u$. 
		It follows from Claim~\ref{claim:contained_is_verified} that this interval matches the definition of a verified candidate interval and that in the process, the principal period of $u$ is found.
		Finally, from Corollary~\ref{cor:fully_con_in_one} follows that $I_\alpha$ is contained only in $u$.
	\end{proof}
	
	From the above lemma we learn that for every maximal $(\tau,\frac \tau 6)$-run there is at least one verified candidate that is unique to it.
	Therefore, by finding all verified candidate intervals, we get at least one ``certificate" for every maximal $(\tau,\frac \tau 6)$-run.
	Thus, all that is left is to mark all those runs.
	
	\paragraph{From candidate intervals to runs-partitioning.}
	Our aim is to partition the string to long periodic substrings with small period lengths and to substrings that contain no such substrings.
	We will build a set $Q$ containing triplets partitioning the string accordingly.
	Each triplet will contain two endpoints of a certain substring, its principal period if it is part of $(\tau,\frac \tau 6)$-run or a marker indicating that it should be considered as non-periodic.
	The triplets will be ordered by their endpoints' positions from left to right, and will cover the whole string.
	
	After finding all candidate intervals, the algorithm stores all verified candidates, meaning all intervals that are contained in some $(\tau,\frac \tau 6)$-run. 
	The algorithm scans the verified candidates from right to left. 
	We divide the scan into iterations, where at the end of each iteration a whole $(\tau,\frac \tau 6)$-run is processed and its endpoints stored in $Q$, perhaps without parts of his right margin in a case of overlapping runs.
	
	Each iteration starts with the rightmost verified candidate interval that was not marked. 
	The algorithm extends the period to the right until the period breaks or until a previously marked run is found, and to the left as long as the period does not break.
	Denote $u_i = S[q_i..q'_i]$ be the run found in the $i$th iteration.
	First, the pair $(q_i,q'_i)$ is stored in the set $Q$ and marked as a run. 
	If the pair found in the previous iteration, $(q_{i-1},q'_{i-1})$, is not adjacent to the current iteration pair --- meaning $q_{i-1}\neq q'_i+1$ --- we add the pair $(q'_i+1,q_{i-1}-1)$ to $Q$ and mark it as the endpoints of a non-run substring.
	Finally, every verified candidate interval that is fully contained in $u_i$ is marked.
	
	Since all verified candidate intervals were processed, it follows from Lemma~\ref{lem:run_certificate} that all maximal $(\tau,\frac \tau 6)$-runs were found, and are stored in $Q$ together with their principal period.
	Since there are at most $\O(\frac{n}{\tau})$ verified candidate intervals, then there are at most $O(\frac{n}{\tau})$ triplets contained in $Q$.
	
	This process scans each character at most once as a part of a $(\tau,\frac \tau 6)$-run, and at most once as a character that breaks the run.
	The verified candidate intervals' marking takes an additional negligible $\O(\frac{n}{\tau})$ time. 
	Therefore, the total running time is $\O(n)$.
	Regarding the space usage --- each iteration used constant words of space, and maintaining $Q$ takes $\O(\frac{n}{\tau})$ words of space due to the bound on its size.
	
	To conclude, using $\O(n)$ time and $\O(\frac{n}{\tau})$ words of space we found all $(\tau,\frac \tau 6)$-runs, and thus could partition it efficiently.
	
	\paragraph{From runs-partitioning to partitioning set.}
	After finding all the runs, we can construct a partitioning set.
	Basically, the technique is to use the scanning algorithm from the proof of Lemma~\ref{lemma:id_alg}, with modifications regarding runs from $Q$.
	We start with an empty set $P$ and add positions to it, picked by the scanning algorithm, as follows.
	If the first pair in $Q$ is non-run substring of length at least $\tau$, the scanning algorithm runs on its positions and updates $P$ accordingly, and finishes at the last interval of length $\tau$ contained in the substring.
	If it is of a length shorter than $\tau$, we pick no position for $P$ from that substring.
	The next pair in $Q$ marks the endpoints of a $(\tau,\frac \tau 6)$-run, denoted $u=S[q..q']$.
	First, we add $q$ to $P$. If the next pair in $Q$ is marked as a run substring, we recursively run the process with the next $(\tau,\frac \tau 6)$-run in $Q$. 
	Otherwise, $q'+1$ is an endpoint of a non-run substring.
	We skip processing the characters of the string from position $q$ until position $q'-\tau+1$.
	Then we run the scanning algorithm from this position, picking positions for $P$ from each interval of size $\tau$ contained in the substring.
	Finally, we iterate using this algorithm on all pairs of $Q$.
	
	The only exception will be the last substring in $Q$ if it is marked as a non-periodic string.
	Let the second to last substring in $Q$ be $S[q..q']$.
	If $q'\in[n-\tau+1..n-1]$, we pick $q'+1$ as the last position in $P$.
	Otherwise, we run the scanning algorithm as described above from position $q'-\tau+1$, until picking the last position of $P$ from the interval $[n-2\tau+1..n-\tau]$, as was done in the non-periodic case.
	
	By using the scan algorithm from the last $\tau$ positions of a $(\tau,\frac \tau 6)$-run, we select the periodic block end, and thus the next non-periodic short block start position according to the \ID{} values.
	This will allow us to prove that the output $P$ is a forward synchronized partitioning set.
	Firstly we prove that at the end of this process, $P$ is a $(2\tau,2\tau)$-partitioning set. 
	We then claim that this partitioning set is also forward synchronized.
	Finally, we show that $P$ can be computed in $\O(n)$ expected time, and commit that $|P|=\O(\frac{n}{\tau})$.
	
	\begin{lemma}
		Using the modified scan algorithm described above, the output $P$ is a $(2\tau,2\tau)$-partitioning set.
	\end{lemma}
	\begin{proof}
		We show that $P$ has the properties of a $(2\tau,2\tau)$-partitioning set according to Definition~\ref{def:landmarks_set}.
		
		\paragraph{Property~\ref{item:blocks_sync}.} 
		Let $i,j\in[1+2\tau..n-2\tau]$ such that $S[i-2\tau..i+2\tau]=S[j-2\tau..j+2\tau]$.
		Denote $u_i,u_j$ the substrings induced from the set $Q$ such that $i,j$ are positions within the substrings $u_i,u_j$ respectively. 
		
		Notice that $u_i$ is either fully contained in $S[i-2\tau..i+2\tau]$ or it intersects with at least $2\tau$ characters of it, and the same holds for $u_j$ and  $S[j-2\tau..j+2\tau]$. 
		Therefore, if $u_i$ is a $(\tau,\frac{\tau}{6})$-run, then position $i$ is contained in such a run.
		Thus, it follows that $u_i$ is a $(\tau,\frac \tau6)$-run if and only if $u_j$ is a $(\tau,\frac \tau6)$-run.
		Hence, both $u_i$ and $u_j$ are marked the same in $Q$.
		We shall differentiate each possibility.

		\begin{enumerate}
			
			\item \emph{($u_i$ and $u_j$ are marked as non-run substrings)} 
			In this case, positions are picked for $P$ by their \ID{}s. 
			Since $S[i-2\tau..i+2\tau]=S[j-2\tau..j+2\tau]$ then from the \ID{} definition it follows immediately that $i\in P \Leftrightarrow j\in P$.
			
			\item \emph{($u_i$ and $u_j$ are both marked as run substrings)}
			Only two positions are picked for $P$ within a $(\tau,\frac \tau 6)$-run: The first position, and the position with the minimum \ID{} out of the last $\tau$ positions of the run if the adjacent substring from $Q$ is marked as a non-run substring.
			First, since $S[i-2\tau..i+2\tau]=S[j-2\tau..j+2\tau]$ then if $i$ is at the start of the run $u_i$ then $j$ is at the start of the run $u_j$, and vice versa.
			Therefore, in this case $i\in P \Leftrightarrow j\in P$.
			
			Notice that, again, since $S[i-2\tau..i+2\tau]=S[j-2\tau..j+2\tau]$, then if $i$ is in the last $\tau$ positions of the run $u_i$, so does $j$ (regarding the run $u_j$).
			Moreover, they are at the same distance from the end of the run, and since the substrings from $Q$ consequent to $u_i,u_j$ respectively share at least $\tau$ characters, then those subsequent substrings are of the same type.
			Therefore, either those subsequent substrings are both $(\tau,\frac \tau 6)$-run, and thus neither $i$ nor $j$ is picked for $P$.
			If those substrings are both non-runs, then positions are picked by their \ID{} values.
			As explained in the case of non-run substrings, since both positions will be selected by their \ID{} value it yields that $i\in P \Leftrightarrow j\in P$.
		\end{enumerate}
		
		\paragraph{Property~\ref{item:blocks_size}.}
		Let $1=p_0<p_1<\cdots<p_k<p_{k+1}=n+1$ be all the positions in the partition $P\cup\{1,n+1\}$.
		
		If $1=p_0\notin P$ then no $(\tau,\frac \tau 6)$-run start at $p_0$.
		Therefore, there is either a $(\tau,\frac \tau 6)$-run that starts at a position $k\in[1..\tau]$ or the algorithm will pick a position $k\in[1..\tau]$ for $P$ using the \ID{} function.
		Hence, in the case where $1=p_0\notin P$, it follows that $p_1-p_0 \le 2\tau$.
		
		As for the general case, let $p_j<p_{j+1}$ be two consecutive positions in $P\cup\{n+1\}$. 
		We would like to prove that either $p_{j+1}-p_j \le 2\tau$ or that $S[p_j..p_{j+1}-1]$ is a periodic substring with period length smaller than $2\tau$.
		There are 3 cases: $p_{j}$ is the beginning of a periodic block, $p_{j}$ is a position selected for $P$ due to its \ID{} value, or $p_j$ is the last position picked for $P$ --- selected after a $(\tau,\frac \tau 6)$-run that ended in the last $\tau$ characters of $S$.
		The last case is trivial since by definition $p_{j+1}-p_j = n+1 - p_j < \tau$.
		
		Assume that $p_{j}$ is the start of a $(\tau,\frac{\tau}{6})$-run. 
		Although it might be that $p_{j+1}-p_j > 2\tau$, since $S[p_j..p_{j+1}-1]$ is a periodic substring with period smaller than $\frac{\tau}{6}$, then this block complies with the requirement of Property~\ref{item:long_blocks} of the partitioning set.
		
		On the other hand, assume that $p_{j}$ was selected for $P$ due to its \ID{} value. 
		If $p_{j+1} = n+1$ then $p_j$ was selected for $P$ due to its minimal \ID{} value in the interval $[n-2\tau-1..n-\tau]$.
		If $p_{j+1}$ was also selected due to its \ID{} value, then from the \ID{} selection criteria it must hold that $p_{j+1} - p_j \le 2\tau$. 
		Otherwise, if $p_{j+1}$ is the first position in a $(\tau,\frac \tau 6)$-run, then $p_j$ must be within at most $2\tau$ characters, and Property~\ref{item:blocks_size} holds.
	\end{proof}
	
	We now prove that $P$ is forward synchronized partitioning set. 
	Recall that due to Definition~\ref{def:forward_sync}, we have to prove that for any pair $p_i,p_j\in P$ such that $\LCE(p_i,p_j)>p_{i+1}-p_i+\delta$ then $p_{i+1}-p_i=p_{j+1}-p_j$, where $p_{i+1}$ and $p_{j+1}$ are the consecutive positions of $p_i$ and $p_j$ in $P$ respectively, and $\delta$ is the locality factor of the partitioning set.
	
	\begin{lemma}\label{lem:rand_is_forward_sync}
		Using the modified scan algorithm described above, the output $P$ is a forward synchronized $(2\tau,2\tau)$-partitioning set.
	\end{lemma}
	\begin{proof}
		Let $p_i,p_j\in P$ such that $\LCE(p_i,p_j) > p_{i+1}-p_i + 2\tau$. 
		Let $\Delta = p_{i+1}-p_i$.
		Hence, $S[p_i..p_{i+1} + 2\tau] = S[p_j..p_{j} + \Delta + 2\tau]$.
		There are four cases for $p_i$ and $p_{i+1}$:
		\begin{enumerate}
			\item Assume both $p_i$ and $p_{i+1}$ were picked to $P$ by their \ID{} value. 
			Therefore, $\ID(p_{i+1})$ was either the first \ID{} after $p_i$ such that $\ID(p_i) \geq \ID(p_{i+1})$ or it has the minimum value in the interval $[p_i +1 .. p_i + \tau]$.
			In either case since $S[p_i..p_{i+1} + 2\tau] = S[p_j..p_{j} + \Delta + 2\tau]$ then the \ID{}s in that interval are equal and thus it also holds that $p_j$ is not the start of a periodic block and that $p_{j+1} = p_j + \Delta$.
			
			\item Assume both $p_i$ and $p_{i+1}$ are the start of periodic blocks. 
			Therefore, $p_{i+1}$ is the start of a $(\tau,\frac{\tau}{6})$-run with period that is different then the period of the $(\tau,\frac{\tau}{6})$-run that starts at $p_i$. 
			But since $S[p_i..p_{i+1} + 2\tau] = S[p_j..p_{j} + \Delta + 2\tau]$ then $p_{j}$ is also a start of a  $(\tau,\frac{\tau}{6})$-run, and $p_{j} + \Delta$ is the start of  $(\tau,\frac{\tau}{6})$-run with different period then  the run that starts at $p_{j}$.
			Therefore, $p_{j+1} = p_{j} + \Delta$.
			
			\item  Assume that $p_i$ is the start of a periodic block and $p_{i+1}$ was selected for its \ID{} value.
			By the position selection algorithm, it means that $p_{i+1}$ is selected from the last $\tau$ positions of the $(\tau,\frac \tau 6)$-run that starts at $p_i$, and thus the run length is less than $\ell < \Delta + \tau$.
			By the selection algorithm selecting $p_{i+1}$ means that $\ID(p_{i+1}) = \min\{\ID(k)\mid k\in[p_i+\ell-\tau..p_i+\ell-1]\}$. 
			Those \ID{}s are based solely on the characters in $[p_i+\ell -\tau..p_i+\ell+\tau-2] \subseteq [p_i..p_i+\Delta+2\tau]$.
			Since $S[p_i..p_{i+1} + 2\tau] = S[p_j..p_{j} + \Delta + 2\tau]$, then $p_j$ is also a start of a $(\tau,\frac \tau 6)$-run, with the exact same length $\ell$, and also $\ID(p_{j}+\Delta) = \min\{\ID(k)\mid k\in[p_j+\ell-\tau..p_j+\ell-1]\}$.
			Therefore, $p_{j+1} = p_j + \Delta$.
			
			\item  Assume that $p_i$ is selected for its \ID{} value and that $p_{i+1}$ is the start of a $(\tau,\frac \tau 6)$-run.
			This means that $\Delta < 2\tau$ and that $\ID(p_i) = \min\{\ID(k)\mid k\in[p_i..p_{i+1}-\tau]\}$.
			Since $S[p_i..p_{i+1} + 2\tau] =\newline S[p_j..p_{j} + \Delta + 2\tau]$ then at $p_j +\Delta$ a periodic string of length at least $\tau$ and of period at most $\frac{\tau}{6}$ starts, right after a non-periodic sub-string.
			Hence, $p_j+\Delta$ will be chosen for $P$.
			In addition, it also follows that $\ID(p_j) = \min\{\ID(k)\mid k\in[p_j..p_{j} +\Delta-\tau]\}$.
			Therefore, no other selection is picked between $p_j$ and $p_j+\Delta$.
			Thus, $p_{j+1} = p_j + \Delta$.
		\end{enumerate} 
	\end{proof}
	
	Now that we ensured that $P$ is a forward synchronized $(2\tau,2\tau)$-partitioning set, finally we show that the process of selecting $P$ using the runs-partition will take expected linear time and ensures that $|P|=\O(\frac n \tau)$.
	From Lemma~\ref{lemma:id_alg}, we know that using the $\ID$ function it is possible to select a partitioning set of a non-periodic substring using time linear in the substring's length, with an additional $\O(\tau)$ time needed for every step back.
	When processing $S$ we skip all $(\tau,\frac \tau 6)$-run completely, other than the last $\tau$ characters, and every position is processed as part of only one non-periodic substring.
	Therefore, similarly to the results in Section~\ref{sec:randomSelection} a partitioning set of size $\O(\frac n \tau)$ is expected to be built in $\O(n)$ time.
	
	Remark that the techniques introduced in Section~\ref{sec:ens_rand_whp} work similarly in linear time.
	
	\paragraph{Conclusion.}
	In order to build a partitioning set of a general string $S$ containing $(\tau,\frac \tau 6)$-runs we broke the problem to three steps.
	We used $\O(n)$ expected time with high probability and $\O(\frac n \tau)$ space to find intervals that might contain a $(\tau,\frac \tau 6)$-run.
	Then another $\O(n)$ time and $\O(\frac n \tau)$ space to find all $(\tau,\frac \tau 6)$-runs and building $Q$ to partition $S$ to different $(\tau,\frac \tau 6)$-runs and to contiguous substring that contain non of those.
	Finally, another $\O(n)$ expected time and $\O(\frac n \tau)$ space were used to create a $(2\tau,2\tau)$-partition-set, or $\O(n)$ expected time with high probability and $\O(\frac n \tau + \log n)$ space using the techniques of Section~\ref{sec:ens_rand_whp}.
	Summing everything up we get the results of Theorem~\ref{thm:lasvegasLCE} for a general string $S$.
	
	\subsection{Faster Partitioning Set Counting For High Probability Method}\label{app:fast_random_parallel_count}
	
	As described in Section~\ref{sec:ens_rand_whp}, in order to reduce the time the algorithm spends on stepping back while scanning the string in order to calculate $|P_i|$ for all $h_i$, the algorithm maintains future candidates in advance --- up to $\O(\frac{n}{\tau \log n})$ per hash function.

	The scanning algorithm will have to step-back after the $j$th iteration if $\ID_{\phi,h_i}(j) < \ID_{\phi,h_i}(k)$ for all $k\in[j+1..j+\tau-1]$.
	In this case the position with the minimal value in $[j+1..j+\tau]$ need to be picked.
	By maintaining the candidates for several windows in advance, the algorithm avoids stepping back when the most recent candidate was out of the current window. 
	A position $\ell$ is a candidate in the $j$th iteration if $j \leq \ell < j+\tau-1$  and $\ID_{\phi,h_i}(\ell) \leq \ID_{\phi,h_i}(k)$ for all $k\in[\ell+1..j+\tau-1]$. 
	Remark that by definition, if the candidates at a certain iteration are $\{\ell_1,\dots,\ell_q\}$ such that $\ell_1 < \cdots < \ell_q$ then $\ID_{\phi,h_i}(\ell_1) \leq\cdots \leq \ID_{\phi,h_i}(\ell_q)$. 
	Hence, at the end of each iteration the leftmost candidate $\ell_1$ is a part of the partitioning sets.
	
	For a certain hash function $h_i$ the candidates and their IDs are maintained in a deque, denoted $D_i$.
	Candidates are pushed into $D_i$ from its right side in an ordered way in both position index and ID value.
	After calculating the ID of some position $\ell$ the algorithm peeks at the right end of $D_i$ and it pops all the obsolete candidates $(k,\ID_{\phi,h_i}(k))$ such that $\ID_{\phi,h_i}(\ell) < \ID_{\phi,h_i}(k)$, and then pushes $(\ell,\ID_{\phi,h_i}(\ell))$ to the right end of $D_i$.
	When $D_i$ reaches the point where it contains $\frac{c'n}{\tau\log n}$ candidates, a flag denoted $f_i$ is risen.
	While $f_i$ is raised, no new candidates are added to $D_i$ unless they remove other candidates (meaning only better candidates are inserted). 
	If such a better candidate is found, the flag $f_i$ is lowered back.
	
	The algorithm maintains throughout its run the rightmost position picked for each $P_i$ --- denoted $q_i$ and the current size $|P_i|$.
	It processes the string as follows for every $i$. 
	Insert all the candidates from $[1..\tau]$ into $D_i$ according to the insertion rules from above.
	Remember the leftmost position in $D_i$ for the partitioning set $P_i$, and increment $|P_i|$ value.
	Afterward the algorithm works in iterations.
	In the $j$th iteration check if the leftmost candidate in $D_i$ is the position $j-1$. 
	If so, pop it from the left end of $D_i$.
	If $D_i$ is not empty --- insert $j+\tau-1$ and its ID into $D_i$ according to the insertion rules from above.
	After inserting a new position to $D_i$, check if the leftmost position in $D_i$ has already been picked for $P_i$ by comparing it to $q_i$.
	If the leftmost position in $D_i$ was not picked for $P_i$ yet, update $q_i$ with the new position and increment the current size $|P_i|$.
	Otherwise, if $D_i$ is empty, it must be because the deque $D_i$ was full, and its flag $f_i$ was raised and never lowered.
	So first --- lower the flag $f_i$.
	Then remark that there is no knowledge on any candidates in the current interval, thus the algorithm needs to step-back and recalculate the fingerprints and \ID{}s of all the positions in $[j..j+\tau-1]$, exactly as it did in the first iteration.
	
	\begin{observation}
		A recalculation of $\tau$ fingerprints and the IDs of candidates may occur at most $\log n$ times.
	\end{observation}
	\begin{proof}
		A recalculation in the above algorithm happens only when the deque $D_i$ is empty and $f_i$ is still raised.
		This happens only when $D_i$ was full at the end of some iteration, the $j'$th, and afterward no new candidates were inserted and all of them were removed since they were not relevant for the new intervals.
		Before a candidate is removed from the left end of the deque $D_i$ it is picked for the partitioning set $P_i$, and since $D_i$ is empty it means that all the candidates that were in $D_i$ at the $j'$th iteration were counted as positions in $P_i$.
		There were $\frac{c'n}{\tau \log n}$ positions at the $j'$th iteration, and therefore every time the deque $D_i$ is empty, it means that $\frac{c'n}{\tau \log n}$ positions are counted as positions in $P_i$.
		Since the algorithm stops its run on $h_i$ once it is determined that $|P_i| > \frac{c'n}{\tau}$ then after $\log n$ step-backs of the algorithm on the hash function $h_i$, the run on $h_i$ will stop.
	\end{proof}
	
	\begin{lemma}
		Calculating $|P_i|$  or determining that $|P_i| > \frac{c'n}{\log n}$ for all $h_i$ takes $\O(n+\tau\log^2 n)$ time using $\O(\frac{n}{\tau} + \log n)$ words of space.
	\end{lemma}
	\begin{proof}
		When all functions move forward simultaneously, we analyzed earlier that it takes $\O(n)$ time.
		For each $h_i$, using at least one word of space for the rightmost picked position, and at most $\O(\frac{n}{\tau\log n})$ words of space for candidates, the run may have to step-back at most $\log n$ times, and stepping back will cost $\O(\tau)$ time.
		This fact holds for all the $\log n$ hash functions, thus in total the steps back may cost $\O(\tau\log^2 n)$ time in the worst case. Also, notice that each position in $P_0$ is inserted exactly once, and removed exactly once from the deque $D_i$. 
		Thus, there are exactly $2n$ insertion and removal operations in all the deques throughout the whole run, costing $\O(n)$ time in total.
		
		Summing all up, it costs $\O(n+\tau\log^2 n)$ time and $\O(\frac{n}{\tau}+\log n)$ words of space to run the above algorithm.
	\end{proof}

	\section{Missing Details for the  Deterministic Selection }\label{app:deterministic_selection}

	In this appendix we complete the algorithm for the deterministic selection of $(\O(\tau),\O(\tau\log^*n))$-partitioning set from Section~\ref{sec:determinsiticSelection}. 
	In Appendix~\ref{app:deterministic_correctness}, we prove that the selected positions are actually a $(\O(\tau),\O(\tau\log^* n))$-partitioning set. Then, in Appendix~\ref{app:deterministic_in_space} we describe how to find these positions while using only $\O(\frac n\tau+\log n)$ words of working space, which for $\tau\le \O(\frac n{\log n})$ is $\O(\frac n\tau)$ words of space. 
	
	\subsection{Correctness Proof}\label{app:deterministic_correctness}
	In this section we prove that the selection we have described in Section~\ref{sec:determinsiticSelection}  yields a $(\O(\tau),\O(\tau\log^* n))$-partitioning set of $S$, as defined in Definition~\ref{def:landmarks_set}.
	Recall that the partitioning of Section~\ref{sec:determinsiticSelection} is a hierarchical decomposition of the range $[1..n]$, where its top-level is level number $\log_{3/2}\tau$. For any $\mu$ from $0$ to $\log_{3/2}\tau$, let $P_\mu\subseteq[n]$ be the set of positions corresponding to level $\mu$ of the decomposition.
	Notice that we have $P_{\ceil{\log_{3/2}\tau}}\subseteq \cdots\subseteq  P_1\subseteq P_0=[n]$.
	So, our goal is to prove that $P_{\log_{3/2}\tau}$ is an $(\O(\tau),\O(\tau\log^*n))$-partitioning set of size $\O(\frac n\tau)$. We will prove it as a special case, for the general case which is that for any $0\le\mu\le\log_{3/2}\tau$ the set $P_\mu$ is an $\left(\O((3/2)^\mu),\O((3/2)^\mu\cdot\log^*n)\right)$-partitioning set
	of size $\O\left(\frac n{(3/2)^\mu}\right)$.

	For any $\mu$ we have to prove three properties of $P_\mu$: the two properties of Definition~\ref{def:landmarks_set}, and that the size of the set is $\O(\frac n{(3/2)^\mu})$.
	We prove each of these properties by induction.

	We first prove that the size of $P_\mu$ is $\O\left(\frac n{(3/2)^\mu}\right)$. This property is trivially deduced from the following lemma.

	\begin{lemma}\label{lem:consecutive_blocks_length}
		For any level $0\le\mu\le\log_{3/2}\tau$ of the hierarchic decomposition, the length of any pair of two consecutive blocks is at least  $(3/2)^\mu$.

	\end{lemma}
	\begin{proof}
		We prove the lemma by induction on the level of the hierarchy. 
		For level $\mu=0$, the lemma is trivial.
		We assume the lemma holds for level $\mu-1$, and  prove it holds for level $\mu$ as well. 
		Let $[p_i..p_{i+1}-1]$ and $[p_{i+1}..p_{i+2}-1]$ be two consecutive blocks at level $\mu$. 
		If at least one of the blocks is a sub-block of type~\ref{item:type_large}, then the claim is obtained from the definition of  type~\ref{item:type_large}.
		Otherwise, if each one of the blocks is composed of at least two sub-blocks of level $\mu-1$, the claim is obtained by the induction hypothesis as their total length must be at least $2\cdot(3/2)^{\mu-1}>(3/2)^\mu$.
		
		The only remaining case where a block is composed of only one sub-block is the case that this sub-block is the only sub-block in a sequence of type~\ref{item:type_edge}. 
		Since we already covered the case where one of the blocks is of type~\ref{item:type_large}, the only remaining case is that the other block is of type~\ref{item:type_contiguous} (periodic block which is composed of sub-blocks corresponding to the same substring). 
		Let us assume without loss of generality that $[p_i..p_{i+1}-1]$ is the block of type~\ref{item:type_contiguous} and that $[p_{i+1}..p_{i+2}-1]$  is also a block in level $\mu-1$.
		We denote the length of each sub-block which assemble $[p_i..p_{i+1}-1]$ as $x$ (recall that all the sub-blocks that assembles $[p_i..p_{i+1}-1]$ are identical, and in particular they have the same length), and the length of $[p_{i+1}..p_{i+2}-1]$ as $y$. By the induction hypothesis, it must be that $x+y\ge (3/2)^{\mu-1}$ and also that $2x\ge (3/2)^{\mu-1}$, since in level $\mu-1$ we have a pair of consecutive blocks of length $x$, and another (overlapping) pair of consecutive blocks of lengths $x$ and $y$. 
		The length of $[p_i..p_{i+1}-1]$ is at least $2x$, therefore the total length of the pair $[p_i..p_{i+1}-1]$ and $[p_{i+1}..p_{i+2}-1]$ is at least $2x+y=(x+y)+x\ge (3/2)^{\mu-1}+\frac{(3/2)^{\mu-1}}2=(3/2)^\mu$.
	\end{proof}

	Property~\ref{item:blocks_sync} of Definition~\ref{def:landmarks_set} states that in any partitioning set the selection of each position is \emph{local}, which means it depends on a substring around this position of length $\O(\delta)$.
	In the following lemma we prove that the positions of the decomposition in each level of the hierarchy have this locality property. 
	Let us denote $S_{a,\mu}=S[a-6\cdot(3/2)^\mu c\log^*n..a+6\cdot(3/2)^\mu c\log^*n]$ to be the substring of $S$ of length $12\cdot(3/2)^\mu c\log^*n+1$ around position $a$.
	In addition, for any $p_i\in P_\mu$, let $p_{i+1}$ be the position following $p_i$ in $P_\mu$ (or $n$ if $p_i$ is the last position in $P_\mu$), we say that $[p_i..p_{i+1}-1]$ is \emph{the block} of $p_i$  in level $\mu$.

	\begin{lemma}\label{lem:deterministic_local_property}
		
		For any level $0\le\mu\le\log_{3/2}\tau$ of the hierarchic decomposition, if $a,b\in [n]$ are two indices such that $S_{a,\mu}=S_{b,\mu}$ then $a\in P_\mu\Leftrightarrow b\in P_\mu$
		
	\end{lemma}
	\begin{proof}
		We prove the lemma by induction on the decomposition's level. For level $\mu=0$, the lemma is trivial, since $P_0=[n]$.
		Let assume the lemma holds for $P_{\mu-1}$, and we prove that it holds for $P_\mu$ as well.
		Let $a$ and $b$ be two positions such that $S_{a,\mu}=S_{b,\mu}$ 
		and $a\in P_\mu$.
		Thus, in particular for any $-2(3/2)^\mu c\log^*n\le\delta\le2(3/2)^\mu c\log^*n$
		we have  $S_{a+\delta,\mu-1}=S_{b+\delta,\mu-1}$.
		Therefore, due to the induction hypothesis, it must be that $a+\delta\in P_{\mu-1}\Leftrightarrow b+\delta\in P_{\mu-1}$.
		Moreover, since $P_\mu\subseteq P_{\mu-1}$ and $a\in P_\mu$, it must be that $a\in P_{\mu-1}$ and therefore $b\in P_{\mu-1}$ as well.
		We prove the lemma by partitioning the cases according to the type of the sequence containing $a$ in the parsing of $S$ for level $\mu$.

		\paragraph{Type~\ref{item:type_large}.} 
		Since the length of the level $\mu-1$  sub-block that $a$ is its leftmost index is at least $(3/2)^\mu$,  it must be that $[a..a+(3/2)^\mu-1]\cap P_{\mu-1} = \{a\}$.
		Hence, it also must be that $[b..b+(3/2)^\mu-1]\cap P_{\mu-1} = \{b\}$. 
		Therefore, the sub-block of $b$  in level $\mu-1$ is of length at least $(3/2)^\mu$ and therefore it is also a block in level $\mu$. Hence, $b\in P_\mu$.

		\paragraph{Type~\ref{item:type_contiguous}.} 
		If the block of $a$ in level $\mu$ is not of type~\ref{item:type_large}, then it must be that its sub-block length is $\delta<(3/2)^\mu$. Hence, the sub-block of $a$ in level $\mu-1$  is $[a..a+\delta-1]$ and $(a+\delta)\in P_{\mu-1}$. 
		Moreover, since  $[a..a+\delta-1]$ is in a sequence of type ~\ref{item:type_contiguous} it must be that the sub-block which follows  $[a..a+\delta-1]$  in level $\mu-1$ is $[a+\delta..a+2\delta-1]$ and that $S[a..a+\delta-1]=S[a+\delta..a+2\delta-1]$. 
		Thus, we have that $[a..a+2\delta]\cap P_{\mu-1}=\{a,a+\delta,a+2\delta\}$.
		In addition, since $a\in P_\mu$, it must be that the preceding sub-block of $[a..a+\delta-1]$  is of different length or $S[a..a+\delta-1]\ne S[a-\delta..a-1]$.
		Recall that $\delta<(3/2)^\mu$ and therefore $2\delta<2\cdot(3/2)^\mu$, hence it must be that $[b..b+2\delta]\cap P_{\mu-1}=\{b,b+\delta,b+2\delta\}$.
		Therefore, in level $\mu-1$ we have two equal sub-blocks $[b..b+\delta-1]$ and $[b+\delta..b+2\delta-1]$ and the preceding sub-block must be different. Therefore, it must be that $b$ is the leftmost endpoint of the corresponding block in level $\mu$, which means  that $b\in P_\mu$.

		\paragraph{Type~\ref{item:type_cointossing}.}

		Let us assume that the block of $a$ in level $\mu-1$ is not part of the last $c\log^*n$ blocks in the sequence, otherwise, the proof is similar to the case of type~\ref{item:type_edge}.
		Let $p_0,p_1\dots,p_h$ be the positions of $P_{\mu-1}$ sorted in ascending order and let $i$ and $j$ be the indices such that $a=p_i$ and $b=p_j$.

		One of the two following cases must hold.
		Either $p_i$ is the leftmost position in the sequence or its label is a local minimum.
		For the first case, it depends on at most $2(3/2)^\mu$ positions to its left, which are part of block created from type~\ref{item:type_large} or type~\ref{item:type_contiguous}. In both cases, since $S_{p_i,\mu}=S_{p_j,\mu}$ it must be that this large block or periodic block appears also before $p_j$ in level $\mu-1$, and therefore $b=p_j\in P_\mu$ as required.
		
		Consider the second case, where $p_i\in P_\mu$  since its label is a local minimum.
		Intuitively, we claim that the partition into blocks around $p_i$ and $p_j$ is the same by the induction hypothesis. Thus, the labels given for each block in this neighborhood are also equal, and therefore the block of $p_j$ must be also the local minimum.
		Formally, we have that $\lab(i)$ is smaller than $\lab(i+1)$ and also smaller than $\lab(i-1)$, where $\lab(x)=\lab_{\mu}(x,c\log^*n)$ is the label of the block $[p_x..p_{x+1}-1]$ after the last inner-iteration on the transformation from $\mu-1$ to $\mu$.
		Thus, we focus only on the values of $\lab(p_{i-1}),\lab(p_i)$ and $\lab(p_{i+1})$. Each of these labels is dependent on the $c\log^*n$ sub-blocks to its right, and the length of each such a sub-block is at most $(3/2)^\mu$. Therefore, the first relevant sub-block begins at a position which is at least $p_i-(3/2)^\mu$ and the last relevant block ended at a position which is at most at $p_i+(c\log^*n+2)(3/2)^\mu$. 
		Since this range is fully contained in $[p_i-2(3/2)^\mu c\log^*n..p_i+2(3/2)^\mu c\log^*n]$ we are guaranteed by the induction hypothesis that all the positions in the same distances from $p_j$ also induced partitioning of $S$ in the same way. Therefore, the labels of $\lab(p_{j-1}), \lab(p_j)$ and $\lab(p_{j+1})$ are exactly the same as $\lab(p_{i-1}),\lab(p_i)$ and $\lab(p_{i+1})$.
		Thus, $\lab(p_j)$ must be also a local minimum, and so $b=p_j\in P_\mu$ as required.

		\paragraph{Type~\ref{item:type_edge}.}	
		As in the proof of the previous type, let $p_0,p_1\dots,p_h$ be the positions of $P_{\mu-1}$ sorted in ascending order and let $i$ and $j$ be the indices such that $a=p_i$ and $b=p_j$. 
		It is guaranteed that the closest type~\ref{item:type_contiguous} or~\ref{item:type_large} block in $P_{\mu}$ after $p_i$ starts at most $c\log^*n$ sub-blocks after $p_i$.
		Since the length of each sub-block in this area is at most $(3/2)^\mu$, it must be that this position is at most $p_i+(3/2)^\mu c\log^*n<p_i+2(3/2)^\mu c\log^*n$.
		Therefore, by the induction hypothesis, the same blocks appear exactly the same also after $p_j$.
		The only dependency on blocks to the left $p_i$ is just if there are at most $2$ sub-blocks in the sequence of sub-blocks to the left of $p_i$ which are not of type~\ref{item:type_large} or type~\ref{item:type_contiguous}. However, due to similar arguments to the above, the number of such sub-blocks also must be the same to the left of $p_i$ and to the left of $p_j$. 
		Therefore, $p_j=b$ must be also a position in $P_\mu$.

		\paragraph{Conclusion.}
		We proved that $a\in P_\mu\Rightarrow b\in P_\mu$, the proof of the  opposite direction is exactly the same. Hence, we have that $a\in P_\mu\Leftrightarrow b\in P_\mu$.
	\end{proof}

	Property~\ref{item:blocks_size} of Definition~\ref{def:landmarks_set} states that any long block in the partitioning is a periodic substring of $S$, which its period length is at most $\tau$. 
	In the following lemma we prove that the positions of the decomposition in each level of the hierarchy have this property. 
	In particular, the set $P_{\log_{3/2}\tau}$ has Property~\ref{item:blocks_size} of Definition~\ref{def:landmarks_set}.
	\begin{lemma}\label{lem:near_landmark}

		For any level $0\le\mu\le\log_{3/2}\tau$ of the hierarchic decomposition, if $p_i<p_{i+1}$ are two consecutive positions in $P_\mu$ such that $p_{i+1}-p_i>12\cdot (3/2)^\mu$ then
		$S[p_i..p_{i+1}-1]$ has period length of at most $(3/2)^\mu$
		
	\end{lemma}
	\begin{proof}
		We prove the lemma by induction on the level of the hierarchy.
		For level $0$, the lemma is trivial, since $P_0=[n]$.
		We assume by induction that the lemma holds for $P_{\mu-1}$, and we prove that it holds for $P_\mu$ as well.
		Let $[p_i..p_{i+1}-1]$ be a block in level $\mu$, i.e., $[p_i..p_{i+1}]\cap P_\mu=\{p_i,p_{i+1}\}$, such that $|[p_i..p_{i+1}-1]|=p_{i+1}-p_i>12\cdot(3/2)^\mu$. 
		If $[p_i..p_{i+1}-1]$ is a block in level $\mu-1$ as well, the lemma holds due to the induction hypothesis.
		Otherwise, it must be that $[p_i..p_{i+1}-1]$ is composed of several sub-blocks of level $\mu-1$ and each of them is of length at most $(3/2)^\mu$ (since otherwise they have corresponded to a block of type~\ref{item:type_large}). Hence, it must be that $[p_i..p_{i+1}-1]$ is composed of more than $12$ blocks of level $\mu-1$.
		The only type of blocks in level $\mu$ which can be composed of more than $12$ sub-blocks of level $\mu-1$ is type~\ref{item:type_contiguous}, and in this case all the substrings corresponding to the sub-blocks are equal. Therefore, the substring corresponding to $[p_i..p_{i+1}-1]$ is $uuuu\cdots u$ where $u$ is a string of length at most $(3/2)^{\mu-1}$.  Thus, the period length of $S[p_i..p_{i+1}-1]|$ is at most $(3/2)^\mu$.
	\end{proof}

	Hence, by combining Lemmas~\ref{lem:consecutive_blocks_length},~\ref{lem:deterministic_local_property} and~\ref{lem:near_landmark} we receive the following corollary.
	\begin{corollary}\label{cor:P_logtau_fine}
		The set  $P_{{\log_{3/2}\tau}}$  is  $(\O(\tau),\O(\tau \log^*n))$-partitioning set of size $\O(\frac n \tau)$.
	\end{corollary}
	
	Or, by adjusting a constant, we get a truly $(\tau,\O(\tau\log^*n))$-partitioning set of size $\O(\frac n\tau)$.

	\begin{corollary}\label{cor:P_logtau_tweleve_fine}
		The set  $P_{{\log_{3/2}(\tau/12)}}$  is a $(\tau,\O(\tau \log^*n))$-partitioning set of size $\O(\frac n \tau)$.
	\end{corollary}

	To complete the missing details for the deterministic construction we prove in the following lemma that  $P_{{\log_{3/2}(\tau/12)}}$ is forward synchronized (see Definition~\ref{def:forward_sync}). 
	This lemma is required for the deterministic $\SST$ construction  introduced in Section~\ref{sec:SST_construction}.

	\begin{lemma}\label{lem:det_is_forward_sync}
		The set $P_{{\log_{3/2}(\tau/12)}}$ is a forward synchronized partitioning set
	\end{lemma}
	\begin{proof}

		Let $q_{1}<\cdots <q_{h}$ be all the positions in $P_\mu$ and $p_{1}<\cdots <p_{h}$ be all the positions in $P_{\mu+1}$
		We prove the lemma by induction on the level of the hierarchy.
		For level $0$, the lemma is trivial, since $P_0=[n]$.
		We assume by induction that the lemma holds for $P_{\mu}$, and we prove that it holds for $P_{\mu+1}$ as well.
		Let $p_{i},p_{j}\in P_{\mu+1}$ such that $\LCE(p_{i},p_{j}) > p_{i+1}-p_{i} + 6(3/2)^{\mu+1} c\log^*n$.
		Since $P_{\mu+1} \subseteq P_{\mu}$, $p_{i},p_{j}\in P_{\mu}$.
		Let $q_{i'}=p_{i}$ and $q_{j'}=p_{j}$ be the corresponding positions in the ordered set of positions $P_{\mu}$. 
		Remark that $q_{i'+1}\leq p_{i+1}$ and $q_{j'+1}\leq p_{j+1}$, thus 
		$$\LCE(q_{i'},q_{j'})=\LCE(p_i,p_j) > p_{i+1}-p_{i} + 6(3/2)^{\mu+1} c\log^*n > q_{i'+1}-q_{i'} + 6(3/2)^{\mu} c\log^*n.$$
		Therefore, by the induction hypothesis $q_{i'+1}-q_{i'} = q_{j'+1}-q_{j'}$.
		Furthermore, for every $k\ge 0$ such that $q_{i'+k+1} \le p_{i+1}+2(3/2)^{\mu}c\log^*n$ it follows that
		\begin{align*}
		\LCE(q_{i'+k},q_{j'+k}) &= \LCE(p_{i},p_{j})-(q_{i'+k}-p_{i})\\
		&> (p_{i+1}-p_{i})-(q_{i'+k}-p_{i}) +  6(3/2)^{\mu+1} c\log^*n \\
		&\ge  q_{i'+k+1} - 2(3/2)^{\mu}c\log^*n - q_{i'+k} + 9(3/2)^{\mu} c\log^*n  \\
		&> q_{i'+k+1} - q_{i'+k} + 6(3/2)^{\mu} c\log^*n
		\end{align*}
		Thus by induction on k, we have that $q_{i'+k+1} - q_{i'+k} = q_{j'+k+1} - q_{j'+k}$ for all $k\ge 0$ such that $q_{i'+k+1} \le p_{i+1}+2(3/2)^{\mu}c\log^*n$.

		We continue the proof by partitioning the cases according to the type of the block $[p_{i}..p_{i+1}]$.
		
		\paragraph{Type~\ref{item:type_large}.}
		In this case that $[p_{i}..p_{i+1}-1]$ is longer than $(3/2)^{\mu+1}$, and it is also a sub-block at level $\mu$. 
		Thus, $q_{i'+1} = p_{i+1}$.
		From the above remark it follows immediately that since $q_{i'+1} - p_{i} = q_{j'+1} - p_{j}$ then $q_{j'+1} = p_{j+1}$ because it is a type~\ref{item:type_large} block as well. 
		Thus, the lemma holds.
		
		\paragraph{Type~\ref{item:type_contiguous}.} 
		In this case that for every $k\ge 0$ such that $q_{i'+k} < p_{i+1}$ all the blocks $[q_{i'+k}..q_{i'+k+1}-1]$ are equal, and the block at level $\mu$ that starts at $p_{i+1}$ is different.
		Let $\delta= q_{i'+1} - q_{i'}$ be the equal sub-blocks' length.
		Notice that $\delta < (3/2)^{\mu+1}$ since it is not a type~\ref{item:type_large} sub-block.
		From the above remark it follows that the same holds for all the blocks $[q_{j'+k}..q_{j'+k+1}-1]$ for every $k\ge 0$ such that $q_{j'+k} < p_{j} + p_{i+1}-p_{i}$, and that the position defined as $q_{\ell}=p_{j} + p_{i+1}-p_{i}$ is a position in $P_\mu$.
		
		The only case where $p_{j+1}\neq q_{\ell}$ is if $[q_{\ell}..q_{\ell+1}-1]$ is equal to the previous sub-blocks. 
		In that case since $\delta < (3/2)^{\mu+1}$, if $q_{\ell+1}=q_{\ell}+\delta$ then $[p_{i+1}..p_{i+1}+\delta-1]$ is also a sub-block at level $\mu$, and it is equal to the previous sub-blocks.
		This contradicts the fact that at $p_{i+1}$ starts a level $\mu$ block that is different from its previous block.
		Therefore, at $q_{\ell}$ starts a block that must be different from its previous equal sub-blocks.
		From all the above it follows that $p_{j+1}=q_{\ell}=p_{j} + p_{i+1}-p_{i}$.
		
		\paragraph{Type~\ref{item:type_cointossing}.}	
		As with the proof of Lemma~\ref{lem:deterministic_local_property}, we assume that the sub-block $q_{i'}$ is not of the rightmost $c\log^*n$ sub-blocks in the sequence of type~\ref{item:type_cointossing} sub-blocks.
		By definition, all sub-blocks in the sequence are shorter than $(3/2)^{\mu+1}$. 
		Also, there are at least $c\log^*n$ sub-blocks in the sequence after the sub-block that starts at $q_{i'}$.
		Therefore, for all $0\le k\le c\log^*n$ it follows that $q_{i'+k+1} \le p_{i+1}+2(3/2)^{\mu}c\log^*n$.
		Hence, from the above remark it follows that the blocks partition at level $\mu$ is forward synchronized for at least $c\log^*n+1$ blocks from $q_{i'}$ and $q_{j'}$.
		
		Notice that $p_{i+1}=q_{i'+\tilde{k}}$ for some $1\le \tilde{k}\le 11 < c\log^*n$ since the position of $p_{i+1}$ was selected for level $\mu+1$ either due to $q_{i'+\tilde{k}}$ having a local minimum label value, or if $q_{i'+\tilde{k}}$ is the right end-point of one of the last $c\log^*n$ sub-blocks in the sequence.
		For the latter case, the proof is similar to that of type~\ref{item:type_edge}.
		
		If $q_{i'+\tilde{k}}$ was chosen for level $\mu+1$ due to its label value, then there are at least $c\log^*n$ short sub-blocks in the sequence after $q_{i'+\tilde{k}}$.
		Thus, the range of $k$ such that $q_{i'+k+1} \le p_{i+1}+2(3/2)^{\mu}c\log^*n$ can be extended to reach up to $\tilde{k}+c\log^*n$, and with it the range of forward synchronization of blocks from $q_{i'}$ and $q_{j'}$.
		Since the $\lab$ of every position is dependent on the next $c\log^*n$ blocks to the right of the block, then for every $0\le k\le \tilde{k}+1$ it follows that $\lab(q_{i'+k})=\lab(p_{j'+k})$.
		Since $q_{i'+\tilde{k}}=p_{i+1}$ is the first block with local minimum label value after $p_{i}$, then the same holds for $q_{j'+\tilde{k}}$ and $p_{j}$, thus $p_{j+1} = p_{j} + (p_{i+1} - p_{i})$.
		
		\paragraph{Type~\ref{item:type_edge}.}
		Let $q_{i'+\tilde{k}}$ for $1\le\tilde{k}\le c\log^*n$ be the first position to the right of the type~\ref{item:type_edge} sub-blocks sequence.
		Since all blocks in the sequence are shorter than $(3/2)^{\mu+1}$, then as shown in the proof of the previous case (type~\ref{item:type_cointossing}) it follows that the blocks at level $\mu$ are forward synchronized from $q_{i'}$ and $q_{j'}$ up to $q_{i'+\tilde{k}}$ and $q_{j'+\tilde{k}}$ respectively.
		Remark that if $q_{j'+\tilde{k}+1}$ is the start of sub-block shorter than $(3/2)^{\mu+1}$, then an equal block must start at $q_{i'+\tilde{k}}$.
		Since $q_{i'+\tilde{k}}$ is the first position after a type~\ref{item:type_edge} sequence, if it is the beginning of a short sub-block, then it must be the beginning of a type~\ref{item:type_contiguous} sub-blocks sequence.
		In that case $q_{i'+\tilde{k}+2}< p_{i+1}+2(3/2)^{\mu}c\log^*n$ and thus from the forward synchronization assumption it follows that $q_{j'+\tilde{k}}$ is also the beginning of a type~\ref{item:type_contiguous} sub-blocks sequence.
		Therefore in any case $q_{j'+\tilde{k}}$ is also the end of a type~\ref{item:type_edge} sequence.
		
		Since in the case of a type~\ref{item:type_edge} sub-block, picking a position for the next level depends only on the distance from the end of the sequence and the sequence's length, and since the positions of $p_{i}$ and $p_{j}$ are synchronized up to the end of the sequence, it follows that $p_{j+1}=p_{j} + p_{i+1} - p_{i}$. 
		
	\end{proof}

	\subsection{Efficient Computation Using Only \texorpdfstring{$\O(\frac n\tau+\log \tau)$}{small amount of} Working Space}\label{app:deterministic_in_space}

	In this section we describe how to compute the set $P_{{\log_{3/2}\tau}}$ using only $\O(\frac n\tau+\log \tau)$ words of space.
	We focus on the perspective of the hierarchical decomposition of $S$ as a tree. Since each block in level $\mu$ of the decomposition is composed of complete sub-blocks of level $\mu-1$, the hierarchical decomposition induces a tree of blocks. Formally, for each block $[x..y-1]$ in level $\mu$, we consider all the blocks of level $\mu-1$ that compose $[x..y-1]$ as its children. 
	The important property for the algorithm is that the decision whether a sub-block from level $\mu-1$ is a leftmost endpoint of a block in level $\mu$ can be taken using information of an only constant number of sub-blocks to the left and to the right of the block. 
	We will denote by $\lab_\mu([x..y-1],j)$ the label of the sub-block $[x..y-1]$ after the $j$th inner-iteration in the outer-iteration that computes the decomposition of level $\mu$, based on the decomposition of level $\mu-1$.

	\paragraph{Labels list.}
	We slightly generalize the definition of labels that we gave in Section~\ref{sec:determinsiticSelection} in order to clarify how local is the transformation from $P_{\mu-1}$ to $P_{\mu}$.
	For each sub-block $B=[x..y-1]$ in level $\mu-1$, we define an ordered list of $\O(\log^*n)$ labels $\mathcal L_B$ as follows. If the length of $B$ is at least $(3/2)^{\mu}$ then $\mathcal L_B=\emptyset$.
	Otherwise, let $B'=[y..z-1]$ be the sub-block following $B$ in level $\mu-1$. 
	If $S[x..y-1]=S[y..z-1]$, we define $\mathcal L_B=\emptyset$ as well.
	On the other hand, if $S[x..y-1]\ne S[y..z-1]$, then the size of $\mathcal L_B$ is $\min\{|\mathcal L_{B'}|+1,c\log^*n\}$. The $j$th element in $\mathcal L_B$ is $\lab_\mu(B,j)$. 
	Notice that for $j>1$, given the $(j-1)$th elements of $\mathcal L_B$ and $\mathcal L_{B'}$ the computation of $\lab_\mu(B,j)$ can be done in constant time without any additional information. 
	The computation of $\lab_\mu(B,1)$ is done in time which is linear in the length of $B$ and $B'$.
	The algorithm maintains with each block $B$ all the elements of $\mathcal L_B$ explicitly, except for the label $\lab_{\mu}(B,0)$, which is stored implicitly by the range of $B$ (recall that $\lab_{\mu}(B,0)$ is the character of $\Pi=\{0,1,\dots,2^{\ceil{\log_2 (|\Sigma|+1)}n}\}$ corresponded to $S[x..y-1]$).
	Since the sizes of the labels in $\mathcal L_B$ decreases logarithmically, the total size in bits of the whole list is dominated by the first element $\lab_{\mu-1}(B,1)$. Therefore, the space usage of the whole list is $O(\log n)$ bits, which is $O(1)$ words of space.  
	Thus, each block in level $\mu-1$ is stored by the algorithm uses $\O(1)$ words of space.

	\paragraph{Efficient computation of labels list.}
	Notice that the $j$th element $\mathcal L_B$ is a function of the of the $(j-1)$th elements of $\mathcal L_B$ and $\mathcal L_{B'}$. 
	Moreover, given $\mathcal L_{B'}$ from position $(j-1)$ to $|\mathcal L_{B'}|$ and the $(j-1)$th element of $\mathcal L_B$, all the elements of $\mathcal L_B$ form position $j$ to $|\mathcal L_B|$ follows deterministically.
	Hence, instead of running all the inner-iterations described above, the algorithm will use a \emph{lookup-table}. 
	First, the algorithm runs three ($3$) inner-iterations, so the size of $\lab_{\mu-1}(B,3)$ is $O(\log\log\log n)$ bits. In addition the total size of the elements in $\mathcal L_{B'}$ from the third elements till the end is also $O(\log\log\log n)$ bits. 
	Therefore, the number of pairs such that the first one is the third element in the list of one block and the second is all the elements in the list of the next block is $2^{\O(\log\log\log n)}=\O(\log^a\log n)$ for some constant $a$. For each such pair, the algorithm computes in advance all the remaining elements after the third element in $\mathcal L_B$, which takes all together $O(\log\log\log n)$ bits. 
	Hence, the total size of the lookup table is $O(\log^a\log n \cdot \log\log\log n)=o(\log n)$ bits, and the query time is constant. 
	Therefore, for each block the algorithm computes just three inner iterations and one query to the lookup table.

	\paragraph{Using the labels list.}
	Now, using the labels lists, we describe how to decide whether the block $B$ in level $\mu-1$ is a leftmost sub-block of some block of level $\mu$, using only the information of  one blocks to the left and four blocks to the right of $B$. 
	First of all, we identify for each sub-block  the type of sequence it belongs. 
	Identifying of sub-blocks of  type~\ref{item:type_large} is done by checking their length, regardless any information of other blocks. 
	Let $B$ be a sub-block from level $\mu-1$, and let $B'$ be the sub-block to the right of $B$. 
	$B$ is in a sequence of type~\ref{item:type_contiguous} if and only if $|B|<(3/2)^{\mu+1}$ and either $\mathcal{L}_B=\emptyset$ (in the case $B$ is not the rightmost sub-block in the sequence) or $\mathcal{L}_{B'}=\emptyset$ and $|B|=|B'|$ (in the case where $B$ is the rightmost sub-block in the sequence). 
	Each block which is not of type~\ref{item:type_large} or~\ref{item:type_contiguous} is obviously of type~\ref{item:type_cointossing} or~\ref{item:type_edge}. Since the treat of the rightmost $c\log^*n$ sub-blocks in a sequence of type~\ref{item:type_cointossing} is exactly as the treat for sub-blocks in a sequence of type~\ref{item:type_edge} the algorithm is not required to distinguish between them.

	Let $B$ be a block in level $\mu-1$. The algorithm sets $B$ as a leftmost sub-block in level $\mu$ due to the type of its sequence. 
	If $B$ is in a sequence of type~\ref{item:type_large} it always moved to level $\mu+1$. 
	If $B$ is in a sequence of type~\ref{item:type_contiguous}, the algorithm sets a block $B$ only if its preceding block is not in the sequence, which can be detected given the two sub-blocks to the left of $B$.
	If $B$ is in a sequence of type~\ref{item:type_cointossing}, and $|\mathcal L_B|=c\log^*n$ then we use the number of the block which is the last element in $\mathcal L_B$ and choose $B$  if this number is smaller than the numbers of the sub-blocks to its left and right. The treat of margin cases is obvious.
	If $B$ is  treated as a block in a sequence of type~\ref{item:type_edge} (which includes the rightmost $c\log^*n$ blocks in a sequence of type~\ref{item:type_cointossing}), we take all the blocks that are in an even distance from the right end of this sequence, which can be identified due to the parity of $|\mathcal L_B|$. The only exceptions are blocks near the ends of the sequence, which can be treated as described above, given the information of $4$ blocks ahead.

	\paragraph{The scheduling.}
	Further to the discussion above, the algorithm maintains any time at most $6$ blocks in any level of the hierarchical decomposition, where each block is maintained with its label set. 
	The process is done in recursion from the bottom level to the top-level and from right to left.
	Whenever the algorithm detects a new block in level $\mu-1$, the algorithm checks whether the block which is $4$ blocks to the right of the current block is a leftmost sub-block in level $\mu$.
	If the answer is positive, the algorithm creates a new block in level $\mu$, from the end of the last block that was created in level $\mu$ until the leftmost end of the new block.
	After the recursive call ends, the algorithm deletes old blocks, so at any time each level maintains explicitly only constant number of blocks.
	The only exceptional level is, of course, the top-level, where the algorithm maintains all the blocks and in particular the corresponding positions which compose $P_{{\log_{3/2}\tau}}$.

	\paragraph{Complexities.}
	The algorithm maintains $\O(\log \tau)$ levels, where for each level the algorithm maintains $\O(1)$ elements with $\O(1)$ words of space per element. Hence, all the levels, except the top-level use $\O(\log\tau)$ words of space. 
	In the top-level the algorithm maintains all the blocks, which due to Corollary~\ref{cor:P_logtau_fine}  is a set of size $\O(\frac n \tau)$. 
	Notice that the blocks in the top-level are maintained without any labels set.
	Thus, the total space usage of the algorithm is $\O(\frac n \tau + \log\tau)$.
	Regarding the running time, we analyze the time required for each level of the hierarchical decomposition and summing among all the levels.
	So, for level $\mu$ of the algorithm, the time usage of the algorithm is dominated by the inner iterations. 
	The first inner-iteration in level $\mu$ takes $\O(n)$ time. 
	Any other inner iteration takes time which is linear in the number of blocks in level $\mu$ which due to Lemma~\ref{lem:consecutive_blocks_length}  is $\O\left(\frac n {(3/2)^\mu}\right)$. Since the algorithm computes $O(1)$ (three) inner iterations for level $\mu$, and then uses the lookup table, the time usage of the algorithm for level $\mu$ is $\O(n+\frac n {(3/2)^\mu})$.
	Summing over all the levels, 
	$$\sum_{\mu=0}^{\log_{3/2}\tau}{\O(n+\frac n {(3/2)^\mu})}= \O(n\log\tau)+\O(n)\sum_{\mu=0}^{\log_{3/2}\tau}\left(\frac 2 3\right)^\mu = \O(n\log\tau+ n) = \O(n\log\tau).$$

	Thus, combining  with Corollary~\ref{cor:P_logtau_tweleve_fine} and Lemma~\ref{lem:det_is_forward_sync} we  
	receive the following corollary
	\begin{corollary}\label{cor:determinstic_set_complete}
		For any $1\le \tau\le \O(\frac n{\log n})$  there exists a deterministic algorithm that computes a forward synchronized $(\tau,\O(\tau\log^*n))$-partitioning set of size $\O(\frac n\tau)$ in $\O(n\log\tau)$ time using only $\O(\frac n\tau)$ words of working space.
	\end{corollary}

	\section{Missing Details for the Improved Deterministic \texorpdfstring{$\LCE$}{LCE} Data Structure}\label{app:better_deterministic}
	
	In this appendix we complete the details for the improved $\LCE$ construction  from Section~\ref{sec:better_deterministic}.
	The algorithm for the case $\tau<\sqrt{\log^*n}$ appears at the end of the section, thus we assume here $\tau\ge \sqrt{\log^* n}$.
	Recall, that we have defined the set $P$ as a $(\O(\tau'),\O(\tau'\log^*n))$-partitioning set ($\tau'=\frac \tau{\sqrt{\log^*n}}$) and $Q\subset P$, as a special subset of $P$ that is used by the algorithm.
	The data structure maintains the $\SST$ of the set $Q$. 
	Then, the query computation for the non-periodic case is similar to the first phase of the $\LCE$ data structure of Section~\ref{sec:LCE_construction}. The algorithm just compares characters until it finds a mismatch or corresponding selected positions. 
	However, in the case that there is no mismatch, the common prefix  can be considered as composed of two parts. In the first part the algorithm reads $\O(\delta)=\O(\tau'\log^*n)$ characters, to reach synchronization due to positions of $P$. Then, in the second part the algorithm reads another $\O(\tau'\log^*n)$ until it finds two corresponding positions that are actually in $Q$ and not only in $P$ (recall that the set $P$ is not maintained by the algorithm).
	
	Let us focus on the second part, in this part the algorithm reads characters from at most $\log^*n$ blocks, and since we have assumed that each block is of length at most $\tau'$, the total length was $\tau'\log^*n$. However, when we have long periodic blocks, this assumption does not hold anymore.
	Hence, we generalize our method. Instead of picking positions from $P$ based only on their rank, we add some ``weights" to the task.
	For each block $[p_i..p_{i+1}-1]$ induced by $P$, let  $\ell_i=p_{i+1}-p_i$. We associate with the block of $p_i$, $\ceil{\frac {\ell_{i-1}}{\tau'}}$ \emph{tokens} (we emphasize that the number of tokens of one block depends on the length of the \emph{preceding} block), where the tokens are numbered ($1,2,3,\dots$). Then, we select all the positions that have a token which its  index $i$ is an integer multiple of $\sqrt{\log^*n}$  or $i\modulo \log^*n<\sqrt{\log^*n}$. Notice that this method is a generalization of the method that was introduced in Section~\ref{sec:better_deterministic}, since there the length of each block is at most $\tau'$ and therefore each block has exactly one token.
	
	First of all, we argue that the size of the set $Q$ produced by the algorithm is still $\O(\frac {n}{\tau})$.
	
	\begin{lemma}
		$|Q|=\O(\frac n\tau)$.
	\end{lemma}
	\begin{proof}
		Recall that $|P|=\O(\frac n{\tau'})=\O(\frac {n}{\tau}\sqrt{\log^*n})$.  First we bound the number of tokens among all blocks induced by $P$.
		$$\sum_{i=1}^{|P|}\ceil{\frac{\ell_{i-1}}{\tau'}}\le \sum_{i=1}^{|P|}1+\frac{\ell_{i-1}}{\tau'}\le |P|+\frac 1{\tau'}\sum_{i=1}^{|P|}\ell_{i-1}\le |P|+\frac n{\tau'}=\O\left(\frac n{\tau'}\right)=\O\left(\frac n\tau\sqrt{\log^*n}\right)$$
		Hence, the number of tokens provided by the algorithm is $\O(\frac n\tau\sqrt{\log^*n})$.
		
		The algorithm adds to $P$ only positions which have token with a number which is an integer multiple of $\sqrt{\log^*n}$ or that its remainder modulo $\log^*n$  is less that $\sqrt{\log^*n}$. Hence, the size of $Q$ is at most $\O(\frac n\tau \sqrt{\log^*n}/\sqrt{\log^*n})=\O(\frac n\tau)$. Notice that, a case where one block has multiple tokens that should be selected just reduces the size of the set $Q$.
	\end{proof}
	
	As in Section~\ref{sec:better_deterministic}, the algorithm creates the $\SST$ of the set $Q$ in the construction phase. If all the blocks induced by $P$ are of length at most $\tau\sqrt{\log^*n}$ then with analysis which is similar to the proof of Lemma~\ref{lem:better_deterministic_non_periodic}, one can prove that the simple algorithm that compares characters until the algorithm finds a mismatch or corresponding positions will work and the query time will be $\O(\tau\sqrt{\log^*n})$. 
	However, as in the simple $\LCE$ data structure, we might have very long blocks, even due to $P$, and therefore we cannot hope to reach a pair of corresponding positions in $\O(\tau\sqrt{\log^*n})$ time. Therefore, for such cases we use similar techniques to the periodic case in the basic $\LCE$ data structure. 
	Let $c'>1$ be a constant such that $P$ is a $(\tau',c'\tau'\log^*n)$-partitioning set (of size $\O(\frac n{\tau'})$), 
	in the terms of Section~\ref{sec:LCE_construction} we have $\delta=c'\tau'\log^*n= c'\tau\sqrt{\log^*n}$.
	
	\begin{lemma}\label{lem:better_determinsitc_periodic}
		If $q_i<q_{i+1}$ are two consecutive positions in $Q$, such that $q_{i+1}-q_i>2\delta$ then it must be that $S[q_i+\delta..q_{i+1}-1]$ is a periodic string with period length of at most $\tau'=\frac \tau{\sqrt{\log^*n}}$.
	\end{lemma}
	\begin{proof}
		Assume by contradiction that $S[q_i+\delta..q_{i+1}-1]$ is not a periodic substring. Hence, it must be that in the partitioning of $S$ induced by $P$, there exists at least one position in the range $q_i+\delta<p<q_{i+1}$. 
		We consider the number of tokens given for blocks beginning from the block following $q_i$ until the block of $p$. 
		Let us denote the set of blocks in this range as $\hat P$, then the total number of tokens given for these blocks is at least
		$$\sum_{i=1}^{|\hat P|}\ceil{\frac{\ell_{i-1}}{\tau'}}\ge \sum_{i=1}^{|\hat P|}\frac{\ell_{i-1}}{\tau'}\ge \frac 1{\tau'}\sum_{i=1}^{|\hat P|}\ell_{i-1} = \frac \delta{\tau'}=c'\log^*n>\log^*n$$
		
		Since in every sequence of $\sqrt{\log^*n}$ tokens the algorithm inserts at least one position from $P$ to $Q$, it must be that there exists some position $q\in Q$ such that $q_i<q\le p<q_{i+1}$, which contradicts the fact that $q_i$ and $q_{i+1}$ are consecutive positions in $Q$. Hence, it must be that $S[q_i+\delta..q_{i+1}-1]$ is a periodic block with period length of at most $\tau'=\frac \tau{\sqrt{\log^*n}}$.
	\end{proof}

	Now, using Lemma~\ref{lem:better_determinsitc_periodic} the algorithm works as follows. First, the algorithm compares $4\delta$ characters. If the algorithm finds a mismatch or a pair of selected corresponding positions in $P$, the algorithm continues as in the non-periodic case.
	Otherwise, due to  Lemma~\ref{lem:center_blocks} it must be that 
	$$\{p-i\,|\,p\in(P\cap[i+\delta..i+\ell-\delta-1])\}=\{p-j\,|\,p\in(P\cap[j+\delta..j+\ell-\delta-1])\}.$$
	In particular, since $\ell>4\delta$ we have that:	$$\Delta=\{p-i\,|\,p\in(P\cap[i+\delta..i+3\delta-1])\}=\{p-j\,|\,p\in(P\cap[j+\delta..j+3\delta-1])\}.$$ 
	Hence, due to Lemma~\ref{lem:better_determinsitc_periodic} it must be that all the characters in $S[i+2\delta..i+3\delta-1]$ and $S[j+2\delta..j+3\delta-1]$ are in long periodic blocks due to partitioning induced by $P$.  
	Let $\alpha_i=\suc_Q(i+2\delta)-i$ and $\alpha_j=\suc_Q(j+2\delta)-j$. If $\alpha_i=\alpha_j$ the algorithm uses the $\SST$ on $Q$ to compute the $\LCE(\suc_Q(i+2\delta),\suc_Q(j+2\delta))$ to compute the $\LCE(i,j)$ query. 
	Otherwise, due to Lemma~\ref{lem:better_determinsitc_periodic} we have that $S[i+2\delta..i+\alpha_i]$ and $S[j+2\delta..j+\alpha_j]$ are periodic substrings with the same period and therefore if $\alpha=\min\{\alpha_i,\alpha_j\}$ then $S[i..i+\alpha]=S[j..j+\alpha]$. Moreover, as in Lemma~\ref{lem:end_of_periodic_block_close_to_LCE}, it must be the case that $\LCE(i,j)\le\alpha+\delta$ and therefore, by reading another at most $\delta$ characters from  $S[i+\alpha]$ and $S[j+\alpha]$ the algorithm is guaranteed to find the first mismatch between the suffixes.

	\subsection{Algorithm for \texorpdfstring{ $\tau<\sqrt {\log^*n}$}{very small tau}}
	If $\tau<\sqrt {\log^*n}$ then we use even simpler algorithm. We define $P=\{1,2,3,\dots,n\}$, and we use Lemma~\ref{lem:diff-cover} with $t=\tau^2$. Let $Q\subseteq P$ be a set of size $|Q|=\O(\frac n \tau)$ such that for every $i,j\in [n-\tau^2]$ there exists $k\le \tau^2$  such that both $i+k\in Q$ and $j+k\in Q$. Thus, the algorithm computes $Q$ and builds the $\SST$ on $Q$. In addition the algorithm computes the LCA information for the $\SST$.
	The query computation is straightforward. The algorithm compares at most $\tau^2$ pairs of corresponding characters. If a mismatch is found, the algorithm reports the $\LCE$. Otherwise, the algorithm finds the smallest $k\ge 0$ such that $i+k,j+k\in Q$ and returns the sum of $k$ and $\LCE(i+k,j+k)$.
	The space usage of this data structure is of course $\O(\frac n\tau)$ words of space, and the query time is $\O(\tau^2)$. Since $\tau<\sqrt{\log^*n}$ we have $\O(\tau^2)\subseteq \O(\tau\sqrt{\log^* n})$, as required.

	\section{Useful Tools and Lemmas}~\label{app:useful_tools}
	
	In this appendix we introduce several theorems and lemmas which we use as basic tools in our data structures and algorithms.
	
	\subsection*{String Periodicity}
	The following theorem is one of the most basic theorems regarding  periodicity of strings.
	\begin{theorem}[Fine and Wilf~\cite{FW65}]\label{thm:FineAndWilf}
		Let $u$ be a string. If $u$ has periods of length $p$ and $q$ and $|u|\ge p+q-\gcd(p,q)$, then $u$ also has a period  of length $\gcd(p,q)$.
	\end{theorem}

	The following lemmas are used in several parts of the paper.
	
	\begin{lemma}\label{lem:periodSubstring}
		Let $u$ be a string with period length $p$. If $v$ is a substring of $u$ of length at least $p$, and $v$ has a period length of $q|p$, then $u$ has a period length of $q$.
	\end{lemma}

	\begin{proof}%[Proof of Lemma~\ref{lem:periodSubstring}]
		If $p=q$ the lemma is trivial.
		Otherwise, $p\ne q$ and $q|p$, hence $q<p$.

		Let $a$ be an index such that $v[1..p]=u[1+a..p+a]$.
		Notice that for each $i\in \mathbb{N}$ there exists $1+a\le j\le p+a$ such that $i\modulo p= j\modulo p$.

		Given any index $1\le i\le |u|-q$ let  $1+a\le  j+a \le p+a$ be the index such that $i\modulo p= (j+a)\modulo p$.
		Since $u$ has period length of $p$, $u[i]=u[j+a]=v[j]$.
		We treat two cases:
		\begin{itemize}
			\item \textbf{Case 1 $j+a+q\le a+p$:} then we have similarly that $u[i+q]=u[j+a+q]=v[j+q]$.
			In this case, since $v$ has period length of $q$ we have $ [j]=v[j+q]$ and therefore $u[i]=v[j]=v[j+q]=u[i+q]$.
			
			\item \textbf{Case 2  $j+a+q>a+p$:} hence $a\le j+a+q-p<j+a\le a+p$.
			We have that $i\modulo p=j+a \modulo p$ and therefore, $i+q\modulo p = j+a+q \modulo p= j+a+q-p \modulo p$.
			Thus, since $u$ has period length of $p$, $u[i+q]=u[j+a+q-p]=v[j+q-p]$.
			Since $q|p$, we have also $q|(q-p)$, thus, since $v$ has period length of $q$ we have $v[j+q-p]=v[j]$.
			To conclude, $u[i]=v[j]=v[j+q-p]=u[j+a+q-p]=u[i+q]$.
		\end{itemize}
		
		Hence, $u[i]=u[i+q]$ for any $1\le i\le |u|-q$, and therefore $q$ is a period length of $u$.
	\end{proof}

	\begin{lemma}\label{lem:twiceperiod_equal}
		Let $u$ and $v$ be two strings of length $|u|=|v|\ge 2t$ for some $t\in\mathbb N$, such that both $\rho_u\le t$ and $\rho_v\le t$. If $u[1..2t]=v[1..2t]$ then $u=v$.
	\end{lemma}
	\begin{proof}
		We first prove that $\rho_u=\rho_v$, and then we use this fact to prove the lemma.
		
		Assume by contradiction that $\rho_u\ne\rho_v$, and without loss of generality, we assume that  $\rho_u<\rho_v$.
		Let $w=u[1..2t]=v[1..2t]$. Due to Theorem~\ref{thm:FineAndWilf} (Fine and Wilf theorem), since $\rho_u$ and $\rho_v$ are both period lengths of $w$, we have that $\rho_w|\gcd(\rho_u,\rho_v)$ and therefore $\rho_w\le \rho_u<\rho_v$.
		Due to Lemma~\ref{lem:periodSubstring} since $w$ is a substring of $v$ length at least $2t\ge\rho_v$ and $\rho_w|\rho_v$, we have that $\rho_w$ is  a period length of $v$, and therefore $\rho_v\le\rho_w\le\rho_u$, which contradicts the assumption. 
		Hence  $\rho_v=\rho_u$, and we denote this length as $\rho$.

		Let $i\in[|u|]$ be an index in the strings, we have to prove that $u[i]=v[i]$. If $i\le 2\rho$, the claim holds from the fact that $u[1..2t]=v[1..2t]$. 
		Otherwise, since $\rho$ is a period length of $u$, we have that $u[i]=u[i+k\cdot\rho]$ for any integer $k\in\mathbb Z$. In particular $u[i]=u[\rho+(i \modulo \rho)]$.
		Similarly, we have that $v[i]=v[\rho+(i \modulo \rho)]$. Since $1\le\rho+(i \modulo \rho)\le2\rho-1\le2t$ we have that $u[\rho+(i \modulo \rho)]=v[\rho+(i \modulo \rho)]$, and therefore $u[i]=v[i]$ as required.
	\end{proof}

	\subsection*{Sorting Strings}

	The following lemmas are  direct results of two facts stated in~\cite{GK17}.
	
	\begin{lemma}[{\cite[Fact 3.1 and Fact 3.2]{GK17}}]\label{lem:string_sorting_log_star}
		For any $1\le \tau\le n$, given random access to $\O(n/\tau)$ strings of length $\O(\tau\log^*n)$, the strings can be sorted in $\O(n\log^* n)$ time and $\O(n/\tau)$ words of space.
	\end{lemma}

	\begin{lemma}[{\cite[Fact 3.1 and Fact 3.2]{GK17}}]\label{lem:string_sorting}
		For any $1\le \tau\le n$, given random access to $\O(n/\tau)$ strings of length $\O(\tau)$, the strings can be sorted in $\O(n)$ time and $\O(n/\tau)$ words of space.
	\end{lemma}

	\subsection*{Lowest Common Ancestor} 
	For any tree $T$, we denote by $LCA(u,v)$ the \emph{lowest common ancestor} of nodes $u,v\in T$. The $LCA$ of any two leaves in a trie is the longest common prefix (LCP) of the corresponding strings. Hence, an efficient data structure for $LCA$ queries implies an efficient data structure for LCP queries.
	\begin{lemma}[{\cite{HT84,BF00}}]\label{lem:LCA}
		Any tree $T$ can be preprocessed in $\O(|T|)$ time and space to support LCA queries in constant time.
	\end{lemma}

	\bibliography{RefsLCE}
	\newpage

\end{document}